\documentclass{article}
\usepackage[svgnames]{xcolor}
\usepackage[margin=1in]{geometry}
\usepackage{fullpage}
\usepackage{orcidlink}
\usepackage[utf8]{inputenc}
\usepackage{amsthm}
\usepackage{thmtools}
\usepackage{amsmath}
\usepackage{amssymb}
\usepackage{algpseudocode}
\usepackage{algorithm}

\usepackage{hyperref}
\hypersetup{
    colorlinks,
    linkcolor={red!50!black},
    citecolor={blue!50!black},
    urlcolor={blue!30!black}
}

\usepackage[nameinlink]{cleveref}
\crefname{figure}{Figure}{Figures}
\crefname{example}{Example}{Examples}
\crefname{theorem}{Theorem}{Theorems}
\crefname{lemma}{Lemma}{Lemmas}
\crefname{proposition}{Proposition}{Propositions}
\crefname{claim}{Claim}{Claims}
\crefname{conjecture}{Conjecture}{Conjectures}
\crefname{section}{Section}{Sections}
\crefname{definition}{Definition}{Definitions}
\crefname{corollary}{Corollary}{Corollary}
\crefname{remark}{Remark}{Remark}

\newtheorem{theorem}{Theorem}[section]

\newtheorem{lemma}[theorem]{Lemma}
\newtheorem{proposition}[theorem]{Proposition}
\newtheorem{remark}[theorem]{Remark}
\newtheorem{claim}[theorem]{Claim}
\newtheorem{corollary}[theorem]{Corollary}
\newtheorem{example}[theorem]{Example}
\theoremstyle{definition}
\newtheorem{definition}[theorem]{Definition}

\usepackage[
    type={CC},
    modifier={by},
    version={4.0},
    imagewidth=6em
]{doclicense}

\usepackage[framemethod=tikz]{mdframed} 
\usepackage{tablefootnote} 
\makeatletter 
\AfterEndEnvironment{mdframed}{%
 \tfn@tablefootnoteprintout%
 \gdef\tfn@fnt{0}%
}
\usepackage[margin=1cm,font=small]{caption}

\usepackage{xspace}
\usepackage{amsmath}
\usepackage{amsfonts}
\usepackage{ifthen}

\usepackage{tikz}
\usetikzlibrary{positioning,shapes.misc,arrows.meta,decorations,decorations.pathmorphing,through,intersections,decorations.markings}

\makeatletter
\newcommand\notsotiny{\@setfontsize\notsotiny\@vipt\@viipt}
\makeatother

\newcommand{\NP}{\textsf{NP}\xspace}

\newcommand{\Z}[1]{\ensuremath{\mathbb{Z}/#1\mathbb{Z}}}
\newcommand{\vname}[1]{\ensuremath{v_{#1}}}

\newcommand{\dsn}{\textsf{DSN}\xspace}
\newcommand{\kdsn}[1]{\mbox{{#1}-\textsf{DSN}}}
\newcommand{\scss}{\textsf{SCSS}\xspace}
\newcommand{\kscss}[1]{\mbox{{#1}-\textsf{SCSS}}}

\newcommand{\ewm}{\textsf{EWM}\xspace}%
\newcommand{\dsnm}{\textsf{DSNM}\xspace}
\newcommand{\kdsnm}[1]{\mbox{{#1}-\textsf{DSNM}}}

\newcommand{\kscssm}[1]{\mbox{{#1}-\textsf{SCSSM}}}

\usepackage{xcolor}
\definecolor{csegm1}{RGB}{230,159,0}
\definecolor{csegm2}{RGB}{0,114,178}
\definecolor{csegm3}{RGB}{204,121,167}
\definecolor{csegm4}{RGB}{0,158,115}
\definecolor{csegm5}{RGB}{0,0,0}
\definecolor{csegm6}{RGB}{213,94,0}
\definecolor{csegm7}{RGB}{86,180,233}
\newcommand{\myparagraph}[1]{\paragraph{#1}}

\tikzset{cross/.style={cross out, draw, 
         minimum size=2mm, 
         inner sep=0pt, outer sep=0pt}}

\title{Edge-Minimum Walk of Modular Length in Polynomial Time}

\author{Antoine Amarilli\thanks{Univ.\ Lille, Inria, CNRS, Centrale Lille, UMR 9189 CRIStAL, F-59000 Lille, France \& LTCI, Télécom Paris, Institut polytechnique de Paris. Partially
 supported by the ANR project EQUUS ANR-19-CE48-0019, by the
 Deutsche Forschungsgemeinschaft (DFG, German Research Foundation) – 431183758, and by the ANR project ANR-18-CE23-0003-02 (“CQFD”). \texttt{antoine.a.amarilli@inria.fr}} \and Benoît Groz\thanks{Paris-Saclay University, CNRS, LISN. \texttt{benoit.groz@lisn.upsaclay.fr}} \and Nicole Wein\thanks{University of Michigan. This work was initiated while the author was affiliated with the Simons Institute. \texttt{nswein@umich.edu}}}

\date{}

\begin{document}

\maketitle

\begin{abstract}
We study the problem of finding, in a directed graph, an $st$-walk of length $r \bmod q$ which is edge-minimum, i.e., uses the smallest number of \emph{distinct} edges. Despite the vast literature on paths and cycles with modularity constraints, to the best of our knowledge we are the first to study this problem. Our main result is a polynomial-time algorithm that solves this task when $r$ and $q$ are constants. 

We also show how our proof technique gives an algorithm to solve a generalization of the well-known Directed Steiner Network problem, in which connections between endpoint pairs are required to satisfy modularity constraints on their length. Our algorithm is polynomial when the number of endpoint pairs and the modularity constraints on the pairs are constants.
\end{abstract}

\vfill

\pagenumbering{gobble}
\pagebreak
 \pagenumbering{arabic}

\section{Introduction}
\label{sec:intro}

We begin with a simple question: Given an $n$-vertex, $m$-edge directed graph $G$ and terminals $s$,$t$, can we efficiently find an odd-length $st$-walk that is ``edge-minimum'', i.e., has the minimum number of \emph{distinct} edges? 
This question may appear similar to classical problems from the vast literature on paths and cycles with parity constraints. 
For instance, one may think of the classical ``shortest odd $st$-path'' problem,
which is well-known to be NP-hard (this is via a simple reduction from the 2-disjoint
paths problem of~\cite{FortuneHW80} and can be proved in
the same way as~\cite[Proposition 2.1]{thomassen1985even}).
However, this hardness result only applies to \emph{simple paths}, whereas the edge-minimum odd $st$-walk may not be a simple path.
One may also think of the ``shortest odd $st$-walk'' problem (i.e., minimizing the length), which is well-known to be in polynomial time\footnote{Make two copies of the vertex set, and for every edge $u\rightarrow v$ in the original graph, add the edges $u_1\rightarrow v_2$ and $u_2\rightarrow v_1$. Then find the shortest walk from $s_1$ to $t_2$. This also trivially generalizes to modularities which are polynomial in $n$.}.
However, the \emph{edge-minimum} odd $st$-walk may be \emph{longer} than the \emph{shortest} odd $st$-walk while using \emph{fewer distinct edges}. 
See~\cref{fig:example2} for an example of a graph in which the solutions to all three of these problems is different.

The focus of this paper is the \emph{Edge-Minimum Walk of Modular Length} problem, which is the above problem generalized to an arbitrary modularity $q$
and remainder $r$. Let us define it formally:

\begin{mdframed} \textbf{Edge-Minimum Walk of Modular Length (\ewm):}\\
\noindent\textbf{Input:} An unweighted directed\tablefootnote{One could also ask
  this question on an undirected graph. However, in this case, one unnatural
  aspect of the problem is that one can always traverse the same edge back and
  forth over and over to add any multiple of 2 to the path length.
  We use this observation to show in
  \cref{apx:undirected} that \ewm on undirected graphs reduces to \ewm on directed
  graphs.} graph $G$, terminals $s$,$t$, and non-negative integers $r<q$.\\
\textbf{Output:} An $st$-walk of length $r\bmod{q}$ which is edge-minimum, i.e., uses the minimum number of distinct edges (or $\emptyset$ if no such walk exists).
\end{mdframed}

We stress that the modularity constraint does not apply to the number of distinct edges, but only to the length of the walk. 
We are most interested in the regime where $q$, and hence $r$, are constants; our algorithms will work without this assumption, but they achieve worse complexities. 

Despite the vast literature on paths and cycles of given
modularities~\cite{MR727545,LapaughP84,thomassen1985even,MR872406,VaziraniY89,MR971619,ArkinPY91,MR1218331,MR1285584,MR1445033,MR1740989,MR1744686,mccuaig2004polya,MR2847879,MR2946427,MR4371468,schlotter2022odd,BjorklundHK22,MR4538070,MR4668327,diwan2024cycles,chauhan2024even,MR4756860} (see~\cite{amarilli2024survey} for a survey),
to the best of our knowledge we are the first to study the \ewm problem.\\

\begin{figure}[h!]
\centering
  \begin{tikzpicture}[yscale=-.4,xscale=.78]
    \node (s) at (0, 0) {$s$};
    \node (a) at (4, -3) {$a$};
    \node (b) at (4, -1) {$b$};
    \node (c) at (5.5, 1) {$c$};
    \node (d) at (7, -1) {$d$};
    \node (e) at (7, -3) {$e$};
    \node (t) at (15, 0) {$t$};
    \node (f) at (4, 3) {$f$};
    \node (g) at (8, 3) {$g$};
    \node (h) at (11, 3) {$h$};
    \node (i) at (11, 1) {$i$};
    \node (j) at (9.5, -1) {$j$};
    \node (k) at (8, 1) {$k$};
    \draw[->] (s) edge[bend right=15] (a);
    \draw[->] (a) -- (b);
    \draw[->] (b) -- (c);
    \draw[->] (c) -- (d);
    \draw[->] (d) -- (e);
    \draw[->] (e) edge[bend right=15] (t);
    \draw[->] (e) -- (a);
    \draw[->] (s) edge[bend left=15] (f);
    \draw[->] (f) -- (g);
    \draw[->] (g) -- (h);
    \draw[->] (h) -- (i);
    \draw[->] (i) -- (j);
    \draw[->] (j) -- (k);
    \draw[->] (k) -- (g);
    \draw[->] (h) edge[bend left=15] (t);
  \end{tikzpicture}
  \caption{An example of a graph with different answers to the problems of finding the edge-minimum odd $st$-walk, the shortest odd $st$-walk, and the shortest odd simple $st$-path. First, every simple path from $s$ to $t$ is of even length, so there is no shortest odd $st$-path. Second, the \emph{shortest} odd $st$-walk uses the bottom cycle $s, f, g, h, i,
  j, k, g, h, t$: it has length 9 and uses 8 distinct edges. Third, the \emph{edge-minimum} odd
  $st$-walk uses the top cycle: $s, a, b, c, d, e,
  a, b, c, d, e, t$. It has length 11 but uses only 7 distinct edges.}
  \label{fig:example2}
\end{figure}
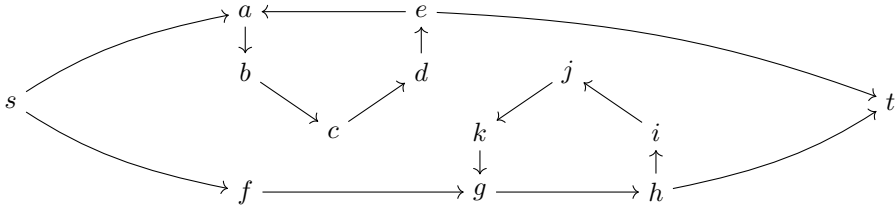

\ewm can also be viewed as a problem in the field of \emph{network design}.
Network design problems ask questions of the form ``find an edge-minimum subgraph with a certain property''. Famous network design problems include, for instance, Minimum Spanning Tree, Traveling Salesperson, and $st$-Shortest Path. We can give an equivalent rephrasing of \ewm in this way: find an edge-minimum subgraph that contains an $st$-walk of length $r \bmod{q}$. The network design problem most related to \ewm is the Directed Steiner Network problem (\dsn)~\cite{FeldmanR06,FeldmannM23,MR2786434,GalbyKMS24,DinurM18,MR4278929,MR3875932,GuoNS11,MR4786546}. In \dsn, the input is a directed graph $G$ and a set of $k$ terminal pairs ($s_1$,$t_1$),\dots,($s_k$,$t_k$); the goal is to find an edge-minimum subgraph that contains a path $s_i\rightarrow t_i$ for all terminal pairs. When $k$ is arbitrary, \dsn generalizes Directed Steiner Tree and is therefore NP-hard. When $k$ is constant, 
Feldman and Ruhl~\cite{FeldmanR06} showed that \dsn can be solved in polynomial time. 
In fact, we show in \cref{sec:dsn} that, in the special case where the set of
terminal pairs is strongly connected,
the \dsn problem with constant $k$ can be directly expressed as a special case of \ewm with constant $q$. As a consequence, our main result will imply, as a byproduct, the known result that this special case of \dsn is in polynomial time for constant $k$~\cite{FeldmanR06} (albeit with a larger exponent).
We also study a problem generalizing both \ewm and \dsn and give an algorithm that subsumes the tractability of both problems, as we will explain later in the introduction. 

\ewm is also related to problems from database theory, in particular the evaluation of \emph{regular path queries} (RPQs) on graph databases.
More precisely, a graph database is a graph whose edges are labeled with symbols from a fixed
alphabet, and an RPQ is a regular expression $e$ over the alphabet. The results of the RPQ are vertex pairs $(s,t)$ such that there is a walk from $s$ to $t$ 
(or sometimes a simple path or a trail,
depending on the variant)
whose edge labels from a word that matches~$e$ (see e.g.~\cite{CruzMW87,ConsensM90,BaganBG20,MartensNP23}).
In particular, RPQs of the form $a^r (a^q)^*$ express walks of length $r \bmod q$ for constant $r$ and $q$.
The \ewm problem then corresponds to the
\emph{smallest witness problem}~\cite{miao2019explaining,hu2024finding} for such queries,
which is the problem of finding a sub-database of minimum size that satisfies the query.
However, to the best of our knowledge there is no prior work on
the smallest witness problem for RPQs enforcing modularity conditions, or indeed for RPQs in general:
our work can be seen as a first step towards addressing this problem.

\myparagraph{Our results.}
What is the complexity of \ewm for constant $q$? It is unclear a priori; \ewm is related to both polynomial-time solvable problems such as Shortest Modular-Length $st$-Walk and Directed Steiner Network, as well as NP-hard problems such as Shortest Modular-Length $st$-Path.

Our main result is that \ewm is in polynomial time for constant $q$:

\begin{restatable}{theorem}{mainthm}\label{th:main-theorem}
There is an algorithm solving \ewm in time:
$n^{O(\log q)} \cdot 2^{O(q \log^2 q)}$.
\end{restatable}

 Specifically, the exponent of $n$ in \cref{th:main-theorem} has reasonable constants: $7 + 3\log_2 q$, though we did not focus on optimizing them.
 
\myparagraph{Generalizations.}
A natural generalization of \ewm is to consider weighted graphs. \ewm admits two natural notions of weights: (1) each edge has a \emph{length}, and the length of a walk is the sum of the lengths of the edges, and (2) each edge (or vertex) has a \emph{cost} and the edge-minimum walk is measured in terms of the sum of the costs of the edges (or vertices) that it traverses.
We show in \cref{sec:costs_distances} that our algorithm extends to accommodate costs, and that it can also accommodate lengths provided that their values are polynomial. 
In the case where lengths and the value $q$ are super-polynomial, we show that the problem is NP-complete by reduction from Subset Sum. 

Inspired by Directed Steiner Network, we also extend our results in \cref{sec:dsn} from one walk to multiple walks. The input is a directed graph $G$, and $k$ endpoint pairs
along with modularity constraints: $(s_1,t_1,q_1,r_1),\dots,(s_k,t_k,q_k,r_k)$, and the goal is to find an edge-minimum subgraph that contains a walk from $s_i$ to $t_i$ of length $r_i \bmod q_i$, for all $i$. 
Generalizing our main result, we
show (\cref{th:kdsnm-tractable}) that this problem, which subsumes the \dsn and \ewm problems, also admits a polynomial-time algorithm for constant $q_i$ values and a constant number $k$ of endpoint pairs. Our approach to show the result focuses on the case of strongly connected solutions (in line with the connection to \dsn mentioned earlier), before extending the proof to arbitrary solutions by re-using lemmas from~\cite{FeldmannM23}.

\section{Technical Overview}
\label{sec:overview}

Our work is related to the Directed Steiner Network (\dsn) problem studied by Feldman and Ruhl~\cite{FeldmanR06}, where we must find an edge-minimum subgraph 
of the input directed graph 
which satisfies connectivity requirements.
Their work shows that the \dsn problem is tractable when 
the number of endpoint pairs in the
connectivity requirements is constant.
Following their work, Feldmann and Marx~\cite{FeldmannM23} have investigated the parameterized complexity of the \dsn problem. 
One key idea of their approach is to 
show bounds on the \emph{cutwidth} of solutions to the \dsn problem.
Informally, cutwidth is a parameter that limits how many edges are ``cut'' when
arranging vertices in a well-chosen order $<$. Bounding the cutwidth of a graph
ensures that we can go over vertices in the order of~$<$ and that, at any cut
along~$<$, all edges except a constant number will be entirely to one side of
the cut. (See~\cref{sec:prelim} for the formal definition of cutwidth.)

Feldmann and Marx show that when the number of endpoint pairs is bounded by a constant~$k$,
then \dsn instances always have a solution of constant \emph{undirected} cutwidth (depending only on~$k$).
They further show that bounded-cutwidth solutions can be computed in
polynomial time by dynamic programming. We follow the same framework 
and show that solutions to \ewm also have bounded cutwidth. 
Then we can find the solution by an approach similar to the dynamic programming algorithm of~\cite{FeldmannM23}, or to the token game of~\cite{FeldmanR06}.

Thus, our main goal is to solve the purely structural problem of bounding the cutwidth of edge-minimum $st$-walks of length $r\bmod{q}$ in arbitrary directed graphs. We only focus on this goal in the rest of the technical overview.
To provide intuition, we will begin by considering the case of $q=2$, $r=1$ (i.e. odd walks), and then outline the obstacles that arise when generalizing to arbitrary $q,r$. 

\myparagraph{Odd walks.}
Let $w$ denote an edge-minimum odd-length $st$-walk, and let $G_w$ be the subgraph of $G$ consisting of the union of edges in the walk $w$. We first observe that, without loss of generality, $w$ does not contain any even cycles: If $w$ contained an even cycle, we could simply delete it, and $w$ would still be of odd length, and still be edge-minimum. (To clarify, we are not referring to even cycles in $G_w$, but rather even cycles in the walk $w$ itself. After deleting a cycle $C$ from $w$, edges from $C$ can still appear in $G_w$ if they are traversed at another point of the walk $w$.)

Now suppose $w$ contains two odd cycles in succession. Again, we could simply delete both cycles, and $w$ would still be of odd length, and still be edge-minimum. So, we can suppose without loss of generality that $w$ has either no cycles (and trivially has small cutwidth), or consists of a simple path $P_1$, followed by an odd cycle $C$, followed by a simple path $P_2$. 
We assume this latter case in the rest of the discussion, and we define $C$ by
looking at first time that the walk re-visits a previously visited vertex: this
ensures that, by construction, $P_1$ is vertex-disjoint from~$C$.
However, the graph does not necessarily have small cutwidth, because $P_2$ could
intersect $P_1$ and $C$ in intricate ways forming an unbounded number of nested
cycles. It is a priori conceivable that, e.g.,~a grid-like structure with high
cutwidth could emerge. However, we can show that this does not happen.

First, we can observe that once $P_2$ leaves $C$, without loss of generality it does not return. This is because if $P_2$ leaves $C$ at vertex $a$ and then returns at vertex $b$, we can delete the subpath of $P_2$ from $a$ to $b$, and instead traverse only the edges of $C$ to get from $a$ to $b$. This way, we can still ensure that $w$ is of odd length since every traversal of $C$ changes the parity of $w$, and we can traverse $C$ for ``free'' without adding any extra edges. 

Finally, we consider how $P_1$ and $P_2$ can interact. We can observe that without loss of generality $P_2$ cannot hit a vertex $a$ on $P_1$, then leave $P_1$, then return to a \emph{later} vertex $b$ on $P_1$. In this case we could delete the subpath of $P_2$ from $a$ to $b$, and instead traverse $P_1$ to get from $a$ to $b$. Again, we can still ensure that $w$ is of odd length since we can always traverse $C$ for free to change the parity of $w$. We stress that this argument only holds when $P_2$ re-enters $P_1$ at a \emph{later} vertex, and in fact, $P_2$ can re-enter $P_1$ at an \emph{earlier} vertex an unbounded number of times.

With all of these observations in mind, we conclude that if $w$ has a cycle, then $w$ must have the very specific structure depicted in \cref{fig:emw2}.
Then, it is visually clear that any vertical cut only has a constant number of crossing edges, and thus $w$ has constant cutwidth.

\begin{figure}
\centering
    \begin{tikzpicture}[xscale=.98,
    p1/.style={thick, blue, ->},
    p2/.style={very thick, orange, densely dotted, ->}
    ]

    \node (v0) at (0, 0) {$s$};
    \foreach \i in {1,...,5} {
        \node (v\i) at (\i, 0) {};
        \node (vb\i) at (\i, .1) {};
    }
    \node[blue] (v6) at (5.5, 0) {$\cdots$};

    \foreach \i in {6,...,12} {
        \node (v\i) at (\i, 0) {};
        \node (vb\i) at (\i, .1) {};
    }

    \foreach \i in {0,...,4} {
        \draw[p1] (v\i) -- (v\the\numexpr\i+1\relax);
    }
    \foreach \i in {6,...,11} {
        \draw[p1] (v\i) -- (v\the\numexpr\i+1\relax);
    }

        \node (o) at (13, 0) [draw,thick,decoration={markings,
        mark=at position .1 with {\arrow{>}},
        mark=at position .3 with {\arrow{>}},
        mark=at position .5 with {\arrow{>}},
        mark=at position 0.7 with {\arrow{>}},
        mark=at position 0.9 with {\arrow{>}},
        },
        postaction={decorate}, circle through=(v12),
        ] {$C$};

\node[inner sep=.25cm] (u1) at ([xshift=13.3cm]3*180/5:1) {};
\node(u2) at (10.5, 1-.05) {};
\draw (u1) edge[p2] (u2);
\draw (u2) edge[p2,bend right=20] (vb10);

\node(w1) at (9.5, -1) {};
\draw[p2] (vb10) -- (vb11);
\draw (vb11) edge[p2,bend left=20] (w1);
\draw (w1) edge[p2,bend left=20] (vb8);

\node(u3) at (7.5, 1) {};
\draw[p2] (vb8) -- (vb9);
\draw (vb9) edge[p2,bend right=20] (u3);
\draw (u3) edge[p2,bend right=20] (vb6);

\node(w2) at (5.5, -1) {};
\draw[p2] (vb6) -- (vb7);
\draw[p2] (vb7) .. controls +(-.25,-.5) and +(.5,0) .. ([xshift=.5cm]w2.east);

        \node[orange] (w6) at (5.5, -1) {$\cdots$};
        \node[orange] (u6) at (5.5, 1) {$\cdots$};

\node(w7) at (4.5, -1) {};
\draw (w7) edge[p2,bend left=20] (vb3);
\draw[p2] (vb3) -- (vb4);

\node(u8) at (2.5, 1) {};
\draw (v4) edge[p2,bend right=20] (u8);
\draw (u8) edge[p2,bend right=20] (vb1);

\node(w9) at (1, -1) {};
\node(t) at (0, -1) {$t$};
\draw[p2] (vb1) -- (vb2);
\draw (v2) edge[p2,bend left=20] (w9);
\draw[p2] (w9) -- (t);

\coordinate (k1) at (0,1.8);
\coordinate (k2) at ([xshift=1.6cm]k1);
\coordinate (k3) at ([xshift=1.6cm]k2);
\draw[p1] (k1) -- +(.6,0) node[right,p1] {$P_1$};
\draw[p2] (k2) -- +(.6,0) node[right,p2] {$P_2$};
\draw[p2] ([yshift=.06cm]k3) -- +(.6,0) node[right] {${\color{blue}P_1}\,{\color{black}\cap}\, {\color{orange}P_2}$};
\draw[p1] ([yshift=-.06cm]k3) -- +(.6,0);

\draw[Grey!40,thin] ([xshift=-.2cm, yshift=.3cm]k1) rectangle +(6,-0.7);
\end{tikzpicture}

\caption{Specific shape of solutions to \ewm for $q=2$. Each ``edge'' in the figure represents a subpath, not necessarily a single edge. The cycle $C$ is of odd length. %
}
\label{fig:emw2}
\end{figure}
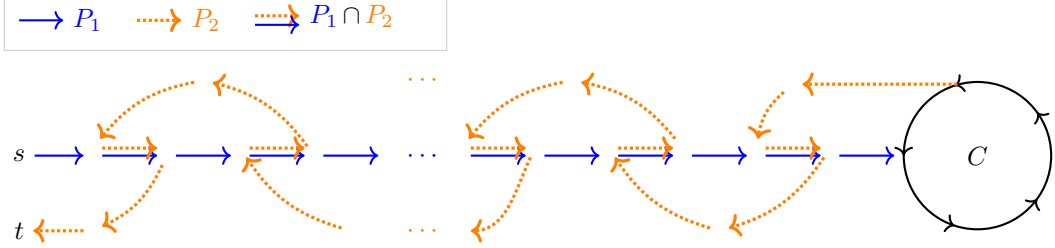

This concludes the outline for the case of odd walks. The same argument holds for even walks, and more generally an argument for a particular choice of constants $q,r$ is easily extendable by reduction to other constant values $r'$ with the same $q$,
simply by adding a new source connected to the original source by a path of constant length $r'-r$ mod~$q$. The challenging part is extending to arbitrary $q$.

\myparagraph{Challenges in extending to arbitrary $q$.}
Given $q,r$, let $w$ denote an edge-minimum $st$-walk of length $r \bmod q$. If $w$ happens to have a cycle $C$ whose length is co-prime with $q$, then we can use a similar argument to the case of odd walks above. This is because traversing the co-prime cycle $C$ over and over allows us to achieve \emph{every} remainder $r$. Then, we can simplify the walk $w$ in a similar way to above without worrying about the modularity changes resulting from our simplifications, because we can always reverse these modularity changes by traversing $C$ a different number of times for free.

The complications arise when $w$'s only cycles are not co-prime with $q$. In this case, the above arguments fall through, and $w$ can exhibit various kinds of complicated behavior: (a) $w$ could leave a cycle and then return to the same cycle, and (b) $w$ could intersect a previous subwalk, then leave, then return to a later point on that subwalk.

Our approach is to bound the cutwidth of such walks~$w$ by intuitively showing
that such behavior can only happen to change the remainder of the length of the
walk modulo~$q$, and so only
happens a number of times that depends on~$q$. More precisely, we define a
\emph{segment decomposition} of the walk and prove that each successive segment
allows us to expand the reachable set of remainder values, so that the number of
segments is bounded as a function of~$q$
(\cref{lem:bounding-number-of-segments}). We then show that the cutwidth of a
walk can be bounded linearly in its number of segments (\cref{prp:segment-bound}): we do this by defining 
a suitable ordering of the vertices following the \emph{chunks} of the walk,
which refine the segments in the decomposition. These two results imply that
solutions to \ewm have bounded cutwidth and hence that they can be found in
polynomial time.

\section{Preliminaries}\label{sec:prelim}
In this section, we present preliminary notions used throughout the proof of our main result (\cref{th:main-theorem}).

\myparagraph{Basic definitions.}
All graphs in the paper are finite and directed. Graphs are not necessarily simple; they may contain self-loops, but they do not contain multi-edges.
Given a directed graph $G =
(V, E)$, a \emph{walk} $w$ in $G$ is a sequence of edges
$w[1],\dots,w[\ell]$ such that the target vertex of $w[i]$ is equal to the source vertex of~$w[i+1]$ for each $1 \leq i < \ell$.
Note that an edge $e$ of~$G$ may occur at several positions in~$G$: considering a position $i$ for which $w[i] = e$,
we write $w[i]$ to refer to that occurrence of~$e$ at the $i$-th position of~$w$.
However, when convenient we abuse notation and identify $w[i]$ with~$e$ itself; e.g., we talk about the source and target vertices of~$w[i]$ to mean those of~$e$.
For $1 \leq i \leq \ell$ we call $w[i]$ 
a \emph{first-visited} edge occurrence if it is the first occurrence of some edge~$e$, i.e., if there is no $j < i$ such that $w[j] = w[i]$. Otherwise, $w[i]$ is a \emph{revisited} edge.
Of course, each edge that occurs in $w$ is first-visited at exactly one position.

The \emph{source vertex} $s$ of~$w$ is that of $w[1]$, and its \emph{target vertex} $t$ is that of~$w[\ell]$:
we then call $w$ an \emph{$st$-walk}.
The \emph{length} $|w|$ of~$w$ is~$\ell$.
The \emph{set of vertices used by~$w$}, denoted $V_w$, is simply the set of vertices that occur in~$w$, i.e., occur in some edge of~$w$; in particular the source and target vertices of~$w$ are in~$V_w$.
The \emph{set of edges used by~$w$} is 
$E_w := \{w[i] \mid 1 \leq i < \ell\}$.
The \emph{subgraph spanned by~$w$} is $G_w = (V_w, E_w)$.
Note that $G_w$ is \emph{not} the subgraph induced by~$V_w$,
as it only contains the edges that actually occur on the walk.
A \emph{subwalk} of~$w$ is a contiguous subsequence of~$w$.
We denote subwalks of $w$ with the slicing convention, i.e., $w[:i]=w[1]\dots
w[i]$ and $w[i:j]=w[i]\dots w[j]$. Note that the right endpoint is included so that, e.g., $w[i:j]$ is empty precisely when $j<i$,
and $w[:i]$ is empty precisely when $i \leq 0$.
When $u$ and $v$ are walks and when the target vertex of~$u$ is the source vertex of~$v$, we write $uv$ to mean the walk obtained by concatenating $u$ and $v$: it admits $u$ and $v$ as subwalks, and of course $|uv| = |u| + |v|$.

For convenience, throughout the paper we denote by $G_{w,i}$ the graph
$G_{w[:i]}$ 
spanned by the subwalk $w[:i]$ (up to $i$ included); by convention
$G_{w,0}$ is always the graph with the single vertex~$s$ and no edges.

\myparagraph{Strongly connected components.}
A \emph{strongly connected component} (SCC) of a graph~$G$ is a maximal set of
vertices $C$ of~$G$ such that, for any two distinct vertices $u, v \in C$, there
is a directed path from $u$ to~$v$ in~$G$.  As usual, belonging to the same SCC is an
equivalence relation, so that SCCs form a partition of the
vertices of~$G$.
We say that an SCC $C$ of~$G$ is \emph{non-trivial} if the subgraph
of~$G$ induced by $C$ contains at least one edge. This is the case if and only
if $C$ contains more than one vertex, or contains a single vertex on which there
is a self-loop.

\myparagraph{\ewm and solutions.}
We study the \emph{Edge-Minimum Walk of Modular Length} problem (\ewm), as
defined in the introduction: given a directed graph $G = (V, E)$, terminals $s$,$t$, and
non-negative integers $r<q$, we wish to compute a subset $E'$ of~$E$ such that
$(V, E')$ has an $st$-walk of $r\bmod{q}$ and
$E'$ is edge-minimum (i.e., the number of edges of $G'$ is minimum). This phrasing is
somewhat different from the one given in the introduction, in which we wanted
to compute an edge-minimum walk. However, we can easily compute an
edge-minimum walk from~$E'$ by a product construction, in time $O(|E'| \times
q)$.

In the sequel, we only consider the computation of edge-minimum
subsets of edges (rather than walks). When fixing inputs $G = (V, E)$, $s$, $t$, $q$, and $r$, we talk about a
\emph{candidate solution} to mean a subset $E'$ of~$E$ such that $(V, E')$ has an $st$-walk of
length $r\bmod{q}$ but $E'$ is not necessarily edge-minimum, and of an
\emph{optimal solution} to mean a candidate solution which is additionally edge-minimum.

\myparagraph{Cutwidth.}
The \emph{cutwidth} is a structural parameter of graphs. We show that there always exists an optimal solution to \ewm with bounded cutwidth, which suffices to ensure that such solutions can be found with an efficient algorithm.
We first recall the definition of cutwidth for undirected graphs.
Let $G=(V,E)$ be an undirected graph.
An \emph{ordering} of~$G$ is a total order over the vertex set~$V$ of~$G$.
A \emph{cut} $(V_-, V_+)$ of $G$ that respects the ordering $<$ 
is a partition $V_-\uplus V_+=V$ such that $v_-<v_+$ for all $(v_-,v_+)\in V_-\times V_+$.
We say that an edge $e$ of $G$ \emph{crosses the cut} if it has one
endpoint in $V_-$ and one in $V_+$, i.e.,
$e \cap V_-$ and $e \cap V_+$ are both nonempty.
Note that self-loops never cross cuts.
The \emph{cutwidth} of a cut $(V_-, V_+)$ that respects $<$ is the number of edges that cross
this cut, and the \emph{cutwidth} of $<$ is the maximum cutwidth of a cut that
respects~$<$.
The \emph{cutwidth} of $G$ is then the minimum cutwidth of an ordering of~$V$.

We now define cutwidth for directed graphs. Let $G = (V, E)$ be a directed graph. Its \emph{underlying undirected graph} is $G' = (V, E')$ where $E' = \{\{u, v\} \mid (u, v) \in E, u \neq v\}$: note that self-loops are not reflected in~$G'$. We then define the \emph{cutwidth} of~$G$ to be that of~$G'$. We underscore
that our definition of cutwidth for directed graphs is always the undirected cutwidth:
we do not consider the directed cutwidth in this paper.
We then define the \emph{cutwidth} of a walk $w$ of~$G$ as the cutwidth of the directed graph~$G_w$.

\section{Segment Decomposition of a Walk}\label{sec:segment}
Our proof of our main theorem (\cref{th:main-theorem}) consists of three steps,
spanning this section and the next two sections. First, in this section we introduce the notion of the \emph{segment decomposition} of a walk, and we show that any 
walk $w$ 
can be transformed into
a walk $w'$ that satisfies the same modularity conditions (i.e., has the same remainder modulo~$q$),
uses a subset of the edges (i.e., $E_{w'} \subseteq E_w$),
and has a segment decomposition with $O(\log q)$ segments.
Second, in the next section, we will show how the cutwidth of the graph $G_{w'}$ spanned by $w'$ can be bounded linearly in this number of segments, implying that 
optimal solutions to \ewm have bounded cutwidth. Third, 
in \cref{sec:cw-algo},
we show that this cutwidth bound makes it possible to solve \ewm in polynomial time.

\myparagraph{Segment decomposition.} 
Let $w$ be a walk. Its \emph{segment decomposition} is a sequence of walks $s_1, \ldots, s_\xi$ such that $w = s_1 \cdots s_\xi$; the number $\xi$ is the \emph{number of segments} of~$w$. \Cref{fig:walk} depicts an example of a segment decomposition of a walk. The segment decomposition is defined by processing the walk~$w$ from left to right as follows. Initially, the segment decomposition is the empty sequence, and the entire walk $w$ remains to be decomposed. At any point of the decomposition, we have already computed the first $\sigma$ segments $s_1, \ldots, s_\sigma$ for some number $\sigma \geq 0$ of segments, and can write $w = s_1 \cdots s_\sigma w'$. If $w'$ is non-empty, we 
compute the next segment as a prefix of~$w'$. 

Informally, we read $w'$ and terminate the segment whenever we have a first-visited edge $w'[j] = (u,v)$ followed (not necessarily immediately) by an edge $w'[k] = (x,y)$ where $y$ can be reached by a path from~$u$ using only edges of~$w$ up to~$w'[j]$ excluded. The intuition of segments is that they capture a moment where the walk revisits a vertex $y$ from a vertex $u$ via a path that starts with an edge $e = (u,v)$ which had not been previously visited in~$w$, even while there was already a path from $u$ to $y$ in the graph spanned by the walk before $e$ was visited. 
Formally, we will look for the \emph{segment end} inside the suffix~$w'$ of~$w$, according to the following definition:

\begin{definition}[Segment end, Segment detour]
  \label{def:segmentend}
Having fixed $w = s_1 \cdots s_\sigma w'$,
let $\lambda_\sigma = |s_1 \cdots s_\sigma| = |w|-|w'|$
  be the total length of the already-computed segments.
For $\lambda_\sigma < k \leq |w|$, we say that segment $s_{\sigma+1}$ \emph{ends at~$k$} if $k$ is minimum such that the following hold:
\begin{itemize}
\item There is $\lambda_\sigma < j \leq k$ such that $w[j]$ is first-visited in~$w$; we write $w[j]$ as $(u, v)$;
\item Writing the edge $w[k]$ as $(x,y)$,
then there is a path from $u$ to~$y$ in the
graph $G_{w,j-1}$
    which contains the edges of the walk $w$ up to $w[j]$ excluded. Note that, in particular, if $u=y$ then this requirement is always satisfied using the empty path from $u$ to $y$.
\end{itemize}

We stress that the path from $u$ to~$y$ required by the second item of the definition is a path in the \emph{graph spanned} by a certain prefix of~$w$; it may be the case that the path does not occur as a subwalk of~$w$.
As for the subwalk $w[j:k]$ from $u$ to $y$ in~$w$, which is a subwalk of~$w'$, we call it the \emph{segment detour} of $s_{\sigma+1}$, denoted $\det_{\sigma+1}$.
Note that, having chosen the minimal~$k$, there could be multiple valid choices for $j$, and we pick one arbitrarily.
\end{definition}

Given the definition of a segment end above, the next segment of the
decomposition after $s_\sigma$ is simply $s_{\sigma+1} :=
w[\lambda_\sigma+1:k]$.
This means that we terminate the $(\sigma+1)$-th segment right after the edge $w[k]$, and continue the decomposition if $s_1 \cdots s_{\sigma+1}$ does not yet cover the entire walk~$w$. If there is no segment end at any~$k$, then we take all the rest of the walk~$w'$ to be the last segment of the decomposition and we finish.
Note that for this reason, unlike other segments, the last segment of the decomposition may not feature a segment detour.

\begin{example}
We give a running example of a walk
in  \cref{fig:walk}, together with its segment decomposition.
 Note that edges may be visited by several different segments, which is
  illustrated by parallel edges and the indexes on edges in \cref{fig:walk}.
  Further note that in general it may also be the case that the same edge is revisited multiple times by the same segment, although this does not happen in the example.

  Let us clarify why the successive segments end:
\begin{itemize}
    \item \textcolor{csegm1}{Segment 1} ends because it takes the first-visited edge $(v_2, v_3)$ and immediately afterwards takes the edge $(v_3, v_2)$ so the segment ends, as this loop forms a detour with the trivial path from $v_2$ to itself.
    \item \textcolor{csegm2}{Segment 2} ends for a similar reason: it takes a first-visited edge $(v_{2}, v_{4})$, and afterwards reaches $v_{2}$ again.
    \item \textcolor{csegm3}{Segment 3} ends for a similar reason: it takes a first-visited edge $(v_{14}, v_{15})$, and afterwards reaches $v_{14}$ again.
    \item \textcolor{csegm4}{Segment 4} ends because it takes the first-visited edge $(v_{15}, v_{18})$, and then reaches $v_{16}$, but before the segment started there was already a path from $v_{15}$ to $v_{16}$.
    \item \textcolor{csegm5}{Segment 5} exemplifies how segments can alternate between first-visited edges and revisits for an unbounded number of times without terminating: this pattern resembles the one from \cref{sec:overview},
    and we will prove in \cref{lem:segbound} (see also \cref{fig:scc-chunks}) that this pattern does not prevent us from bounding the contribution of the segment to the cutwidth. Note how, for instance, the first-visited edge $(v_{17}, v_{13})$ cannot be used for a segment end, because the segment never revisits a vertex to which $v_{17}$ has a path using the edges that existed until then. The end of segment 4 is then similar to segment 1.
 \item \textcolor{csegm6}{Segment 6} shows how a segment can end by taking an edge which is not first-visited. Indeed, the segment takes the first-visited edge $e = (v_{23}, v_{21})$, and then takes the edge $(v_{21}, v_{22})$, but before $e$ there was a (single-edge) path from~$v_{23}$ to $v_{22}$.
 \item \textcolor{csegm7}{Segment 7} ends because the walk ends.
\end{itemize}
\end{example}

\begin{figure}[h!]
\centering
\begin{tikzpicture}[xscale=1.1,yscale=1.15,
	ve/.style={minimum width=2.5mm,inner sep=0.0cm,draw,circle,black},
	segm/.style={->,>=stealth,font=\footnotesize,thick},
    segm1/.style={segm,csegm1},
    segm2/.style={segm,csegm2},
    segm3/.style={segm,csegm3},
    segm4/.style={segm,csegm4},
    segm5/.style={segm,csegm5},
    segm6/.style={segm,csegm6},
    segm7/.style={segm,csegm7},
	every label/.style={inner sep=0.01cm,label distance=.08cm,font=\footnotesize,black},
	wi/.style={font=\notsotiny\sffamily,red},
	wab/.style={above=2ex,wi},
	wbe/.style={below=2ex,wi},
	wabr/.style={above right=1ex and 0.5ex,wi}
]

\foreach \i [count=\j from 20] in {1,...,5} {
\node[ve,label={90:\vname{\j}}] (na\i) at (-.5+\i*1.5,3) {};
}

\node[ve,label={90:\vname{4}}] (nm1) at (3,1.5) {};
\node[ve,label={90:\vname{5}}] (nm2) at (4,1.5) {};

\foreach \i in {1,2} {
\node[ve,label={90:\vname{\i}}] (n\i) at (\i,0) {};
}
\foreach \i [count=\j from 6] in {3,...,12} {
\node[ve,label={90:\vname{\j}}] (n\i) at (\i,0) {};
}

\node[ve,label={270:\vname{3}}] (nb1) at (2,-1) {};
\node[ve,label={270:\vname{16}}] (nb2) at (12,-1) {};
\node[ve,label={0:\vname{18}}] (nb2b) at (12.5,-.5) {};
\node[ve,label={270:\vname{17}}] (nb3) at (11,-1) {};
\node[ve,label={270:\vname{19}}] (nb4) at (4.5,-1) {};

\draw[segm1] (n1) -- node[above] (ne1) {1} (n2);
\draw[segm1] (n2) -- node[right] (ne2) {2} (nb1);
\draw[segm1] (nb1) .. controls  +(-1,0) and +(-1,-1) .. node[above left,segm1] (ne3) {3}(n2);

\draw[segm2] (n2) -- node[left,pos=.8] (ne4) {4}  (nm1);
\draw[segm2] (nm1) -- node[above] (ne5) {5} (nm2);
\draw[segm2] (nm2) -- node[right] (ne6) {6} (n2);

\foreach \i [count=\j from 3, count=\jlab from 7] in {2,...,11} {
\draw[segm3] (n\i) -- node[above] (ne\jlab) {\jlab} (n\j);
}
\draw[segm3] (n12) -- node[left] (ne17) {17} (nb2);
\draw[segm3] (nb2) -- node[below] (ne18) {18} (nb3);
\draw[segm3] (nb3) -- node[left] (ne19) {19} (n11);

\draw[segm4,transform canvas={yshift=.1cm}] (n11) -- node[above=1ex] (ne20) {20} (n12);
\draw[segm4] (n12) -- node[above right] (ne21) {21} (nb2b);
\draw[segm4] (nb2b) -- node[below right] (ne22) {22} (nb2);
\draw[segm5,transform canvas={yshift=-.1cm}] (nb2) -- node[below=1ex] (ne23) {23} (nb3);
\draw[segm5] (nb3) -- node[below left] (ne24) {24} (n10);
\draw[segm5] (n10) edge[bend left=50]
node[below] (ne25) {25} (n9);
\draw[segm5] (n9) edge[bend left=60] node[below] (ne26) {26} (n7);
\draw[segm5,transform canvas={yshift=-.1cm}] (n7) -- node[below] (ne27) {27}(n8);
\draw[segm5] (n8) edge[bend right=45] node[above] (ne28) {28} (n5);
\draw[segm5,transform canvas={yshift=-.1cm}] (n5) -- node[below] (ne29) {29} (n6);
\draw[segm5] (n6) .. controls +(-.2,-.5) and +(.5,0) .. node[below] (ne30) {30}(nb4);
\draw[segm5] (nb4) edge[bend left=20] node[below] (ne31) {31} (n3);
\draw[segm5] (n3) .. controls +(-.1,-.1) and +(.75,0) .. node[pos=.7,below right] (ne32) {32}(nb1);
\draw[segm5,transform canvas={shift={(-.15cm,-0.05cm)}}] (nb1) .. controls  +(-1,0) and +(-1,-1) .. node[above left=.05cm and .1cm] (ne33) {33} (n2);
\draw[segm5] (n2) edge[bend left=20] node[below left] (ne34) {34} (na1);
\draw[segm5] (na1) -- node[above] (ne35) {35}(na2);
\draw[segm5] (na2) -- node[below] (ne36) {36}(na3);
\draw[segm5] (na3) -- node[below=.8ex] (ne37) {37} (na4);
\draw[segm5] (na4) edge[bend right] node[above] (ne38) {38} (na3);
\draw[segm6,transform canvas={yshift=-.1cm}] (na3) -- node[below=1.8ex] (ne39) {39} (na4);
\draw[segm6] (na4) edge[bend right=55] node[below] (ne40) {40} (na2);
\draw[segm6,transform canvas={yshift=-.1cm}] (na2) -- node[below=1ex] (ne41) {41} (na3);
\draw[segm7,transform canvas={yshift=-.2cm}] (na3) -- node[below=2.8ex] (ne42) {42} (na4);
\draw[segm7] (na4) -- node[below] (ne43) {43} (na5);

\foreach \i in {1,...,7} {
\node[right,segm\i] at (9.5, 1.4+7*.35-.35*\i) (ns\i) {Segment \i{} =};
}
\node[segm1,right,inner sep=0cm] at (ns1.east) (ns1b) {$w[1:3]$};
\node[segm2,right,inner sep=0cm] at (ns2.east) {$w[4:6]$};
\node[segm3,right,inner sep=0cm] at (ns3.east) {$w[7:19]$};
\node[segm4,right,inner sep=0cm] at (ns4.east) {$w[20:22]$};
\node[segm5,right,inner sep=0cm] at (ns5.east) {$w[23:38]$};
\node[segm6,right,inner sep=0cm] at (ns6.east) {$w[39:41]$};
\node[segm7,right,inner sep=0cm] at (ns7.east) (ns7b) {$w[42:43]$};

\draw[Grey!40,thin] ([xshift=-.1cm, yshift=.1cm]ns1.north west) rectangle ([xshift=.1cm, yshift=-.1cm]ns7b.south east);
\end{tikzpicture}
\caption{A walk $w$ and its decomposition into  7 segments, denoted by various
colors.}
\label{fig:walk}
\end{figure}
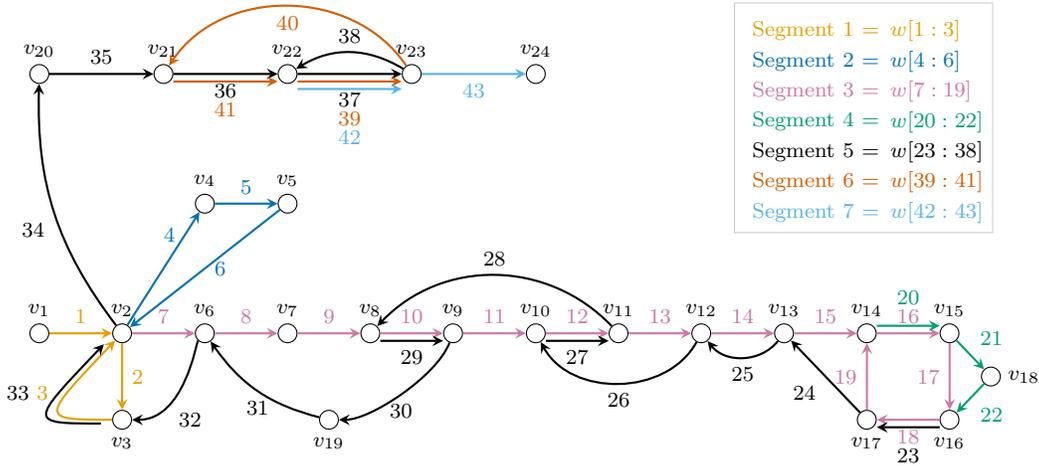

Let us formalize which paths are known to exist at the moment when a segment ends: 

\begin{lemma}
\label{lem:segment-contains-detour}
    Let $w$ be a walk and let $w[i:k]=s_\sigma$ be the $\sigma$-th segment of the walk.
    Let $y$ be the ending vertex of~$w[i:k]$, and assume that $k<|w|$, i.e., the segment finishes before the end of the walk.
    Let $w[j] = (u,v)$ in $w[i:k]$ be the first edge of the segment detour $\textnormal{det}_\sigma$. %
    There is:
    \begin{itemize}
      \item the segment detour $\textnormal{det}_\sigma$ from $u$ to~$y$ in~$G_{w,k}$  
        \item a path $p_\sigma$ from $u$ to~$y$ in $G_{w,j-1}$.
        \item a path $\overline{p_\sigma}$ from $y$ to~$u$ in $G_{w,j-1}$.
    \end{itemize}
\end{lemma}

Notice that we may have $y = u$, and then $p_\sigma$ and $\overline{p_\sigma}$ may be empty; but $\textnormal{det}_\sigma$ is never empty because it contains at least $w[j]$.
The paths are illustrated in \cref{fig:segment-paths}.

\begin{figure}[h!]
\null\hfill
\begin{tikzpicture}
    \tikzstyle{fvedge} = [->, thick, solid, blue]
    \tikzstyle{detourend} = [->, very thick, dotted, blue]
    \tikzstyle{psigma} = [->, thick, dashed, red]
    \tikzstyle{psigmabar} = [->, thin, dashed, magenta]
    \tikzstyle{socket outline} = [{rounded corners, thick, Grey, line width=8pt, opacity=0.3}]
    \node (u) at (0,0) {$u$};
    \node (y) at (0,-3) {$y$};

    \node[inner sep=.05cm] (z1) at (1.5,-.75) {};  %

    \draw[socket outline] [rounded corners](u) -- (1.5,-.75) -- node[right,opacity=1] {$\det_\sigma$ (detour)} +(0,-1.5)  -- (y);

    \draw[fvedge] (u) -- node[inner sep=0.05cm, midway, above right] {$w[j]$ (first visit)} (z1);
    \draw[detourend] (z1) -- +(0,-1.5) -- node[inner sep=0.05cm, pos=.7, below right=-0.05cm and .1cm] {$w[k]$ (end)}  (y);

    \draw[psigma] (u) edge[bend left=0] node[inner sep=0.05cm, midway, left] {$p_\sigma$} (y);
    \draw[psigmabar] (y) edge[bend left=50] node[inner sep=0.05cm, midway, left] {$\overline{p_\sigma}$} (u);

    \node[] (n1) at (6.5,-1) {${\color{red}p_\sigma}\subseteq G_{w,j-1}$};
    \node[below right] (n2) at (n1.south west) {${\color{magenta}\overline{p_\sigma}}\subseteq G_{w,j-1}$};

    \end{tikzpicture}
    \hfill\null

    \caption{Schema illustrating the two paths $p_\sigma$ and $\overline{p_\sigma}$,
    the detour $\mathrm{det}_\sigma$, and the two edges $w[j]$ and $w[k]$ that
    are part of $\mathrm{det}_\sigma$,
    from \cref{lem:segment-contains-detour}. The edge $w[j]$ is the only edge
    of the detour which is guaranteed to
    be a first-visited edge.}
    \label{fig:segment-paths}
\end{figure}
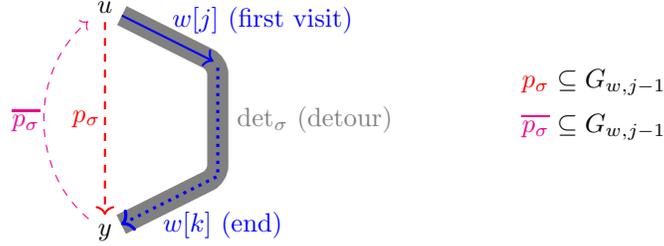

\begin{proof}[Proof of Lemma~\ref{lem:segment-contains-detour}]
  The segment detour $\textnormal{det}_\sigma$ exists by definition, and the path $p_\sigma$ is the path which is also required to exist by definition of segments. For the path $\overline{p_\sigma}$, the claim is vacuous if $y = u$, so assume that $y \neq u$. Then the path $p_\sigma$ from $u$ to~$y$ is non-empty. Let $w[j']$ be its last edge: the target of this edge is~$y$, and we have $j' \leq j-1$ because $p_\sigma$ is a path in $G_{w,j-1}$. We know that $j'\not=j-1$ because the target of $w[j']$ is $y$ while the target of $w[j-1]$ is $u$ and $y\not=u$; hence, $j'\leq j-2$.
   Then, $\overline{p_\sigma} = w[j'+1:j-1]$ is a subwalk of~$w$ from $y$ to~$u$, hence in particular $\overline{p_\sigma}$ witnesses that there is a path from~$y$ to~$u$ in~$G_{w,j-1}$.
\end{proof}

\begin{example}
    Referring back to \cref{fig:walk}, note that 
    $\textcolor{csegm1}{p_1}$, $\textcolor{csegm1}{\overline{p_1}}$, 
    $\textcolor{csegm2}{p_2}$, $\textcolor{csegm2}{\overline{p_2}}$, 
    $\textcolor{csegm3}{p_3}$, $\textcolor{csegm3}{\overline{p_3}}$,
    $\textcolor{csegm5}{p_5}$, and $\textcolor{csegm5}{\overline{p_5}}$, are empty paths from a vertex to itself. 
    The path $\textcolor{csegm4}{p_4}$ consists of the edge $(v_{15}, v_{16})$ and 
    $\textcolor{csegm4}{\overline{p_4}}$ goes from $v_{16}$ to $v_{15}$ via $v_{17}$ and $v_{14}$. 
    The path $\textcolor{csegm6}{p_6}$ is the singe edge $\textcolor{csegm5}{(v_{23}, v_{22})}$ and 
    $\textcolor{csegm6}{\overline{p_6}}$ is the reverse edge. 
    There are no paths $\textcolor{csegm7}{p_7}$ and $\textcolor{csegm7}{\overline{p_7}}$ because \textcolor{csegm7}{Segment 7} does not finish before the end of the walk.
\end{example}

In the sequel for each $1 \leq \sigma < \xi$ we denote by $p_\sigma$ and $\overline{p_\sigma}$ a pair of paths obtained by applying \cref{lem:segment-contains-detour} (if there are multiple valid choices for $p_\sigma$ and $\overline{p_\sigma}$, we choose arbitrarily).

\myparagraph{Achievable differences.}
For a walk $w$ and segment $\sigma$, we denote by $\Delta_q(w, \sigma)$ the value $(|\det_\sigma| - |p_\sigma|) \bmod q$
and call this the\footnote{%
Note that there could be several choices for the difference achievable by~$\sigma$ depending on the arbitrary choices made earlier, e.g., depending on the first edge $j$ for the segment detour, and on the choice of paths $p_\sigma$ and $\overline{p_\sigma}$; nevertheless, we fix 
one single value $\Delta_q(w, \sigma)$ according to the choices made.}
\emph{difference achievable by~$\sigma$}. 
We often drop the subscript $q$ when clear from context. Intuitively, we know that we can modify the walk $w$ to replace the subwalk $\det_\sigma$ from $u$ to $y$ by the existing path $p_\sigma$ that goes from $u$ to~$y$, and this change subtracts $\Delta_q(w,\sigma)$ from the length of the walk modulo~$q$; we will make this formal in the sequel.
It should also be intuitively clear that achievable differences can be assumed to be non-zero in an edge-minimum walk: intuitively, an achievable difference of 0 means that we can replace the detour $\det_\sigma$ by the existing path $p_\sigma$ that has the same endpoints, without changing the length of the walk modulo $q$. We know that this change does not harm the edge-minimality of $w$ because $p_\sigma$ consists of already visited edges.
What is more, having too many segments with the same achievable difference is useless, because we can intuitively replace $\det_\sigma$ by $p_\sigma$ and compensate the effect of this change by doing similar substitutions in earlier segments. 
Then, thanks to Lagrange's theorem on the additive group $\Z{q}$, the number of segments can in fact be bounded by $1+\log_2 q$.
We formalize this in the following lemma, which is the main result of this section:

\begin{lemma}[Segment Decomposition Lemma]\label{lem:bounding-number-of-segments}
For every $st$-walk $w$ with length $r\bmod q$, there is another $st$-walk of
  length $r \bmod q$ using only edges from $w$ whose number of segments is at most $1+\log_2 q$. 
\end{lemma}

The rest of this section is devoted to proving the Segment Decomposition Lemma; we will then use it in the next section to bound the cutwidth of some edge-minimum solution.

Recall that a \emph{preorder} over a set is a binary relation that is both reflexive and transitive (this is a weaker requirement than non-strict partial orders which are additionally required to be antisymmetric).
To prove the Segment Decomposition Lemma (\cref{lem:bounding-number-of-segments}), we define a preorder over  the edge-minimum walks that are using a specific set of edges. 
Intuitively, the preorder favors walks which do their first visit of edges as late as possible relative to the end of the walk.
More precisely, the preorder criterion asks us to minimize the number of remaining steps of the walk at the moment where the last first-visited edge is visited; then break ties by minimizing the number of remaining edges at the moment where the second-to-last first-visited edge is visited; and so on. This is designed to ensure that replacing a detour $\det_\sigma$ by $p_\sigma$ improves the walk according to the preorder.

Let us fix from now on the graph $G = (V, E)$, the source $s$ and target $t$, and the integers $r$ and $q$. 
For any subset of edges $E' \subseteq E$, we define an \emph{$E'$-walk} to mean an $st$-walk of length $r\bmod{q}$ in $G$ which uses
precisely the edges of~$E'$. Let us then define the order:

\begin{definition}
\label{def:order}
Let $E' \subseteq E$ be a subset of edges and let $m := |E'|$. Let $w$ be an $E'$-walk, and let $i_1, \ldots, i_m$ be the indexes of the first-visited edges in ascending order, i.e., $w[i_1], \ldots, w[i_m]$ are the first-visited edges of $w$ and $i_1 < \cdots < i_m$. The \emph{first-visited timestamp} of $w$ is then the $m$-tuple $(|w|-i_m, \ldots, |w|-i_1)$.

We define as follows a preorder relationship $\trianglelefteq$ called the
  \emph{timestamp preorder} on $E'$-walks. 
Let $w$ and $w'$ be two $E'$-walks, and let $t_w$ and $t_w'$ be their first-visited timestamps, respectively. Then we have $w \trianglelefteq w'$ if we have $t_w \leq t_w'$ in lexicographic ordering.
\end{definition}

In other words, the timestamp preorder compares two walks by the number of remaining edges to traverse when visiting the last first-visited edge, then the second-to-last, and so on. Note that the timestamp preorder is not antisymmetric
because two different walks on the same set of edges may happen to have the same first-visited timestamp (i.e., they visit first-visited edges at the same positions from the end, even though the identity of these edges and the revisits may be different).
We then define:
\begin{definition}
  \label{def:timestampmin}
If $w$ is an $E'$-walk, we say $w$ is \emph{timestamp-minimum} if, for every $E'$-walk $w'$, we have $w \trianglelefteq w'$.
\end{definition}

Note that, whenever there is an $E'$-walk, then there is a timestamp-minimum $E'$-walk, but it is not necessarily unique because the timestamp preorder is not antisymmetric.

We now turn to the proof of
\cref{lem:bounding-number-of-segments}: we want to show that a timestamp-minimal
walk has at most $1 + \log_2 q$ segments. The intuition is the following.
  We will consider the successive subgroups of $\Z{q}$ generated by the achievable
  differences of the successive segments, and show that these must be a strictly
  increasing subsequence, so that the bound follows by Lagrange's theorem. To do
  this, we will show an intermediate claim (dubbed (*) in the formal proof
  below) according to which we can modify walks to augment their length by any combination of achievable differences of previous segments, while still using the same edges. Thanks to this claim, assuming by
  contradiction that there is a segment $s_\sigma$ whose achievable difference is already
  achievable using the preceding segments, we will rewrite the segment to
  replace its detour $\textnormal{det}_\sigma$ by the path $p_\sigma$
  from \cref{lem:segment-contains-detour}, and use the claim to modify the walk
  and fix its length. We will show that this modification yields a walk which is
  smaller in the timestamp preorder, contradicting the minimality of~$w$.

We are now ready to conclude the section and give the formal proof of the  Segment Decomposition Lemma (\cref{lem:bounding-number-of-segments}): 
\begin{proof}[Proof of \cref{lem:bounding-number-of-segments}]
Let $E'$ be an optimal solution and let $w$ be a timestamp-minimum $E'$-walk. Our goal is to show that $w$ has at most $1+\log_2 q$ segments. 

Let $w = s_1 \cdots s_\xi$ be the segment decomposition of~$w$.
For each $1 \leq \sigma < \xi$, i.e., for each segment $s_\sigma$ except the last, 
we let $S(w, \sigma)$ denote the subgroup of $\Z{q}$ generated by the achievable differences of the preceding segments, i.e., $S(w, \sigma)$ is generated by $\{\Delta(w, \sigma')\mid \sigma' \leq \sigma\}$.
For convenience, we write $S(w,0)$ to denote the singleton group containing only the identity $0_{\Z{q}}$.
By definition, as $\sigma$ increases, the $S(w,\sigma)$ form a sequence $\mathcal{S}$ so that each subgroup is contained in the next subgroup in the sequence.
The idea of the proof is to show that each subgroup in this sequence $\mathcal{S}$ must in fact be strictly contained in the next subgroup in~$\mathcal{S}$.

The proof consists of three steps. First, we characterize which walk lengths can be achieved using the first $\sigma$ segments of the walk $w$, using the notion of achievable differences: we show this as an invariant, which we dub (*). Second, using this invariant, we show that each subgroup in sequence $\mathcal{S}$ must be strictly contained in the next subgroup in $\mathcal{S}$: two consecutive subgroups in $\mathcal{S}$ 
that are equal would show a violation of the timestamp-minimality of~$w$. Third, we use the latter claim to conclude that there are at most $1 + \log_2 q$ segments thanks to Lagrange's theorem.

\medskip

\textbf{Step 1: An invariant on achievable differences.}
Let us first prove by induction on~$\sigma$ the following invariant, dubbed (*): for each $0 \leq \sigma < \xi$, letting $y_\sigma$ be the end of the segment $s_\sigma$ (or $y_\sigma = s$ for $\sigma = 0$), for each $r' \in S(w,\sigma)$, there is a walk from $s$ to~$y_\sigma$ using exactly the edges of $s_1 \cdots s_\sigma$ and whose length modulo $q$ is precisely $|s_1 \cdots s_\sigma| + r'$.

The base case of the induction is trivial with the empty path from $s$ to itself, so let us show the inductive case. 
Fix $1 \leq \sigma < \xi$. 
Let $r'$ be the element of $S(w,\sigma)$ to achieve: 
we can write it as 
$r'' + c \Delta(w,\sigma) \bmod q$
for some $r''\in S(w,\sigma-1)$ and $c \in \mathbb{N}$, and we can further ensure that $c <q$.
Using the inductive hypothesis, 
we obtain a walk $w_{r''}$ from $s$ to $y_{\sigma-1}$ using exactly the edges of $s_1 \cdots s_{\sigma-1}$ and whose length modulo~$q$ is $|s_1 \cdots s_{\sigma-1}| + r''$. 
Recalling the definition of $\overline{p_\sigma}$ and $p_\sigma$ from \cref{lem:segment-contains-detour}, 
let us build $s_\sigma' = s_\sigma (\overline{p_\sigma} p_\sigma)^q$. Note that $s_\sigma'$ is still a walk from $y_{\sigma-1}$ to $y_\sigma$ because $s_\sigma$ goes from $y_{\sigma-1}$ to $y_\sigma$, $\overline{p_\sigma}$ goes from $y_\sigma$ to the vertex $u$ defined in~\cref{lem:segment-contains-detour}, and $p_\sigma$ goes from $u$ back to $y_\sigma$.
Furthermore, $|s_\sigma'|\equiv |s_\sigma|\bmod q$ because the length of the inserted part is a multiple of~$q$. 

Now we build $s_\sigma''$ from $s_\sigma'$ to ensure that $|s_\sigma''| \equiv |s_\sigma| + c \Delta(w,\sigma) \bmod q$, simply by replacing $c$ occurrences of $p_\sigma$ by $\det_\sigma$ (which have the same endpoints) in~$s_\sigma'$.
Then we consider the walk $w' := w_{r''} s_\sigma''$. It is a walk from $s$ to~$y_\sigma$
whose length differs from $s_1 \cdots s_{\sigma}$ by adding $r' = r'' + c \Delta(w,\sigma)$.
Further, it uses exactly the same edges as $s_1 \cdots s_\sigma$.
Indeed, all edges used in $w'$ also occur in~$w$. Conversely,
all edges of $s_1 \cdots s_{\sigma-1}$ occur in $w_{r''}$ by the inductive hypothesis, and $s_\sigma$ is a prefix of $s_\sigma''$ so all its edges appear in $s_\sigma''$.
Hence, the walk $w'$ allows us to conclude the proof of invariant (*).  (Note that we do not claim to know anything about the segment decomposition of~$w'$, which could be completely different from that of $w$.)

\medskip
\textbf{Step 2: All subgroups are different thanks to timestamp-minimality.}
Let us now show that each subgroup in sequence $\mathcal{S}$ must be strictly contained in the next subgroup in $\mathcal{S}$, thanks to the timestamp-minimality of~$w$. Indeed, let us assume by contradiction that this is not true, so that there is $2 \leq \sigma < \xi$ such that $S(w,\sigma)$ is not a strict superset of~$S(w,\sigma-1)$. By definition, as $S(w,\sigma-1) \subseteq S(w, \sigma)$,
this must mean that $S(w,\sigma) = S(w,\sigma-1)$ and that $\Delta(w,\sigma) \in S(w,\sigma-1)$. We will use this to build a walk that contradicts the timestamp-minimality of~$w$.
Recall that $w = s_1 \cdots s_\xi$.
We will modify $w$ in two ways: first replace $\det_\sigma$ by $p_\sigma$ in $s_\sigma$, then use invariant (*) to change the first $\sigma-1$ segments of the walk
in order to adjust the remainder. 
Specifically, let $y_{\sigma-1}$ be the end vertex of $s_{\sigma-1}$.
From our assumption that $\Delta(w,\sigma) = |\det_\sigma| - |p_\sigma|$ is in $S(w, \sigma-1)$, 
we know by invariant (*)
that there is a walk $w'$ from $s$ to $y_{\sigma-1}$ whose length modulo $q$ is $|s_1 \cdots s_{\sigma-1}| + |\det_\sigma| - |p_\sigma|$ and which uses exactly the edges of $s_1 \cdots s_{\sigma-1}$.  

Now, let us build
$w''$, which is obtained from $w$ by replacing the first $\sigma-1$ segments of $w$
by~$w'$, and by replacing $\det_\sigma$ by $p_\sigma$ in~$s_\sigma$. Specifically, let $\hat{s}_\sigma$
be the portion of $s_\sigma$ before $\det_\sigma$. Then $w''=w' \hat{s}_\sigma p_\sigma s_{\sigma+1} \cdots s_\xi$. Note that $w''$ is still an $st$-walk because $p_\sigma$ is a path with the same endpoints as $\det_\sigma$, and $w'$ is a walk with the same endpoints as $s_1 \cdots s_{\sigma-1}$. Furthermore, this walk $w''$ uses exactly the edges of~$E'$: Indeed $w'$ uses exactly the same edges as $s_1 \cdots s_{\sigma-1}$, the path $p_\sigma$ uses edges of 
$s_1 \cdots s_{\sigma-1}$ by definition of a segment,
and the other edges of the walk are already edges of~$w$. Thus $w''$ uses no other edges than $E'$, and by the edge-minimality of $w$, $w''$ must use all of the edges of $w$.
What is more, $|w''| \equiv |w| \bmod q$ since the replacement of $\det_\sigma$ by $p_\sigma$ and the replacement of $s_1 \cdots s_{\sigma-1}$ by $w'$  cancel each other out.

It remains to show that we do not have $w \trianglelefteq w''$, and we will have reached a contradiction because we assumed $w$ to be timestamp-minimum. We recall that $w$ ends with $\hat{s}_\sigma \det_\sigma s_{\sigma+1} \cdots s_\xi$ while 
$w''$ ends with $\hat{s}_\sigma p_\sigma s_{\sigma+1} \cdots s_\xi$; what is more, in $w$ this is preceded by $s_1 \cdots s_{\sigma-1}$ and in $w''$ this is preceded by $w'$ with both using exactly the same edges by invariant (*). This ensures that exactly the same edges are first-visited in $p_\sigma s_{\sigma+1} \cdots s_\xi$ and $\det_\sigma s_{\sigma+1} \cdots s_\xi$.
But we know that, in $w''$, no edge in $p_\sigma$ is first-visited, because all these edges occur in $w'$. By contrast, in $w$, there is at least one edge in $\det_\sigma$
which is first-visited.
Let $E''$ be the non-empty set of such edges.
In $w''$, all edges of $E''$ are still visited (because $w''$ visits all edges of~$E'$ as we explained in the previous paragraph),
but they are first-visited in $s_{\sigma+1} \cdots s_\xi$, i.e., each such edge $e \in E''$ is visited closer to the end in $w''$ than in~$w$. The remaining first-visited edges in the suffix $p_\sigma s_{\sigma+1} \cdots s_\xi$ of $w''$ are first-visited at the same position (from the end) as $w$. 
By definition of the timestamp ordering, this implies that we do not have $w \trianglelefteq w''$, which contradicts the timestamp-minimality of~$w$. Hence, this shows by contradiction that each subgroup in sequence $\mathcal{S}$ must be strictly contained in the next subgroup in~$\mathcal{S}$. 

\medskip
\textbf{Step 3: Concluding the proof.}
We have now shown that, in our timestamp-minimum walk $w$, each $S(w,\sigma-1)$ is a proper subgroup of $S(w,\sigma)$. By Lagrange's theorem, this implies that 
$S(w, \sigma-1)$ contains at most half of the elements of $S(w, \sigma)$ for all $\sigma<\xi$. (This reasoning does not apply to the last segment, because it may be the case that $s_\xi$ is not useful to achieve the right remainder and is only here to reach $t$.)
Hence, after $\log_2 q$ segments, we must have $S(w,\log_2 q)=\Z{q}$, which means that there are at most $1 + \log_2 q$ segments (to account for the last segment).
This establishes the claim and concludes the proof.
\end{proof}

\section{Number of Segments and Cutwidth Bounds}\label{sec:walk}
In this section, we continue our proof of \cref{th:main-theorem} by showing that, for any walk $w$, the number of segments in the decomposition of the previous section gives a bound on the cutwidth of~$G_w$ up to a constant factor. This result is completely independent from the definition of the \ewm problem. Formally, we show:

\begin{proposition}[Segment Cutwidth Bound]
\label{prp:segment-bound}
    For any walk $w$, letting $\xi$ be the number of segments in its segment decomposition, the cutwidth of $G_w$ is at most $3 \xi$.
\end{proposition}

We remark that there is no converse of this result: a walk formed of a succession of cycles connected by single edges will have an unbounded number of segments but has cutwidth 2. %
Combining \cref{prp:segment-bound} with
\cref{lem:bounding-number-of-segments} establishes the following:
\begin{corollary}
  \label{cor:combined}
For every $st$-walk $w'$ with length $r\bmod q$, there is another $st$-walk $w$
  of length $r \bmod q$ using
  only edges from $w'$ such that the cutwidth of~$G_w$ is at most $3+3\log_2 q$.
\end{corollary}

This implies the following bound on the cutwidth of optimal solutions to the
\ewm problem, on which our algorithm relies:
\begin{corollary}
  \label{cor:optbound}
  For any graph $G = (V, E)$, terminals $s$,$t$, and
non-negative integers $r<q$, every optimal solution to the \ewm problem on $G$,
  $s$, $t$, $r$, and $q$ has cutwidth at most $3+3\log_2 q$.
\end{corollary}

\begin{proof}
  Let $E' \subseteq E$ be any optimal solution, and consider any $st$-walk $w'$ of
  length $r \bmod q$ witnessing that $E'$ is a solution. By \cref{cor:combined},
  there is an $st$-walk $w$ using only edges from~$w'$ which achieves the same
  remainder and such that $G_w$ satisfies the cutwidth bound.
  As $E'$ is subset-minimum, we know that $w$ must use precisely the edges of
  $E'$, so that in fact $G_w = (V, E')$, so $E'$ obeys the cutwidth bound.
\end{proof}

To prove \cref{prp:segment-bound}, 
we need several intermediate steps. First, we define the notion of \emph{chunk} of a walk, which is a contiguous sequence of first-visited edges whose intermediate vertices are also first-visited.
Second, we define the ordering $\prec$ on the vertices of $G_w$ along which the
cutwidth bound will be shown: this order is defined using the notion of chunks,
intuitively because all vertices of a chunk will be ordered relative to the
already-visited vertices that are the endpoints of the chunk (also
distinguishing special cases like \emph{cycle chunks} and \emph{tadpole chunks}).
Third, we show that the cutwidth along the order $\prec$ is bounded by $3\xi$, by showing for each segment that the number of times its first-visited
edges can cross the cut is at most~$3$.
Throughout this section, we fix an arbitrary walk~$w$ in a graph~$G$, and call $s$ and $t$ its source and target vertices.

\myparagraph{Chunks.}
We define first-visited vertices similarly to first-visited edges:
for every vertex $v \in V_w$ occurring in~$w$, we say $v$ is \emph{first-visited
at $w[i]$} if $w[i]$ is the first edge in which~$v$ occurs. Note that a vertex
$v$ which is first-visited at $w[i]$ always occurs as the target vertex of~$w[i]$, except in the specific case of the first edge $w[1]$ where both the source $s$ and the target vertex are first-visited at $w[1]$.
Of course, every vertex of~$V_w$ is first-visited at exactly one position.
Further, if a vertex $v$ is the target vertex of an edge $w[i]$ and is first-visited at~$w[i]$, then both $w[i]$ and $w[i+1]$ (if it exists) must be first-visited edges.

We can now define \emph{chunks}:
\begin{definition}
    \label{def:chunk}
A \emph{chunk} of~$w$ is a maximal subwalk $w[i:j]$ of~$w$ where all edges are first-visited, and where, for every $i \leq k < j$ (if any), the target vertex of $w[k]$ is first-visited at~$w[k]$. 
\end{definition}
In other words, a chunk is a maximal sequence of one or more consecutive first-visited edges such that the first and last vertex are not first-visited (unless they are extremities of the whole walk) but all intermediate vertices are first-visited.
A chunk may consist of a single first-visited edge $w[i]$ between two
already-visited vertices, in which case there are no intermediate vertices. 

The \emph{first} and \emph{last} vertices of chunk $w[i:j]$ are the source vertex of~$w[i]$, and the target vertex of~$w[j]$, respectively.
Note that two successive chunks in the walk need not be separated by revisited
edges and can simply be separated by a revisited vertex: for instance, for the
length-$2$ walk $(s, s), (s, t)$, all edges are first-visited, but there are two chunks of length 1 which are separated by the revisit of the vertex $s$.

It will be useful to distinguish two special kinds of chunks. First, \emph{tadpole} chunks, which loop back on an intermediate vertex of the chunk:
\begin{definition}
    A chunk is a \emph{tadpole} if it consists of at least two edges and if,
    letting $u_1, \ldots, u_\ell$ be the successive vertices that it visits,
    with $u_1$ its first vertex and $u_\ell$ its last vertex, we have $u_\ell = u_{\ell'}$ for some $1 < \ell' < \ell$.
\end{definition}
In other words, a \emph{tadpole} is a chunk that consists of a path (containing at least one edge), followed by a cycle (possibly a self-loop).

Second, \emph{cycle} chunks, which loop back on the vertex from which they started:
\begin{definition}
    A chunk is a \emph{cycle} if its first vertex and last vertex are identical.
\end{definition}

Note that these two cases are mutually exclusive. A chunk which is neither a tadpole nor a cycle, and thus is a simple path, is a \emph{normal chunk}. 

The notion of chunk must be distinguished from the notion of segments used in the segment decomposition of the previous section, but we can notice the following connections between the two notions:
\begin{itemize}
    \item Segment ends never happen within a chunk: indeed, the intermediate vertices $y$ reached within a chunk are first-visited by definition, so there is no way to reach them except by the edge that precedes them, and thus no earlier path can reach $y$ in the walk.
    
    \item At the end of a tadpole chunk or cycle chunk, the current segment always ends. Indeed, the last edge $e'= (x,y)$ of the chunk reaches a vertex $y$ which already has an outgoing edge in the chunk: this edge is a first-visited edge $e=(y,v)$, and the empty path from $y$ to itself allows us to finish the segment at the end of the chunk, i.e., just after~$e'$.
\end{itemize}

In summary, chunks are contiguous subsequences of the walk that never straddle segment boundaries, and the end of a tadpole chunk or cycle chunk always triggers the end of a segment. 
Also note that, by definition, chunks form a partition of the first-visited edges.
This implies that, except possibly for the last segment, every segment must
contain at least one chunk, because they contain at least one first-visited
edge.

\begin{example}
We refer back to \cref{fig:walk} on page~\pageref{fig:walk}, and describe the chunks of the segments exemplified there:
\begin{itemize}
    \item \textcolor{csegm1}{Segment 1} consists of a single tadpole chunk.
    \item \textcolor{csegm2}{Segment 2} consists of a single cycle chunk.
    \item \textcolor{csegm3}{Segment 3} consists of a single tadpole chunk.
    \item \textcolor{csegm4}{Segment 4} consists of one revisited edge followed by a single normal
      chunk and the segment ends at the end of that chunk.
    \item \textcolor{csegm5}{Segment 5} alternates between revisited edges and normal chunks, before finishing by a tadpole chunk (which starts by the edge $(v_2, v_{20})$).
    \item \textcolor{csegm6}{Segment 6} revisits one edge, then does a single-edge normal chunk,
      then revisits edge $(v_{21}, v_{22})$ which ends the segment. Note how this illustrates
      that segment ends do not necessarily occur at the end of chunks.
    \item \textcolor{csegm7}{Segment 7} finishes by a normal chunk. Note that the last vertex of
      this chunk is first-visited (this is only possible for the last segment).
\end{itemize}
\end{example}

\myparagraph{Defining the ordering.}
Having defined the notion of chunks of the walk $w$, we now define the ordering
$\prec$ along which we will show that the cutwidth is bounded. This definition
only depends on chunks; it does not depend on the segment decomposition.
We see the total order $\prec$ as a sequence of vertices. 
Initially, the order is the empty order on the single vertex~$s$.
Then, for every chunk $w[i:j]$ successively, we consider the (possibly
empty) sequence $\sigma$ of the intermediate vertices of~$w[i:j]$.
Recall that the first vertex of $w[i:j]$ is already in the domain of~$\prec$: either it is $s$ (for the first chunk), or it is a vertex which is already visited.
Then there are four cases:
\begin{itemize}
\item Tadpole: If $w[i:j]$ is a tadpole, then we insert all its intermediate
  vertices at the end of the current ordering $\prec$, in the order in which
    they were first-visited.
\item Cycle: Otherwise, if $w[i:j]$ is a cycle, then, letting $v$ be its first and last vertex, we know that $v$ already occurs in~$\prec$, and we insert all its intermediate vertices right after the vertex~$v$, in the order in which they were visited.
\item Normal: We consider two subcases: 
\begin{itemize} \item First, suppose 
that the last vertex $t$ of $w[i:j]$ is first-visited at $w[j]$. This is only
possible with $j=|w|$, so that $w[i:j]$ is the last chunk. Then we do the same
    as in the tadpole case: insert the intermediate vertices and $t$ at the end
    of the current ordering~$\prec$, in the order in which they were
    first-visited.
\item Otherwise, $w[i:j]$ is a chunk whose last vertex $v$ is first-visited at
  $w[k]$ with $k<i$ so $v$ already occurs in~$\prec$, and as we explained the
    first vertex $u$ of the chunk also already occurs in~$\prec$. Further, we
    have $v \neq u$ because the chunk is not a cycle.
Then, we insert the intermediate vertices so that the vertices along the chunk are ordered in a monotone fashion in the ordering. In other words, 
let $x_1, \ldots, x_\ell$ be the successive intermediate vertices of the chunk, so that its edges are $(u, x_1), (x_1, x_2), \ldots, (x_{\ell-1}, x_\ell), (x_\ell, v)$.
If $u \prec v$ then we insert $x_1, \ldots, x_\ell$ in order between $u$ and
    $v$, and if $v \prec u$ we insert $x_\ell, \ldots, x_1$ in order between $v$
    and $u$. The order $\prec$ between the newly inserted elements and the
    existing elements between $u$ and $v$ is arbitrary, for instance we can arbitrarily say that we insert the new vertices just after the smallest of $u$ and $v$ in $\prec$.
\end{itemize}
\end{itemize}

Note that, in all cases above, the (intermediate) vertices of a chunk are always ordered as a monotone sequence:

\begin{remark}
  \label{rem:monotone}
  Consider the sequence $u, x_1, \dots, x_\ell, v$ formed by the first vertex of the chunk $u$, the intermediate vertices $x_1, \ldots, x_\ell$ in the order where they are first-visited, and the last vertex $v$ of the chunk. We then have 
  $u \prec x_1 \prec \cdots \prec x_\ell \prec v$ except that:
\begin{itemize}
  \item In the second subcase of Normal chunks we may have instead 
$u \succ x_1 \succ \cdots \succ x_\ell \succ v$;
\item For cycle chunks and tadpole chunks, $v$ already occurs earlier in
  the chunk, so we can only write: $u \prec x_1 \prec \cdots \prec x_\ell$.
\end{itemize}
\end{remark}

\begin{example}
We refer again to \cref{fig:walk} and describe the order $\prec$ after each segment:
\begin{itemize}
    \item The vertices visited by the first segment are ordered in the order of
      their first visit, i.e., $v_1$ then $v_2$ then $v_3$.
    \item For the second segment, vertices $v_4$ and $v_5$ are inserted right after $v_2$ (and therefore before $v_3$).
    \item For the third segment, first-visited vertices are again ordered in
      the order of their first visit, and are inserted after those of the first
      two segments, i.e., after $v_3$.
    \item Vertex $v_{18}$ is inserted between $v_{15}$ and $v_{16}$ in the
      order.
    \item Vertex $v_{19}$ is inserted between $v_6$ and $v_7$ in the order.
    \item Vertices $v_{20}$ to $v_{23}$ are inserted in that order at the end of
      the order, i.e., after $v_{17}$.
    \item The last segment inserts $v_{24}$ at the end of the order.
\end{itemize}
Overall, one possible ordering of all vertices of the walk which satisfies the
  conditions above is the following: 
  \[v_1, v_2, v_4, v_5, v_3, v_6, v_{19}, v_7, v_8, v_9, v_{10}, v_{11},
  v_{12}, v_{13}, v_{14}, v_{15}, v_{18}, v_{16}, v_ {17}, v_{20}, v_{21},
  v_{22}, v_{23}, v_{24}.\]
\end{example}

\myparagraph{SCCs of the segment decomposition.}
We have defined, from our walk $w$, the order $\prec$ along which we will bound the cutwidth.
To show the cutwidth bound, we will now study
the SCC decompositions of the graphs spanned by prefixes of~$w$. 

Let us consider the successive decompositions of $G_{w,0}, \ldots, G_{w,|w|}$
into SCCs. Graphs generated by a walk have very simple SCC decompositions:

\begin{lemma}
\label{lem:sccstrut}
At each step $0 \leq i \leq |w|$, 
the successive SCCs $C^i_1, \ldots, C^i_{\kappa_i}$ of~$G_{w,i}$ 
are linearly ordered by the reachability relationship 
(i.e., every vertex in $C^i_b$ is reachable from every vertex in $C^i_a$ iff $a < b$),
and the target vertex of the last edge $w[i]$ is in the last SCC $C^i_{\kappa_i}$ of~$G_{w,i}$.
\end{lemma}

\begin{proof}
We show this claim by induction on the position $i$ in the walk $w$.
Let $0\leq i\leq |w|$, let $v_i$ be the target vertex of $w[i]$, and let us show the claim.
  Remember that $G_{w,0}$ is by convention the graph with only the vertex~$s$ 
  (in a singleton trivial SCC) and no
  edges, so the lemma statement is immediately true.
 Now, for each $i \geq 1$, exactly one of the following three cases happens (pictured in order in \cref{fig:sccstruct}):
\begin{itemize}
    \item The edge $w[i]$ is not first-visited. In this case $G_{w,i} =
      G_{w,i-1}$, the SCCs are also unchanged, and the walk $w[:i-1]$ finished
      at a vertex in the last SCC $C^{i-1}_{\kappa_{i-1}}$ by induction
      hypothesis. The walk remains in that SCC by definition of it being the last SCC; formally $C^{i-1}_{\kappa_{i-1}}$
    is also the SCC of $v_{i}$ 
    in~$G_{w,i}$ and the invariant is true.
    \item 
    The edge $w[i]$ is first-visited, and its target $v_i$ is also first-visited at~$w[i]$. In this case
    the SCC decomposition of $G_{w,i}$
    is the same as that of $G_{w,i-1}$
    except we have a new vertex $v_i$ in a new trivial SCC with a single edge $(v_{i-1}, v_i)$, this edge comes from the SCC $C^{i-1}_{\kappa_{i-1}}$ of $G_{w,i-1}$ by the inductive hypothesis
    which as we explained is the same as the SCC $C^i_{\kappa_i-1}$:
    note that $\kappa_i = \kappa_{i-1}+1$.
    The fact that the SCCs of $G_i$ is linearly ordered by the reachability
    relationship follows from the inductive hypothesis.
    \item The edge $w[i]$ is first-visited, but its target $v_i$ is not first-visited at~$w[i]$. In this case, $G_{w,i}$ differs from $G_{w,i-1}$ by the addition of the edge $w[i]$ going from $v_{i-1}$ to~$v_i$. Starting from the SCC decomposition $C^{i-1}_1, \ldots, C^{i-1}_{\kappa_{i-1}}$ of $G_{w,i-1}$ which satisfies the induction hypothesis, and letting $C^{i-1}_j$ 
    with $1 \leq j \leq \kappa_{i-1}$
    be the SCC of $v_i$ in~$G_{w,i-1}$, we see that the addition of the new edge will merge all SCCs $C^{i-1}_j, \ldots, C^{i-1}_{\kappa_{i-1}}$ into one SCC which is the last SCC in the decomposition of~$G_{w,i}$. 
    (Note that if $j = \kappa_{i-1}$ then the merge is trivial.)
    The fact that the SCCs of $G_i$ is linearly ordered by the reachability relationship is again by the inductive hypothesis.\qedhere
\end{itemize}
\end{proof}

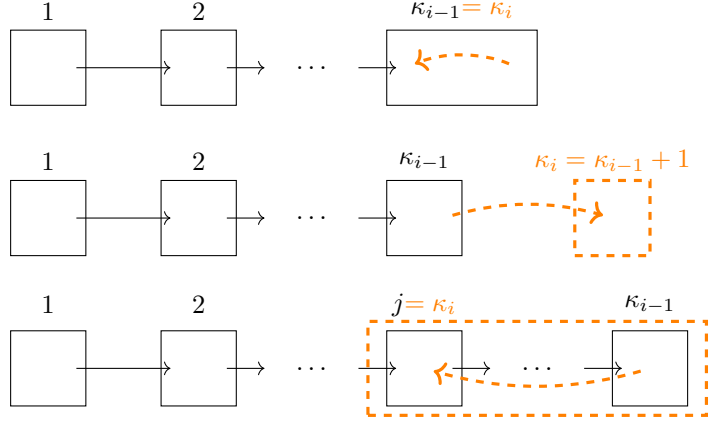
\begin{figure}[h!]
\centering
    \begin{tikzpicture}[xscale=1,yscale=.5,
        snew/.style={very thick, orange,dashed}
        ]
    \begin{scope}
        \node (al1) at (.5, 1.5) {1};
        \node (al2) at (2.5, 1.5) {2};
        \node (alk) at (6, 1.5) {$\kappa_{i-1} \textcolor{orange}{= \kappa_i}$};
        \draw (0, -1) rectangle (1, 1);
        \draw (2, -1) rectangle (3, 1);
        \draw (5, -1) rectangle (7, 1);
        \node (dots) at (4,0) {$\ldots$};
        \node (a1) at (.75, 0) {};
        \node (b1) at (2.25, 0) {};
        \node (a2) at (2.75, 0) {};
        \node (b2) at (3.5, 0) {};
        \node (ak) at (4.5, 0) {};
        \node (bk) at (5.25, 0) {};
        \draw[->] (a1) -- (b1);
        \draw[->] (a2) -- (b2);
        \draw[->] (ak) -- (bk);
\node (akk) at (6.75, 0) {};
        \draw (akk) edge [snew,->,bend right=40] (bk);
        
    \end{scope}

    \begin{scope}[yshift=-4cm]
        \node (al1) at (.5, 1.5) {1};
        \node (al2) at (2.5, 1.5) {2};
        \node (alk) at (5.5, 1.5) {$\kappa_{i-1}$};
        \node (alkk) at (8, 1.5) {$ \textcolor{orange}{\kappa_i
        = \kappa_{i-1}+1}$};
        \draw (0, -1) rectangle (1, 1);
        \draw (2, -1) rectangle (3, 1);
        \draw (5, -1) rectangle (6, 1);
        \draw[snew] (7.5, -1) rectangle (8.5, 1);
        \node (dots) at (4,0) {$\ldots$};
        \node (a1) at (.75, 0) {};
        \node (b1) at (2.25, 0) {};
        \node (a2) at (2.75, 0) {};
        \node (b2) at (3.5, 0) {};
        \node (ak) at (4.5, 0) {};
        \node (bk) at (5.25, 0) {};
        \node (akk) at (5.75, 0) {};
        \node (bkk) at (8, 0) {};
        \draw[->] (a1) -- (b1);
        \draw[->] (a2) -- (b2);
        \draw[->] (ak) -- (bk);
        \draw (akk) edge[->,snew,bend left=30] (bkk);
    \end{scope}

    \begin{scope}[yshift=-8cm]
        \node (al1) at (.5, 1.7) {1};
        \node (al2) at (2.5, 1.7) {2};
        \node (alj) at (5.5, 1.7) {$j \textcolor{orange}{
        = \kappa_i}$};
        \node (alk) at (8.5, 1.7) {$\kappa_{i-1}$};
        \draw[snew] (4.75, -1.25) rectangle (9.25, 1.25);
        \draw (0, -1) rectangle (1, 1);
        \draw (2, -1) rectangle (3, 1);
        \node (dots) at (4,0) {$\ldots$};
        \draw (5, -1) rectangle (6, 1);
        \node (dots2) at (7,0) {$\ldots$};
        \draw (8, -1) rectangle (9, 1);
        \node (a1) at (.75, 0) {};
        \node (b1) at (2.25, 0) {};
        \node (a2) at (2.75, 0) {};
        \node (b2) at (3.5, 0) {};
        \node (aj) at (4.5, 0) {};
        \node (bj) at (5.25, 0) {};
        \node (ajj) at (5.75, 0) {};
        \node (bjj) at (6.5, 0) {};
        \node (ak) at (7.5, 0) {};
        \node (bk) at (8.25, 0) {};
        \node (akk) at (8.5, 0) {};
        \node (bkk) at (5.5, 0) {};
        \draw[->] (a1) -- (b1);
        \draw[->] (a2) -- (b2);
        \draw[->] (aj) -- (bj);
        \draw[->] (ajj) -- (bjj);
        \draw[->] (ak) -- (bk);
        \draw (akk) edge[->,snew,bend left=30] (bkk);
    \end{scope}
    \end{tikzpicture}

    \caption{Illustrating the three cases of the proof of \cref{lem:sccstrut}
    to show that the successive SCCs of a walk are always linearly ordered by the
    reachability relation.
    Squares denote strongly connected components, labeled with their SCC number in the decomposition of $G_{w,i-1}$. The new edge is dashed in orange, and new SCCs are also dashed in orange.}
\label{fig:sccstruct}
\end{figure}

To connect the proof above with the chunk decomposition, notice that, when we are processing the intermediate vertices of a chunk, it is always the second case that applies.
The third case only applies when a chunk ends, either because it is a cycle or a
tadpole or because it revisits a vertex that was first-visited in a previous
chunk. The first case applies during the revisits between chunks.

To connect the proof with the ordering that we define on vertices, notice that we can inductively show the following claim: for any $1 \leq i \leq |w|$, each SCC of $G_{w,i}$ is a contiguous group of vertices in the ordering $\prec$ defined at step~$i$, and these groups are ordered along the (linear) topological order on the SCCs. Formally:

\begin{lemma}\label{lem:ordering-consistent-with-topological-order-scc}
For every prefix $w[:i]$ of the walk such that a chunk of $w$ ends at $w[i]$, the ordering $\prec$ induced by $w[:i]$  is consistent with the topological order of the SCCs of $w[:i]$. 
\end{lemma}

\begin{proof}
We show the claim by induction on successive chunks. Remember that the first vertex of a chunk $w[j:j']$ always belongs to the last 
SCC of $G_{w,j}$ in the topological order by \cref{lem:sccstrut}. 
We consider the possible cases of chunks:

\begin{itemize}
\item Tadpole: The addition of a tadpole chunk creates 
a sequence of trivial SCCs (on the initial path of the tadpole) followed by a
    new non-trivial SCC at the end of the topological order (corresponding to
    the cycle of the tadpole), and the vertices of these SCC are placed in
    first-visited order at the end of the ordering $\prec$. Thus, the ordering $\prec$ remains consistent with the topological order of SCCs.

\item Cycle: The addition of a cycle chunk grows the last SCC, and the new
  vertices are added in $\prec$ just after the vertex $v$ to which the cycle
    is attached. Thus, the ordering $\prec$ remains consistent with the topological order of SCCs.

\item Normal: We consider the same two subcases as in the definition of the
  ordering:
\begin{itemize}
    \item The addition of a chunk that ends the entire walk $w$ at a first-visited vertex adds new trivial SCCs to the end of the topological order, while adding these new vertices in the same order to the end of $\prec$. Thus, the ordering $\prec$ remains consistent with the topological order of SCCs.
    \item Otherwise, the addition of a chunk
    merges some suffix of the SCCs (possibly the merge is trivial if the chunk is re-entering the last SCC). 
    Further, the chunk adds new vertices to the resulting last SCC (possibly none if there are no intermediate vertices). 
It is clear that the ordering $\prec$ remains consistent with the topological
    order of the SCCs for all of the pre-existing vertices (not in the current
    chunk), because a suffix of the SCCs (in topological order and in the
    order of $\prec$ by induction hypothesis) get merged into the same SCC. As
    for the vertices added in the current chunk (if any), they are inserted in
    $\prec$ between two vertices which both belong to the SCC resulting
    from the merge, so the order $\prec$ on these vertices is also consistent
    with the topological order of the SCCs. \qedhere
\end{itemize}
\end{itemize}
\end{proof}

\myparagraph{Bounding the cutwidth from the number of segments.}
We can now turn back to the segment decomposition introduced in the previous
section for the walk~$w$, and show how the number of segments of~$w$ can be used as a bound on the cutwidth of the graph $G_w$ spanned by~$w$, following the order $\prec$ that we defined.
For this, we will consider an arbitrary cut $V_- \uplus V_+$ of $V_w$ following $\prec$,
and count how many edges of~$w$ cross the cut. 
As each edge is first-visited once, and first-visited in exactly one segment, it suffices to bound how many edges each segment contributes to the cut, and to count only the first-visited edges of each segment. We want to show:

\begin{lemma}\label{lem:segbound}
    For each segment $w[i:j]$ of the walk $w$, the number of first-visited edges in $w[i:j]$ that cross the cut is at most 3.
\end{lemma}

This immediately implies the Segment Cutwidth Bound (\cref{prp:segment-bound}) stated at the beginning of the section. We show \cref{lem:segbound} in the rest of the section.

We will study segments chunk by chunk. The intuition is that a chunk usually
crosses the cut at most once, and that only constantly many chunks can do so.
However, we must first eliminate the special case of cycle chunks and tadpole
chunks, which may sometimes cross the cut twice. First, let us state a lemma
describing what may happen with cycle chunks:

\begin{lemma}\label{lem:cycle-twice}
  Let $w[i:j]$ be a segment of $w$ and assume that the last chunk of $w[i:j]$ is a cycle $w[k:j]$ (any cycle chunk must end its segment, as previously observed).
  Then $w[k:j]$ crosses the cut at most twice. Further, if it does cross the cut
  twice then it must be the case that the first and last vertex of $w[k:j]$ is
  in $V_-$ and that $w[k:j]$ contains some intermediate vertex in $V_+$.
\end{lemma}

\begin{proof}
    We refer back to the definition of the ordering $\prec$ on cycle chunks (case Cycle). Let $v_1, \ldots, v_\ell$ be the vertices traversed by the cycle chunk, where all vertices except $v_1=v_\ell$ are first-visited. 
    We have modified $\prec$ after the chunk by inserting $v_2, \ldots, v_{\ell-1}$ in that order right after~$v_1$. Hence, if $v_1$ is in $V_+$ then so are all other vertices and there is no crossing of the cut. Likewise, if all vertices are in $V_-$ then there is no crossing of the cut. The remaining case is that $v_1$ is in $V_-$ and at least one intermediate vertex is in $V_+$,
    in which case the chunk crosses the cut exactly twice: once forwards when reaching the first intermediate vertex in~$V_+$, and once again backwards when going back to~$v_1$.
\end{proof}

Second, let us state a lemma about what may happen with tadpole chunks:

\begin{lemma}\label{lem:tadpole-twice}
  Let $w[i:j]$ be a segment of $w$ and assume that the last chunk of $w[i:j]$ is a tadpole $w[k:j]$ (any tadpole chunk must end its segment, as previously observed). 
  Then $w[k:j]$ crosses the cut at most twice. Further, if it crosses the cut
  twice then it has an intermediate vertex in~$V_-$ and an intermediate vertex
  in~$V_+$.
\end{lemma}

\begin{proof}
    We refer back to the definition of $\prec$ on tadpole chunks (case Tadpole).
Let $v_1, \ldots, v_\ell$ be the vertices traversed, with $v_\ell = v_{\ell'}$
  for some $1 < \ell' < \ell$. We have modified $\prec$ after the chunk by inserting $v_2,
  \ldots, v_{\ell-1}$ at the end of~$\prec$, in that order. Hence, there are 3 cases:
\begin{itemize}
    \item All vertices are in $V_+$, or all vertices are in $V_-$, and the cut is not crossed
    \item There is $\ell'' < \ell'$ such that vertices $v_1, \ldots, v_{\ell''}$ are in $V_-$ and the others are in $V_+$, then the cut is crossed once by the edge $(v_{\ell''}, v_{\ell''+1})$
    \item There is $\ell'' \geq \ell' > 1$ such that vertices $v_1, \ldots, v_{\ell''}$ are in $V_-$ and the others are in $V_+$, then the cut is crossed twice: by the edge $(v_{\ell''}, v_{\ell''+1})$, and by the edge $(v_{\ell-1}, v_{\ell})$.
\end{itemize}
In all cases the conditions of the lemma statement are respected.
\end{proof}

In the other cases, chunks only cross the cut at most once:

\begin{lemma}\label{lem:normal-once}
  Let $w[i:j]$ be a normal chunk of~$w$, then it crosses the cut at most once.
\end{lemma}

\begin{proof}
    By case Normal in the definition of $\prec$, the vertices of the chunk form
    a monotone sequence in $\prec$
    (see \cref{rem:monotone}).
    Further, in that case, the edges of the chunk always go from one vertex
    to the next vertex in the order. Hence, the edges of the chunk can cross the
    cut at most once.
\end{proof}

What is more, whenever a normal chunk crosses the cut forwards, then the segment ends:

\begin{lemma}\label{lem:normal-forwards}
  Let $w[i:j]$ be a normal chunk of~$w$ which crosses the cut forwards, i.e.,
  the starting vertex of the chunk is in~$V_-$ and its ending vertex is
  in~$V_+$. Then the segment containing $w[i:j]$ ends at~$j$.
\end{lemma}

\begin{proof}
    Consider the graph $G_{w,i-1}$ spanned by the walk right before $w[i:j]$. 
    As $w[i:j]$ is a chunk, it starts with a first-visited edge $(u,v)$. 
    By case Normal in the definition of $\prec$, for a normal chunk to cross the cut forwards, it must be the case that its first vertex $u$ and its last vertex
    $y$ are such that $u \in V_-$ and $y \in V_+$: in particular we have $u
    \prec y$. We know that $y$ is already-visited when the chunk $w[i:j]$ reaches it and ends, so $y$ is in $G_{w,i-1}$. Using
    \cref{lem:sccstrut}
    we know that the SCCs $C^{i-1}_1, \ldots, C^{i-1}_{\kappa_{i-1}}$ 
    of~$G_{w,i-1}$
    are linearly ordered and $u\in C^{i-1}_{\kappa_{i-1}}$, and thus $y\in C^{i-1}_{\kappa_{i-1}}$ since $u \prec y$. 
    But then by definition of $C^{i-1}_{\kappa_{i-1}}$ being an SCC, there is a path from $u$ to $y$ in $G_{w,i-1}$. We also know that the edge $w[i]=(u,v)$ is first-visited, and $w[j]$ finishes at~$y$. These are precisely the criteria for ending at $j$ the segment that contains $w[i:j]$.
\end{proof}

All that remains is to bound the contribution to the cut of the normal chunks that cross the cut backwards. To this end, let us distinguish the last chunk of a segment which we call the \emph{final} chunk, and the remaining chunks in that segment which we call \emph{non-final} chunks. 
Non-final chunks must be normal (because the segment ends right after cycle chunks or tadpole chunks), and they cannot cross the cut forward by the previous lemma.

We are ready to conclude the proof of \cref{lem:segbound}:

\begin{proof}[Proof of \cref{lem:segbound}]
    Our goal is to show that for each segment $w[i:j]$ of the walk $w$, the number of first-visited edges in $w[i:j]$ that cross the cut is at most 3. 
    Fix a segment $w[i:j]$.
    If $w[i:j]$ contains no chunks, which is possible only for the last segment,
    then there is nothing to show. Hence, we assume that $w[i:j]$ has
    at least one chunk, and number the chunks $c_1, \ldots, c_{\ell+1}$ with
    $\ell+1 \geq 1$ being the number of chunks. We distinguish the non-final
  chunks $c_1, \ldots, c_\ell$ from the final chunk $c_{\ell+1}$.

    Recall that non-final chunks must be normal chunks.
    By \cref{lem:normal-forwards} we know that non-final chunks cannot cross the
    cut forwards. So it suffices to bound the number of times that non-final
    chunks cross the cut backwards, and the number of times that the final chunk
    crosses the cut. 
    Let us consider the SCC decomposition of $G_{w,i-1}$ before the segment
    $w[i:j]$ started, and recall by 
    \cref{lem:sccstrut}
    that the SCCs $C^{i-1}_1, \ldots, C^{i-1}_{\kappa_{i-1}}$ 
    of~$G_{w,i-1}$ are linearly ordered: for simplicity we write the SCCs
    as $C_1, \ldots, C_\kappa$, dropping the dependency on~$i$.
   By 
    \cref{lem:sccstrut} we also know that the segment $w[i:j]$ starts in $C_\kappa$.
    Considering the cut $V_- \uplus V_+$ on the vertices of $G_{w,i-1}$, we call
    a set of vertices \emph{left} if all its vertices are in~$V_-$, \emph{right} if all its
    vertices are in~$V_+$, and \emph{middle} if it contains some vertices
    of~$V_-$ and some vertices of~$V_+$.
    By \cref{lem:ordering-consistent-with-topological-order-scc}, there is
    an index $1 \leq \theta \leq \kappa+1$ so that all SCCs of $C_1, \ldots,
    C_{\theta-1}$ are left, $C_\theta$ is middle or right, and all
    SCCs of $C_{\theta+1}, \ldots,
    C_\kappa$ are right.

  The overall intuition for the rest of the proof is that we will distinguish three cases.
  If all chunks but the final one stay on the left of the cut (Case
  1), or on the right of the cut (Case 2), then only the final chunk can cross
  the cut and the preceding lemmas (\cref{lem:cycle-twice},
  \cref{lem:tadpole-twice} and \cref{lem:normal-once})
  show that the bound is satisfied. Otherwise (Case 3), there is a first
  non-final chunk which ends on some SCC $C_\phi$ that contains nodes to the
  left of the cut and this chunk thus potentially crosses the cut (backwards,
  by \cref{lem:normal-forwards}). If the next chunk is final, the bound is
  satisfied, 
  and if the next chunk is non-final we show that it is the last non-final chunk
  which may cross the cut.
  This implies that the cut is crossed at most 4 times in total: once for each
  of the two non-final chunks considered, and potentially twice for the final
  chunk. We will lower the bound to 3 by showing that in fact the final chunk
  can only cross the cut once if it is preceded by two chunks that already
  crossed the cut.

    \medskip
    \textbf{Case 1}: $\theta = \kappa+1$. In that case, all SCCs $C_1, \ldots, C_\kappa$
    are left.
    We reason by induction and show the following two claims 
    by induction on $1 \leq \chi \leq \ell+1$:
    (1.) all vertices visited until~$c_\chi$ starts (including its starting vertex) 
    are in~$V_-$; and (2.) the
    cut is never crossed by any chunk of~$w[i:j]$ until $c_\chi$ excluded. For the base
    case, notice that (1.) all vertices visited before $c_1$ starts are
    in~$V_-$; and (2.) the cut is not yet crossed then because all edges
    are revisits. For the induction case, considering any $1 < \chi \leq
    \ell+1$, we know that $c_{\chi-1}$ is
    normal, and by induction hypothesis point (1.) we know that all
    vertices visited until $c_{\chi-1}$ starts are in~$V_-$. In particular
    $c_{\chi-1}$  starts at a vertex of $V_-$ and ends
    (because it is non-final) by revisiting an
    already-visited vertex, which is also in $V_-$.
    Thus, by definition of $\prec$ in case Normal (subcase~2),
    we know that
    all intermediate vertices of~$c_{\chi-1}$ are also in $V_-$, which,
    together with induction hypothesis point (1.), establishes (1.).
    Further, $c_{\chi-1}$ does not cross the cut, so we conclude (2.) by induction
    hypothesis point~(2.). This concludes the
    induction.  So, only the last chunk can cross the cut and 
    \cref{lem:cycle-twice}, \cref{lem:tadpole-twice} and \cref{lem:normal-once} conclude.
    Thus, we can assume that $\theta\leq\kappa$,
    so that $C_\kappa$ is middle or right.
    
    \medskip
    \textbf{Case 2}: 
    the vertices visited by $w[i:j]$ up to the starting vertex of the final chunk $c_{\ell+1}$ (included) do not belong to $C_1 \uplus \dots \uplus C_\theta$.
    In that case, given that the segment starts in~$C_\kappa$, we know that we must have $\kappa >\theta$.
    Further, in this case, all vertices visited before
    $c_{\ell+1}$ in $w[i:j]$ are in~$V_+$.
    Indeed, either these vertices are in
    $C_{\theta+1}, \ldots, C_{\kappa}$
    or they are intermediate
    vertices of non-final chunks with endpoints in $V_+$ and so are in $V_+$, by an inductive reasoning which is similar to Case~1.
    Hence, in this case only
    the final chunk could cross the cut 
    and we conclude like in the previous
    case.

    \medskip
    \textbf{Case 3}: the two previous cases do not apply. We focus on Case~3 in the rest of the proof.
    We then know that
    $C_\kappa$ is middle or right (otherwise Case~1 would have applied) and
    that at least one vertex of $w[i:j]$ up to the first vertex (included) of its final chunk $c_{\ell+1}$ belongs to $C_1 \uplus \dots \uplus C_\theta$ (otherwise Case~2 would have applied).
    Let $u$ be the first vertex of the segment $w[i:j]$ that belongs to $C_1 \uplus \dots \uplus C_\theta$, let $C_\phi$ ($1\leq \phi \leq \theta$) be the SCC of~$u$, and let $h$ be the index of the first edge that contains $u$ in the segment (smallest $i \leq h \leq j$ such that $w[h]$ has an extremity on $u$).
    There are two subcases: either $\theta=\kappa$ (Subcase~1), in which case
    $h=i$ and $u$ is the source of~$w[h]$; or $\theta<\kappa$ (Subcase~2), in which case
    $h\geq i$, and $u$ will necessarily be the target of~$w[h]$.
    In Subcase~2, any chunk whose last edge is in $w[i:h-1]$ must be non-final (because the
    final chunk is in $w[h:j]$), and all vertices visited by such chunks must belong to~$V_+$ by a reasoning similar to Case~2. 
    Further, in Subcase~2, there may be only one chunk whose first edge is in $w[i:h]$ and whose last edge is in $w[h:j]$ (note that it ends with $w[h]$ because in that case $w[h]$ ends on the already-visited vertex $u$): such a chunk is non-final, so it is normal, and it can cross the cut at most once by \cref{lem:normal-once}.
    In \cref{fig:scc-chunks} we represent Subcase~2 of Case $3$, where $\theta <
    \kappa$: the chunk whose first edge is in $w[i:h]$ and whose last edge is in $w[h:j]$ is
    called  $c_{\chi-1}$, and starts from an SCC $C_\lambda$ with $\theta<
    \lambda \leq \kappa$.\\
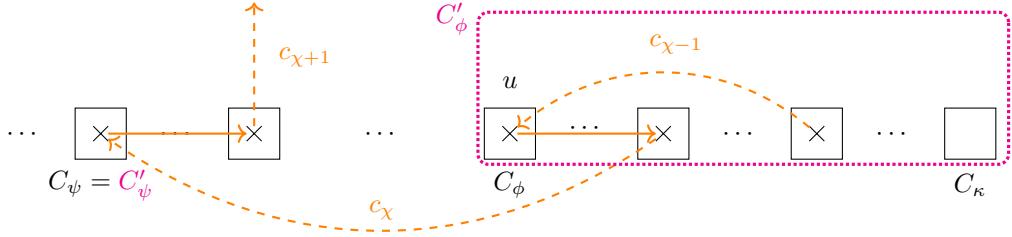
\begin{figure}[h!]
\centering
    \begin{tikzpicture}[scale=.34]
        \node (cx) at (30.5,3.5) {$\textcolor{orange}{c_{\chi-1}}$};
        \node (cxx) at (19,-3) {$\textcolor{orange}{c_{\chi}}$};
        \node (cxxx) at (16,3) {$\textcolor{orange}{c_{\chi+1}}$};
        \node (jx) at (24,2) {$u$};
  \node(dots0) at (5,0) {$\cdots$};
        \draw (7,1) rectangle +(2,-2);
\node(dots0) at (11,0) {$\cdots$};
        \draw (13,1) rectangle +(2,-2);
\node(dots2) at (19,0) {$\cdots$};
        \draw (23,1) rectangle (25,-1);
\node(dots3) at (27,.2) {$\cdots$};
        \draw (29,1) rectangle (31,-1);
\node(dots4) at (33,0) {$\cdots$};
        \draw (35,1) rectangle (37,-1);
\node(dots5) at (39,0) {$\cdots$};
        \draw (41,1) rectangle (43,-1);
        \node (tchi) at (24,-2) {$C_\phi$};
        \node (schii) at (42,-2) {$C_\kappa$};
        \node (tchii) at (8,-2) {$C_\psi=\textcolor{magenta}{C_\psi'}$};
        \node (schiii) at (16,-2) {};
        \node[cross] (f) at (36, 0) {};
        \node[cross] (v1) at (24, 0) {};
        \node[cross] (v2) at (30, 0) {};
        \node[cross] (v3) at (8, 0) {};
        \node[cross] (v4) at (14, 0) {};
        \node (t) at (14, 5.5) {};

        \draw (f) edge[->,dashed,orange,thick,bend right=40] (v1);
        \draw (v1) edge[->,orange,thick] (v2);
        \draw (v2) edge[->,dashed,orange,thick,bend left=35] (v3);
        \draw (v3) edge[->,orange,thick] (v4);
        \draw (v4) edge[->,dashed,orange,thick] (t);

        \draw[rounded corners,magenta,densely dotted,very thick] (43.5,4.75) rectangle (22.75,-1.2);
        \node[magenta] at (21.75,4.5) {$C_\phi'$};

    \end{tikzpicture}
    \caption{Illustration of non-final normal chunks relative to the SCCs of $G_{w,i-1}$. The successive squares correspond to the SCCs $C_1^{i-1}, \ldots, C^{i-1}_{\kappa_{i-1}}$. The successive chunks are depicted in curved dashed orange lines, separated by revisited edges (straight solid orange lines). The dashed pink rectangle shows the last SCC $C'_\phi$ in the SCC decomposition of~$G_{w,h'-1}$. We number the chunks as in Case~3 of the proof of \cref{lem:segbound}. We only illustrate the starting SCC of chunk $c_{\chi+1}$ as this one may be the final chunk.
    }
    \label{fig:scc-chunks}
\end{figure}

    Let $h'$ be the index of the first edge which starts from $u$ in the segment
    $w[i:j]$: i.e., we set $h'=h$ if we were in Subcase~1 above, and $h'=h+1$ if
    we were in Subcase~2.
    Let $c_{\chi}$
    be the first chunk contained in $w[h':j]$. If $c_{\chi}$
    is the final chunk, then it crosses the cut at most twice (by
    \cref{lem:cycle-twice} and \cref{lem:tadpole-twice} and \cref{lem:normal-once}), and we
    can conclude that the segment $w[i:j]$ crosses the cut at most three times in total.
    Hence, in the rest of the proof, we only need to consider the case where $c_{\chi}$ is not the final chunk.

    Let us consider the SCC decomposition of $G_{w,h'-1}$,
    which we write $C_1', \ldots, C_{\kappa'}'$. We
    make the following claims, dubbed $(\dagger)$: we have $\kappa' = \phi$, we
    have $C_{\eta} = C'_{\eta}$ for all $1
    \leq \eta < \phi$,
    and we have that $C_{\phi}'$ contains all vertices of $C_{\phi}, \ldots,
    C_{\kappa}$ together with all intermediate vertices of the non-final chunks
    $c_1, \ldots, c_{\chi-1}$.

    Let us show $(\dagger)$. The subwalk $w[i:h'-1]$ (which is empty precisely in Subcase 1)
    goes from a vertex of $C_\kappa$
    to a vertex of $C_\phi$, so the SCCs $C_\phi, \ldots, C_\kappa$ and the
    intermediate vertices of non-final chunks $c_1, \ldots, c_{\chi-1}$ have all
    been merged into the SCC $C_\phi'$ (this has no effect in Subcase 1);
    and $w[i:h'-1]$
    did not revisit any vertex of $C_\eta$ for $\eta < \phi$ by minimality of~$h$.
    Thus, we have shown $(\dagger)$. 

    Let us now continue the proof, and remember that  the
    first chunk $c_{\chi}$ of $w[h':j]$ is non-final. 
    We claim that $c_{\chi}$ must conclude by
    re-visiting a vertex in an SCC $C_{\psi}'$ with $\psi < \phi$, as pictured
    in \cref{fig:scc-chunks}. Indeed, let
    us proceed by contradiction and assume that $c_{\chi}$ ends by reaching another 
    vertex $y$. Then that vertex must be already visited (because the chunk is non-final so it is not a tadpole), 
    so $y$ is a vertex of $C_\phi'$.
    Then, letting $e= (u,v)$ be the first edge of $c_{\chi}$, 
    we know that $e$ is first-visited. 
    Further, by definition of~$C_{\phi}'$ being an SCC, 
    we have a path $\pi$ from $u$ to~$y$ in~$C_{\phi}'$.
    Thus, $e$ and $\pi$ witness a segment end when $c_{\chi}$ ends. 
    This contradicts the assumption that $c_{\chi}$ is non-final.
    Hence, $c_{\chi}$ finishes by revisiting a vertex in an SCC
    $C_{\psi}'$ with $\psi < \phi$. By \cref{lem:normal-once} the chunk
    $c_{\chi}$ also crosses the cut at most once, and we proved above that the segment cannot cut the cross more than once in $w[i:h']$.
    We next show two last properties, which suffice to conclude: 
    (1.) no non-final chunk can cross the cut after $c_{\chi}$, 
    and (2.) the final chunk $c_{\ell+1}$ crosses the cut at most once.
    Note that (2.) is not necessary to get a constant bound on the number of times
    $w[i:j]$ crosses the cut, but we use the property to get down to the precise bound of 3 that we claimed.

    To show (1.), let us proceed by contradiction. Assume that there is a
    non-final chunk $c_{\chi'}$ with $\chi' > \chi$ which crosses the cut.
    The chunk $c_{\chi'}$ is non-final, so by \cref{lem:normal-forwards} it 
    crosses the cut backwards. Now, recall that after $c_{\chi}$ concludes, the walk $w$
    starts in some SCC
    $C_{\psi}' = C_{\psi}$
    which is a left SCC. Hence, to cross the cut
    backwards, the segment $w[i:j]$
    must reach a vertex 
    $y$ 
    of~$V_+$ after $c_{\chi}$ has concluded. 
    The vertex $y$ cannot be in $C_{\eta}'$ for any $\eta <
    \phi$, because for such values of $\eta$ we have $C_{\eta}' =
    C_{\eta}$ and such SCCs are left.
    By the
    definition of $\prec$ in case Normal (subcase~2),
    the only way for the segment $w[i:j]$ to reach a vertex $y$ of~$V_+$, is for
    the vertex $y$ to be in $C_{\phi}'$.
    Now, we will again
    exhibit a contradiction by showing that the segment should have ended before
    $c_{\ell+1}$.
    Indeed, the first edge of $c_{\chi}$ is a
    first-visited edge $e'$ whose source $u'$ is in $C_{\phi}'$. 
    By definition of the SCC $C_{\phi}'$ we know that there was a path
    $\pi'$ from $u'$
    to~$y$ in the graph spanned by $w$ before $c_{\chi}$ started.
    The path $\pi'$ and edge $e'$ then witness that the segment $w[i:j]$
    ends upon reaching $y$, contradicting the assumption that a normal chunk of~$w[i:j]$
    will again cross the cut backwards after that point. Thus, we have shown
    by contradiction point (1.): non-final chunks after $c_{\chi}$ no longer
    cross the cut backwards.

    To show (2.), we make a case disjunction depending on the type of the final
    chunk $c_{\ell+1}$. If $c_{\ell+1}$ is a normal chunk, we immediately conclude by \cref{lem:normal-once}.
    If $c_{\ell+1}$ is a tadpole chunk, by \cref{lem:tadpole-twice},
    the chunk $c_{\ell+1}$
    only crosses
    the cut twice if it contains intermediate vertices in $V_-$ and $V_+$.
    But in this case, 
    the definition of $\prec$ in case Tadpole would ensure that all SCCs of
    $C_1, \ldots, C_\kappa$ were left (i.e., the cut is within the tadpole
    vertices, which are put in $\prec$ after all vertices of these SCCs); but
    the case where all SCCs $C_1, \ldots, C_\kappa$ are left was excluded
    earlier in the proof.

    The last case is when $c_{\ell+1}$ is a cycle chunk. By
    \cref{lem:cycle-twice},
    $c_{\ell+1}$ crosses the cut twice
    precisely when its starting and ending vertex $v$ is in $V_-$ but one of
    the intermediate vertices is in~$V_+$.
    By the reasoning in (1.) above, we know that after $c_{\chi}$ and
    before $c_{\ell+1}$ the walk~$w$ never visits a vertex of $C_\phi'$, so
    the attachment point $v$ of $c_{\ell+1}$
    must have been outside of $C_{\phi}'$.
    We then know by case Cycle in the
    definition of~$\prec$ that the intermediate vertices of $c_{\ell+1}$ are put
    right after $v$ in~$\prec$, hence before all vertices of~$C_{\phi}'$. But
    this implies that the vertices of $C_{\phi}'$ are all in $V_+$, so
    $C_{\phi}'$ is right. But we had shown that $C_{\phi}'$ contains $C_\phi$
    which by definition is middle or left. Thus, we have reached a contradiction and $c_{\ell+1}$ does not cross the cut twice.

    We have now covered all cases, so we have successfully established that the
    segment $w[i:j]$ crosses the cut at most three times in total, concluding the proof.
\end{proof}

\section{Computing an Edge-Minimum Walk of Bounded Cutwidth}\label{sec:cw-algo}
Up to now, we have shown \cref{cor:optbound}: optimal solutions to the
\ewm problem have cutwidth at most $3 + 3\log_2 q$. In this section, we conclude the proof of \cref{th:main-theorem} by showing that optimal solutions of bounded cutwidth can be efficiently found.

At a high level, our algorithm proceeds by searching for a shortest path in a graph of configurations. The approach is similar to the token game approach used in~\cite{FeldmanR06} (which finds a shortest path in a graph whose nodes represents the position of tokens moving along the edges of the solution), or to the dynamic programming approach used in~\cite{FeldmannM23} (but we use cutwidth instead of treewidth for simplicity).

The intuitive principle of our algorithm is the following. As we follow a path in the graph of configurations and move from one configuration to another, we choose which edges of the original graph to keep in the solution. 
To be more precise, the path will start on the empty configuration which denotes that no edges have been selected, and will iteratively add vertices and some of their incident edges to the solution. 

To make sure that the selected edges do form a solution, we could try to record
in the configurations the exact set of edges kept in the solution so far.
However, the space of configurations would then be exponential. To avoid
this, a configuration does not really store the entire set of chosen edges: instead it concisely represents the possible lengths modulo $q$ of walks between a small subset of vertices of~$G$, whose size is bounded as a function of the cutwidth.

The distance in the configuration graph from the initial configuration to any configuration $\Xi$ will reflect the smallest cardinality of a set of edges that achieve the set of walks witnessed by~$\Xi$. 
The algorithm therefore looks for a shortest path in the configuration graph from the initial configuration to any configuration which witnesses the $st$-walk of length $r \bmod q$ required by the \ewm problem.

Formally, fix the input graph $G = (V, E)$, the source $s \in V$ and the target $t \in V$, and the modularity requirement $r \bmod q$.
Let $\omega \in \mathbb{N}$ be a \emph{domain size bound}, which is related to the cutwidth, but will be precisely defined later. 
Let us define the notion of configurations formally: 

\begin{definition} [Configuration]
\label{def:config}
A \emph{$V,\omega,q$-configuration} $\Xi=(D,\rho)$ 
is a subset $D$ of at most $\omega$ vertices of~$V$,
called the \emph{domain} of the configuration, together with a function $\rho$ mapping each ordered pair $(u,v)$ for $u,v \in D$ to a subset of $\{0, \ldots, q-1\}$. Having fixed the vertex set $V$ of the input graph $G$, we denote by $\Phi_{\omega,q}$ the set of all possible $V,\omega,q$-configurations, and omit the subscript when clear from context. 
\end{definition}

To compute transitions in the configuration graph, our algorithm will need to compute the \emph{closure} of a configuration $(D, \rho)$. 
The point of the closure is to ensure that, whenever we can achieve a walk from $u\in D$ to $v \in D$ via intermediate vertices of~$D$ and using remainders given by~$\rho$, then that walk can be witnessed directly by a value in $\rho(u,v)$.
Formally:

\begin{definition} [Closure]
    \label{def:closure}
Given a configuration $\Xi = (D, \rho)$, an 
\emph{internal walk} of~$\Xi$ is a sequence of vertices of~$D$ together with a choice of remainder for each step. Formally, it is a sequence $v_1, \ldots, v_\ell \in D$ and a choice of remainders $0\leq r_1, \ldots, r_{\ell-1} < q$ such that we have $r_i \in \rho(v_i, v_{i+1})$ for each $1 \leq i < \ell$. The \emph{total length} of the internal walk is the sum of the remainders modulo~$q$, namely, $\sum_{1 \leq i < \ell} r_i \bmod q$. 

The \emph{closure} of $(D,\rho)$ is the configuration $(D,\rho')$ where, for each $u,v\in D$, for each $r' \in \{0, \ldots, q-1\}$, we have $r' \in \rho'(u,v)$ iff there is an internal walk of total length $r'$ from~$u$ to~$v$ in~$(D,\rho)$.
Note that we always have $\rho'(u,v) \supseteq \rho(u,v)$ because for each $r' \in \rho(u,v)$ we can take the single-edge internal walk $u, v$ with remainder $r'$.
\end{definition}
We can easily compute the closure of a configuration in polynomial time in $q$
and $\omega$, for instance using a product construction. Specifically, create a graph on vertices $D \times \{0, \ldots, q-1\}$, then add the following edges: for each $u,v \in D$ and each $r' \in \rho(u,v)$, for each $i \in \{0, \ldots, q-1\}$ create an edge from $(u,i)$ to $(v,i+r'\bmod q)$. Then compute the transitive closure and define $\rho'(u,v)$ to be the set of $r'$ such that $(u,0)$ has a path to $(v,r')$.

We now define the graph over configurations, which we call $\Gamma_{\omega,q}$
and which also depends on the directed graph $G$ given as input to \ewm.
We omit the subscripts when clear from context.
\begin{definition}
\label{def:configuration-graph}
    The \emph{graph of configurations} $\Gamma_{\omega,q}$  is 
a weighted and labeled directed graph: each edge carries an integer cost and is labeled by a subset of edges of the original graph~$G$.
The vertex set of $\Gamma$ is the set $\Phi$ of all $V,\omega,q$-configurations. To define the edges, let us choose any configuration $\Xi \in \Phi$ and define the outgoing edges of~$\Xi$. These edges are of two kinds:
\begin{itemize}
    \item \emph{Forget}: from $\Xi = (D, \rho)$ with $D$ nonempty, for each $v \in D$,
    letting $D' := D \setminus \{v\}$ be the new domain,
    we have an edge leaving $\Xi$ which has cost 0, is labeled with the empty set of edges, and
    leads to 
    $(D', \rho_{|D'})$ where $\rho_{|D'}$ is the restriction of~$\rho$ to~$D'$
    \item \emph{Introduce}: from $\Xi = (D, \rho)$ with $|D| < \omega$, for each $v \in V \setminus D$, 
   let $D' := D \uplus \{v\}$ be the new domain,
    and let $E_{v,D'}$ be the set of edges of~$G$ which are of the form $(v,
    v')$ or $(v',v)$ with $v' \in D'$; we also add to $E_{v,D'}$ the self-loop edge on~$v$ if it exists.
    Then for each $E' \subseteq E_{v,D'}$, we have an edge leaving $\Xi$ which has cost $|E'|$, is labeled by~$E'$, and leads to
    the closure of the configuration $(D', \rho')$ with $\rho'$ intuitively defined from $\rho$ by adding the edges of~$E'$, formally:
\begin{itemize}
  \item For all $u,u' \in D \cup \{v\}$, we initialize $\rho'(u, u') := \emptyset$.
  \item For each $(u, u') \in D \times D$, we set $\rho'(u,u') := \rho(u,u')$.
  \item For each edge $e \in E_{v,D'}$, we set $\rho'(e) := \{1\}$. Note that by definition of $E_{v,D'}$ this case is disjoint from the previous one.
\end{itemize}
\end{itemize}
Note that we may have several Introduce edges with the same source and target configurations but labeled with different sets of edges and with different costs; in this case all of these edges exist in $\Gamma$ as defined above, although we could equivalently have decided to keep only one of these edges among those with minimum cost.
\end{definition}

We have now defined the graph $\Gamma$ of configurations. We will look for shortest paths in this graph from the \emph{initial configuration} to a \emph{final configuration}, namely:

\begin{definition}
\label{def:initial-final-configuration}
The \emph{initial configuration} is simply the configuration with empty domain.

A configuration is \emph{final} if it contains $s$ and $t$ and features a walk with the requisite remainder, i.e., 
the configuration $(D, \rho)$ is final if $s,t \in D$ and 
$r \in \rho(s,t)$.
\end{definition}

We can now define our algorithm to solve the \ewm problem, which we call the
\emph{\ewm algorithm}.
We first set  the value $\omega$, following the cutwidth bound, to be: $\omega := 6 + 3\log_2 q$,
i.e., three more than the cutwidth bound that follows from \cref{cor:optbound}.
(Intuitively, we add 2 to make sure that the sources $s$ and $t$ can always be part of the domain of configurations, and we add 1 extra to make sure that we can always perform Introduce steps before Forget steps.)
The \ewm algorithm then builds explicitly the graph~$\Gamma_{\omega,q}$
and computes a shortest path $\pi$ in $\Gamma$ from the initial configuration to a final configuration.
Once such a shortest path $\pi$ is found\footnote{
If there is no path in~$\Gamma$ from an initial
configuration to a final configuration, then we return $\emptyset$ to 
indicate that there is no subgraph of $G$ with an $st$-walk of length $r \bmod
q$. Note that this case can be excluded from the outset, simply by checking
whether $G$ contains an $st$-walk of length $r \bmod q$: this can be done, e.g.,
with the product construction.},
then the algorithm returns the subgraph
of~$G$ formed of the edges of~$G$ obtained as the union of all edge labels
in~$\pi$.

We first analyze the time complexity of the \ewm algorithm, and then we show that it is correct. We have:

\begin{claim}
\label{clm:algo-runningtime}
    The \ewm algorithm has running time $n^{\omega+1} \cdot 2^{O(\omega^2 q)}$.
\end{claim}

\begin{proof}
The number of configurations is bounded from above by
$(\omega+1) \cdot n^\omega \cdot 2^{\omega^2 q}$:
we choose the cardinality of the domain from $\omega+1$ possible sizes, then choose concrete elements from the domain and values for the function $\rho$ (we overestimate the number of elements and values by assuming the domain has size $\omega$).
Further, the space to store a configuration is $O(\omega^2 \cdot q)$.
For each configuration, we define at most $\omega$ outgoing edges of type Forget, and
  $O(n \cdot 2^{\omega})$ edges of type Introduce, each of which induces a
  running time of $O(\text{poly}(\omega \cdot q))$ (in particular to
  compute the closure). Hence:
\begin{itemize}
    \item The number $N$ of vertices of the graph is  
    $O(\omega \cdot n^\omega \cdot 2^{\omega^2 q})$
    i.e., $n^\omega \cdot 2^{O(\omega^2 q)}$.
\item The number $M$ of edges of the graph is $n^{\omega+1} \cdot 2^{O(\omega^2 q)}$.
\item The complexity of building the graph is $n^{\omega+1} \cdot 2^{O(\omega^2
  q)} \cdot \text{poly}(\omega\cdot q)$, i.e., 
$n^{\omega+1} \cdot 2^{O(\omega^2 q)}$.
\end{itemize}
Then the complexity of computing the shortest path, e.g., with Dijkstra's algorithm, is in
$O(M + N\log N)$, 
which is again 
$n^{\omega+1} \cdot 2^{O(\omega^2 q)}$.
Hence the overall complexity is $n^{\omega+1} \cdot 2^{O(\omega^2 q)}$.
\end{proof}

We now show that the \ewm algorithm correctly solves the \ewm problem. There are two directions to the proof. First, we must show that the algorithm is \emph{sound}:
whenever the algorithm returns a subgraph $E'$ then indeed that subgraph is a candidate solution, i.e., it contains an $st$-walk of length $r \bmod q$. 
Second, we must show that the algorithm is \emph{complete}: whenever $E'$ is an optimal solution, i.e., an edge-minimum subgraph of~$G$ with an $st$-walk of length $r\bmod q$,
then there is a path in~$\Gamma$ of cost $|E'|$ whose union of labels is~$E'$
without duplicates (hence the path cost is $|E'|$) that goes from an initial
configuration to a final configuration. This implies that the algorithm always returns a subgraph of cost no more than $E'$.
Combining the soundness and completeness, we deduce that the algorithm always
returns an optimal solution if one exists, i.e., when the entire graph $G$
contains an $st$-walk of length $r \bmod q$. Further, when there is no candidate solution
at all then the algorithm correctly identifies it.

For the soundness and completeness proof, it will be useful to state an invariant about the configurations reached in any path starting from an initial configuration, which intuitively states that the remainders expressed in the configuration correctly reflect the possible remainders of paths between vertices of the configuration using the edges taken so far. Namely:

\begin{claim}
\label{clm:invariant-dyn}
  Let $\pi = \Xi_1, \ldots, \Xi_\ell$ be a path in $\Gamma$ from the initial
  configuration $\Xi_1$. Let $G_1 = (V, E_1)$, $\ldots$, $G_\ell = (V, E_\ell)$
  be the sequence of subgraphs of~$G$ defined in the following way: we have $E_1
  = \emptyset$, and for each $1 < i \leq \ell$ we set $E_i = E_{i-1} \cup E'_i$
  where $E_i'$ is the edge set that labels the edge of~$\Gamma$ used to go from
  $\Xi_{i-1}$ to~$\Xi_i$ in~$\Gamma$. Then the following is true: for any $1
  \leq i \leq \ell$, writing $\Xi_i = (D_i, \rho_i)$, for any $u_1, u_2 \in D_i$,
  for any $r' \in \{0, \ldots, q-1\}$, we have $r' \in \rho_i(u_1, u_2)$ iff there
  is a $u_1 u_2$-walk in $G_i$ whose length modulo~$q$ is~$r'$.
\end{claim}

Note that, in the definition of the graphs $G_i$, we may add the same edge
of~$G$ multiple times because it occurs in the label of several edges of~$\pi$;
this can happen even if $\pi$ is a simple path in~$\Gamma$. As it turns out,
this never happens when $\pi$ is a shortest path (because we are then, in
effect, paying twice for the same edge); but the invariant of
\cref{clm:invariant-dyn} applies to paths $\pi$ even if they do traverse edges labeled with non-disjoint edge sets.

\begin{proof}[Proof of Claim~\ref{clm:invariant-dyn}]
We show this claim by induction on the length of the path~$\pi$.
The base case of the induction is immediate: if $\ell=1$ then $\pi$ only consists of the initial configuration, then $\Xi_1 = (D_1, \rho_1)$ is the initial configuration so both $D_1$ and $G_1$ are empty. 

For the induction case, let us show the claim for $1 < i \leq \ell$. If the edge from $\Xi_{i-1}$ to $\Xi_i = (D_i, \rho_i)$ is a Forget edge, then we have $G_i = G_{i-1}$ and 
$D_i \subseteq D_{i-1}$ so we immediately conclude using the induction hypothesis.
If the edge is an Introduce edge, then let us show both directions.
Let us fix the endpoints $u_1,u_2 \in D_i$ and the remainder $r' \in \{0, \ldots, q-1\}$.

For the forward implication, assume that we have $r' \in \rho_i(u_1,u_2)$.
  We know that $\Xi_i$ was defined by applying the closure operator: let $\Xi_i' = (D_i',\rho_i')$ be the configuration defined in \cref{def:configuration-graph} before applying the closure operator. 
Recalling \cref{def:closure},
the fact that $r' \in \rho_i(u_1,u_2)$ must be witnessed by an internal walk of $\Xi_i'$ of total length $r'$. 
By definition of $\Xi_i'$ in \cref{def:configuration-graph}, this internal walk
  can traverse two kinds of edges of $\Xi_i'$: edges that already existed in
  $\Xi_{i-1}$, and edges that correspond to edges of~$G$ added as part of
  $E_i'$. We apply the induction hypothesis to the first kind of edges to obtain witnessing walks in $G_{i-1}$ with the same endpoints and remainder.
  Together with the edges of~$E_i'$ used in the internal walk, this gives us a witnessing walk whose length modulo~$q$ is the total length $r'$ of the
  internal walk. This establishes that $G_i$ indeed contains the requested
  $u_1 u_2$-walk.
  
For the backward implication, assume that there is a walk from $u_1$ to $u_2$ in $G_i$ whose length modulo $q$ is $r'$.
  Decompose this walk to isolate the occurrences of the edges of $E_i'$ on the
  one hand, and the subwalks using only edges of $G_{i-1}$ on the other hand.
  Using the induction hypothesis, we know that the subwalks of the second kind
  are reflected in the function $\rho_{i-1}$ which witnesses the existence of
  these walks between their endpoints and with the right remainder. The
  definition of $\rho_i'$ ensures that whenever we have $r' \in
  \rho_{i-1}(u_1',u_2')$ for some $r' \in \{0, \ldots, q-1\}$ and some
  vertices $u_1'$ and $u_2'$, then we have
  $r' \in \rho_i'(u_1',u_2')$, and it also ensures that for each edge $(u,v) \in
  E_i'$ with $v$ the one vertex of $D_i \setminus
  D_{i-1}$ and $u \in D_i$,
  we have $1 \in \rho_i'(u,v)$; the same applies to edges of the form $(v,u)
  \in E_i'$. This information on $\rho_i'$ allows us to
  conclude to the existence of an internal walk from $u_1$ to $u_2$ with total
  length~$r'$, which ensures that $r' \in \rho_i(u_1,u_2)$ because we took the closure of~$\Xi_i'$ to define~$\Xi_i$.

We have shown the equivalence, which establishes the inductive step and concludes the proof of the invariant.
\end{proof}

This invariant immediately implies the soundness of the algorithm:
\begin{claim}
\label{clm:algo-correct}
    The \ewm algorithm is sound.
\end{claim}

\begin{proof}
Assume that the algorithm finds a shortest path $\pi = \Xi_1, \ldots, \Xi_\ell$ from the initial configuration to a final configuration, and let $E'$ be the set of edges that the algorithm returns. 
Let us use \cref{clm:invariant-dyn} -- note that we do not even use the fact that $\pi$ is a shortest path.
Applying that claim to the configuration $\Xi_\ell$, since the configuration is final,
we know that $G_\ell = (V, E')$ has an $st$-walk of length $r \bmod q$, so it is a candidate solution. This establishes the soundness of the algorithm.
\end{proof}

We then show completeness, which will rely on the other direction of the invariant, but also use this time the cutwidth bound:
\begin{claim}
\label{clm:algo-complete}
    The \ewm algorithm is complete.
\end{claim}

\begin{proof}
  Let $E_0$ be an optimal solution to \ewm.
  We use \cref{cor:optbound} to argue that the cutwidth of~$E_0$ is at most $3 + 3 \log_2 q$, i.e., remembering our definition of~$\omega$, the cutwidth of $E_0$ is at most $\omega-3$.
Let $\prec$ be an order that witnesses the fact that $E_0$ has cutwidth at most $\omega-3$, and number the vertices of $V$ accordingly as $v_1, \ldots, v_n$. We consider the sequence of 
vertex subsets $D_0, \ldots, D_n$ where $D_0 = \emptyset$ and for each $1 \leq i \leq n$ we let $D_i$ be the subset of vertices of $\{v_1, \ldots, v_i\}$ that occur in an edge together with a vertex of $\{v_{i+1}, \ldots, v_n\}$.
We claim that $|D_i| \leq \omega-3$ for each $1 \leq i \leq n$. Indeed, considering the cut $V_- = \{v_1, \ldots, v_i\}$ and $V_+ = \{v_{i+1}, \ldots, v_n\}$, given that the cutwidth is at most $\omega-3$ we know that the number of edges involving both a vertex of~$V_-$ and of $V_+$ is at most $\omega-3$, so at most $\omega-3$ vertices in $V_-$ occur in an edge of $E_0$ together with a vertex of~$V_+$, and $D_i$ is precisely the set of these vertices.

We will show the existence of a path in $\Gamma$ via configurations 
  defined from the sets $D_i$ (with small changes to handle the source $s$ and target $t$ differently).
  We will specify the configurations $(D', \rho')$ along the path in terms of the set of vertices $D'$ that they have as first components, and the sets $E'$ of edges that they keep: the functions $\rho'$ are then given accordingly following \cref{def:configuration-graph}. Specifically, let us build a sequence of configurations
$\pi = \Xi_0, \Xi_1', \Xi_1, \ldots, \Xi_n', \Xi_{n}$ as follows:
\begin{itemize}
\item $\Xi_0$ is the initial configuration
\item Each $\Xi_i' = (D_i', \rho_i')$ is obtained from the previous configuration $\Xi_{i-1}$
by setting
$D_i' := D_{i-1}'' \cup \{v_i\}$ and performing
    an Introduce step which adds the vertex~$v_i$ and selects the set of edges $E'_i$ consisting of all edges $e$ of $E_0$ such that one vertex of~$e$ is $v_i$ and the other is a vertex $v_j$ with $j\leq i$; in particular, if $E_0$ contains a self-loop on $v_i$, then we add it at that moment.
    (We explain below why these edges $e$ of~$E_0$ can actually be selected in the definition of~$\Xi_i'$, i.e., why they involve vertices that are all in~$D_i'$.)
    The function $\rho'_i$ is defined according to the definition of Introduce steps.
  \item Each $\Xi_i = (D''_i, \rho_i)$ is obtained from $\Xi_i'$ by Forget steps (possibly none, one, or many) to forget the vertices of $D_{i-1} \setminus (D_i \cup \{s, t\})$, i.e., we forget the vertices that are not in $D_i$ except that $s$ and $t$ are never forgotten.
In other words each $\Xi_i$ is of the form $(D_i'', \rho_i)$ with $\rho_i$  being the restriction of $\rho_i'$
    to $D_i''$, and $D_i''$ being $D_i$ with the possible addition of~$s$ and~$t$.
\end{itemize}
  Note how the cutwidth bound of $\omega-3$ ensures that the domain of configurations always has the right cardinality: we have shown $|D_i| \leq \omega-3$ for each $1 \leq i \leq n$, and for all $i$ we have $|D_i''| \leq |D_i| + 2$, and $|D_i'| = |D_{i-1}''| + 2$.
Thus, this process correctly defines a sequence of configurations, and 
the second item ensures that,
in the edges of $\Gamma$ traversed to define $\pi$, the union of the $E_i'$ will be precisely $E_0$, and every edge of~$E_0$ will occur only once. Indeed, for every edge $e$ of~$E_0$ containing $v_i$ and $v_j$ with $i\leq j$, then $e$ gets added at the point where we consider $v_j$, in particular $e$ is either a self-loop on $v_i = v_j$ and we have $v_i \in D_i'$, or $v_i \prec v_j$ and $v_i$ is in $D_i'$ because it is in $D_{i-1}$ by definition (indeed $v_i$ occurs in an edge together with $v_j$ and $j > j-1$).
What is more, the total cost of~$\pi$ will be $|E_0|$ as required.
  (This argument also ensures that the sets of edges $E_i'$ selected in the second bullet point are edges involving vertices that are all present in the domain~$D_i'$.)

The only remaining point is to show that~$\pi$ finishes on a final configuration.
By assumption, the graph $(V, E_0)$ contains an $st$-walk of length $r \bmod q$. Applying the invariant of \cref{clm:invariant-dyn} to the path~$\pi$, considering the sequence of graphs $G_i$ defined in the lemma statement, we know that the last graph in this sequence is $(V, E_0)$, because we argued that the union of the $E_i'$ is precisely $E_0$. So we know that there is an $st$-walk of length $r \bmod q$ in this last graph. The construction above ensures that $s$ and $t$ get added at some point in the sequence of configurations and that they never get removed, so we know that $s$ and $t$ are in the domain of the final configuration. Hence, applying the other direction of the invariant, we know that $r \in \rho(s,t)$ for $\rho$ the second component of the final configuration. This means that the path $\pi$ does indeed finish on a final configuration.
This is all that remained to be shown, so the proof is complete.
\end{proof}

Putting everything together, we have defined the \ewm algorithm which ensures by completeness
(\cref{clm:algo-complete})
that optimal solutions $(V, E_0)$ to \ewm must
give rise to a path in $\Gamma$ from the initial configuration to a final configuration labeled by edge sets whose union is $E_0$ without duplicates, so that the cost of the path is~$|E_0|$.
(This direction uses \cref{cor:optbound}: optimal solutions have bounded cutwidth.)
Conversely, since the \ewm algorithm is sound
(\cref{clm:algo-correct}),
it can only return paths labeled by edge sets whose union gives a subgraph $(V,
E_0)$ which is a candidate solution, and whose cost is then at least $|E_0|$
(and possibly more if there are duplicates, though again this does not actually
occur along shortest paths). This implies that, indeed, either there is no
candidate solution in~$G$ at all and the \ewm algorithm correctly fails, or there are optimal solutions and the algorithm must return one of them. The running time is given as \cref{clm:algo-runningtime}. 
Altogether, we have shown that the algorithm correctly solves the \ewm problem, and has complexity 
$n^{\omega+1} \cdot 2^{O(\omega^2 q)}$. Instantiated with our choice of $\omega = 3 + 3 \log_2 q$, we get a final complexity of 
$n^{O(\log q)} \cdot 2^{O(q \log^2 q)}$ for the algorithm.
This concludes the proof of our main result (\cref{th:main-theorem}).

\section{Extension to Weighted Graphs}\label{sec:costs_distances}

Having shown our main result (\cref{th:main-theorem}), in this section we start exploring variants and extensions of the \ewm problem. We first study two extensions of \ewm in this section. First, we study the addition of \emph{costs} on edges, meaning that we look for a walk where the total cost is minimum (instead of the number of edges).
Second, we study the addition of integer \emph{lengths} on edges. 

In this section and the next, the problems that we define and study are always
posed on a directed graph $G$ and ask about the existence of a subgraph of~$G$
satisfying certain properties and which is optimal according to some criterion
(usually, being edge-minimum). We use the same terminology as before: a
\emph{candidate solution} is a subgraph which satisfies the properties but is
not necessarily edge-minimum, and an \emph{optimal solution} is a candidate
solution which is additionally edge-minimum.

\myparagraph{Costs on edges and vertices.}
We first study the extension of the \ewm problem with costs on edges. Specifically, the \emph{\ewm problem with costs on edges} takes as input a directed graph $G = (V, E)$, a cost function $\gamma$ giving to each edge $e \in E$ a cost $\gamma(e)$, a pair of a source $s \in V$ and target $t \in V$, and a modularity requirement $r \bmod q$ for integers $q$ and~$r$.
The output to the \ewm problem with costs on edges is 
an $st$-walk of length $r \bmod q$ 
such that the cost $\gamma(E_w) := \sum_{e \in E_w} \gamma(e)$ is minimum.
We assume that all costs given by~$\gamma$ are nonnegative. Indeed, in the presence of negative weights, we lose the correspondence explained in \cref{sec:prelim}: computing a minimum-weight \emph{subgraph} that contains a walk reduces to the case of positive weights (all negative-weight edges will always be included in the optimal solution subgraph); but computing a minimum-weight \emph{walk} is \NP-hard even without modularity constraints:

\begin{proposition}
  \label{prp:npc}
  The following problem is \NP-complete: given a directed graph $G = (V, E)$, a cost function $\gamma\colon E \mapsto \mathbb{Z}$, and source and target $s, t \in V$, compute a minimum-weight subset of edges $E' \subseteq E$ such that there is an $st$-walk using precisely the edges in~$E'$.
\end{proposition}

\begin{proof}
  We reduce from the Directed Steiner Tree (DST) problem, which is
    known to be \NP-hard (the proof is similar to the one for undirected Steiner Tree~\cite{karp1972reducibility}).
  Given a directed graph $G$ of weight-1 edges, a root $r$, and a subset $T$ of vertices, the DST problem asks for a minimum-weight graph that contains a path from $r$ to every vertex of $T$.
  We assume without loss of generality that the input graph $G$ contains a path from~$r$ to every vertex of~$T$, as there is clearly no solution otherwise.
  Given $G$, $r$, and $T$, to perform our reduction, we add to $G$ an edge of weight $-(|E|+1)$ from every vertex of $T$ to $r$, and give weight~$1$ to the original edges of~$G$.

  We then take $s := r$ and $t := r$. An $rr$-walk of minimum weight in the resulting graph will then consist of the $|T|$ negative-weight edges, together with some minimum-cardinality subset of original edges of~$G$.
  To visit the $|T|$ negative edges, the walk must go from $r$ to every vertex in $T$, so the optimal solution returns a minimum-weight Steiner Tree 
  plus the $|T|$ negative edges. This completes the reduction.
\end{proof}

Our tractability result for \ewm (\cref{th:main-theorem}) can be easily extended to solve the \ewm problem with costs on edges. Specifically, the Segment Decomposition Lemma (\cref{lem:bounding-number-of-segments}) shows that there must be an optimal solution to \ewm with costs on edges that obeys the bound on the number of segments, because the $\gamma$ function on subset of edges is monotone. Then the Segment Cutwidth Bound (\cref{prp:segment-bound}) applies as is, and the algorithm to find an optimal solution of bounded cutwidth from the previous section can be easily extended: simply modify the definition of the graph of configurations (\cref{def:configuration-graph}) so that the cost of an edge labeled with a set $E'$ of edges of~$G$ is no longer the cardinality $|E'|$ of $E$ but the total cost $\gamma(E')$. Other than that, the algorithm proceeds in the same way, and computes a shortest path which now reflects the subgraph of~$G$ satisfying the requirements which has minimum cost instead of minimum cardinality.

We also mention that the \ewm problem can be defined to have costs (unit costs or not) on vertices instead of edges. 
In this case, the cost of a candidate solution is the number of different vertices traversed by the walk (or more generally their total cost).
However, this problem is interreducible to the \ewm problem with costs on edges:

\begin{lemma}
    The \ewm problem (with unit costs on edges) reduces to the \ewm problem with unit costs on vertices, and the \ewm problem with costs on edges reduces to the \ewm problem with costs on vertices.
    Conversely, the \ewm problem with unit costs on vertices reduces to the \ewm problem with unit costs on edges, and the \ewm problem with costs on vertices reduces to the \ewm problem with costs on edges.
\end{lemma}

\begin{proof}
In one direction, we subdivide edges with cost $c$ to have length $q+1$, give cost $c$ to one of the intermediate vertices, and give cost zero to other vertices in the subdivision and to the original vertices of the graph. 
Further, if the edges all have unit cost, we can still reduce in polynomial time to the \ewm problem with unit cost of vertices: we subdivide each edge to add $qn+1$ intermediate vertices each having cost 1. This then ensures that the path lengths are correct and that the unit cost of the $n$ original vertices is negligible, so an optimal solution in the rewritten graph will minimize the number of traversed edges in the original graph.

For the converse direction, we replace each vertex having cost $c$ by a directed path of length $q$, with incoming edges leading to the first vertex and outgoing edges leaving from the last vertex. We put the cost of $c$ on one of the edges, giving $0$ cost to the other newly created edges and to the original edges. (Depending on how the cost of the source and target vertices is taken into account, we may take the source or the target in the resulting graph to be the first or the last vertex of the path representing $s$ and $t$ respectively.) 

If all vertices have unit cost, we can replace each vertex by a path of length $q(m+1)$ of unit-cost edges. This ensures that the cost of the $m$ original edges is negligible, so that an optimal solution in the rewritten graph will minimize the number of traversed vertices in the original graph.
\end{proof}

\myparagraph{Lengths on edges.}
Having studied the use of costs on edges to change the optimization criterion, we turn to a different problem variant where we annotate edges with integer lengths. Specifically, the \emph{\ewm problem with lengths} takes as input a directed graph $G = (V, E)$, a length function $\delta\colon E \to \mathbb{N}$, and a pair of a source $s \in V$ and target $t \in V$ and a modularity requirement $r \bmod q$ for integers $q$ and $r$. The answer to the \ewm problem with lengths is an $st$-walk $w$ that traverses a minimum number of distinct edges and whose total length is $r \bmod q$, i.e., $\delta(w) := \sum_{1 \leq i \leq |w|} \delta(w[i])$ is $r \bmod q$. (Note that the length of an edge is summed as many times as the edge is traversed.)

The impact of allowing lengths is different depending on the regime. In the case where $r$ and $q$ are constant, or given in unary, then we can rewrite the graph in polynomial time to eliminate edge lengths.
Specifically, letting $m$ be the number of edges in the graph,
we replace each edge $e$ of length $\delta(e)$ by a path on $(m+1)q + (\delta(e) \bmod q)$ edges (where $m$ is the number of edges in the graph).
This ensures that a walk in the new graph has the same modularity as the corresponding walk in the original graph, and the cost of a candidate solution in the rewritten graph is $(m+1)qM+ \epsilon$, where $M$ is the number of original edges traversed and $\epsilon \leq m q$. This ensures that an optimal solution in the original graph indeed minimizes the number of edges taken from the original graph.

One different regime is when $r$ and $q$ are written in binary and do not necessarily have polynomial value. 
In this case, the \ewm problem with lengths written in binary can be shown to be
\NP-hard by an easy reduction from Subset Sum. In fact, just the problem of
deciding the existence of a walk of length $q\bmod r$ is \NP-hard:

\begin{proposition}\label{prop:walk-of-given-mod-on-binary-length-is-np-hard}
    The problem of deciding whether an $st$-walk of length $q \bmod r$ exists, where $q,r$ and the lengths of $\delta$ are written in binary, is \NP-hard.
\end{proposition}

\begin{proof}
  We reduce from the Subset Sum problem, which asks, given a target $T$ and a set of numbers $s_1$,\dots,$s_N\leq T$, whether there is a subset of numbers that sum to $T$. 
  Create a graph with vertices $s = u_1, \ldots, u_{N+1} = t$ where, for each $1 \leq i \leq N$, 
  there are two parallel edges from $u_i$ to $u_{i+1}$, one of length 0 and the other of length $s_i$. 
  More precisely, to avoid parallel edges so as to make the graph simple, replace each pair of parallel edges with a pair of paths on 2 edges each). 
  Set $r:=T$ and take a sufficiently large $q$, for instance $q:=N\cdot T$.
  Then there is (deterministically) no $st$-walk of length $>N\cdot T$, and there is an $st$-walk of length $T$ if and only if the answer to the subset sum instance is ``yes''.
\end{proof}

We do not know whether an analogous hardness result holds for the original \ewm
problem (without edge lengths) in the setting where $q$ and $r$ can be
written in binary and do not necessarily have polynomial value.
Indeed, in the proof above, we crucially use the edge lengths to concisely write the numbers given in the Subset Sum instance; writing them in unary as paths of the corresponding length will not give a polynomial-time reduction and indeed the Subset Sum problem can be solved in pseudo-polynomial time.

We last address the question of showing an upper bound on the complexity of \ewm with lengths, by showing an \NP upper bound. Of course the bound is phrased on the decision version of \ewm with lengths, i.e., given an instance of \ewm with lengths and a threshold $k$, we wish to decide whether there exists a candidate solution having at most $k$ different edges. We have:

\begin{proposition}
  The following problem is in \NP: given a directed graph $G = (V, E)$, a
  function $\gamma\colon E \to \mathbb{N}$ that assigns lengths (written in unary) to
  each edge, integers $q$, $r$ and $k$ (written in binary), source and target $s,t
  \in V$, 
  decide if there is a subset $E' \subseteq E$ of at most $k$ edges that contains an $st$-walk of length $r\bmod{q}$ according to~$\gamma$.
\end{proposition}

\begin{proof}
 The difficulty is that the certificate cannot be the walk itself because it is not necessarily of polynomial size. First we observe that the length of the walk is at most $qm$: If a walk uses a single edge more than $q$ times, by the pigeonhole principle the walk reached that edge twice having the same remainder. Thus, we can delete the portion of the walk between those two traversals of that edge, without changing the remainder. Now, to show that \ewm is in \NP, we guess the length of the walk and, more importantly, the set of edges used by the walk (altogether forming a polynomial-sized guess)
 and we then use a known certificate to find a walk of that exact given
  length~\cite{basagni1997difficulty}.
  This gives us in particular the number of distinct edges traversed by the walk. 
\end{proof}

\section{Connections to \dsn and \scss}
\label{sec:dsn}
We conclude the paper by exploring how \ewm relates to the Directed Steiner
Network (\dsn) problem mentioned in the introduction, and more specifically 
the \emph{Strongly Connected Steiner Subgraph (\scss)} problem.
We first define the \scss problem and show that it is subsumed by \ewm, in the
sense that a polynomial algorithm for \ewm (with fixed $q$ and $r$) gives a
polynomial algorithm for \scss with a constant number of terminals. Then, we
introduce a problem generalizing both \scss and \dsn on the one hand, and \ewm in the other: we study how to find the smallest subgraph satisfying connectivity requirements on specified endpoint pairs with specified modularities. We show that we can generalize our results to this problem, and provide a polynomial algorithm for the setting when the number of endpoint pairs and the modularity requirements are constants.

\myparagraph{Reducing \scss to \ewm.}
Let us recall the definition of the \scss problem~\cite{FeldmanR06,FeldmannM23}
before explaining how it can be reduced to \ewm. In the \scss problem, the input
simply consists of a directed graph $G = (V, E)$ with an input set $T \subseteq
V$ of terminals, and we want to find the smallest subgraph $E' \subseteq E$ that
is strongly connected and contains edges incident to each terminal in~$T$. In
other words, \scss is a special case of \dsn where the connectivity requirements on the vertices of~$T$ require that they are strongly connected. The \scss problem is NP-hard in general (by an easy reduction from the directed Steiner tree problem)
but it can be solved in polynomial time provided that the size of $T$ is bounded by a constant, as shown in~\cite{FeldmanR06}.
Thus, for $k \in \mathbb{N}$, we write \kscss{$k$} to refer to the \scss problem where the set $T$ of terminals is required to contain at most $k$ vertices.
Following our previous terminology in the paper, when we talk of \emph{candidate
solutions} we mean a subgraph satisfying the requirements of a problem (e.g., for
\kscss{$k$}, a strongly connected subgraph containing the requisite terminals),
and by \emph{optimal solution} we mean a
candidate solution which is also edge-minimum.

Let us show that the \scss problem can be reduced to \ewm, which implies that \cref{th:main-theorem} gives an alternative proof of the results of~\cite{FeldmanR06} (with worse bounds). More precisely, we show:

\begin{lemma}\label{lem:reducing-scss-to-swm}
For any constant $k >0$, we can compute constants $q$ and $r$ in $2^{O(k\log
  k)}$ such that there is a linear-time reduction from \kscss{$k$} to \ewm with the constant values $q$ and $r$.
\end{lemma}

\begin{proof}
For each $i >0$, we let
$p_i$ denote the $i$-th prime number. Recall that the \emph{primorial} of $p_i$ is the value $p_i\# := \Pi_{j=1}^i p_j$, 
and that it is bounded by $e^{(1+o(1))i\log i}$~\cite[A002110]{oeis}.
Thus, for any $k >0$, we set $q := p_k\#$
and we let $r=1$. 

Let us now define the reduction. From the graph $G = (V, E)$, given a set $T'$ of terminals with $|T'| \leq k$, we pad $T'$ to
  $T = \{v_1, \ldots, v_k\}$ by repeating terminals if necessary so that $T$ consists of exactly $k$ terminals.
  Then we define $G' = (V', E')$ by subdividing the edges of~$G$ to length $q$, and by adding on each $v_i$ a self-loop of length $\ell_i := p_k\#/p_i$ for $1 \leq i \leq k$.
We set the source and target in~$G'$ for \ewm to be $v_1$.
Note that, thanks to the assumption that $k$ is constant, this construction runs in linear time because $q$ and $r$ are constants and the values $\ell_i$ can also be computed.

We claim that walks in $G'$ achieving a modularity of $r \bmod q$ must include all of the self-loops. Indeed, if the $i$-th self-loop is not part of the walk, then the walk length must be a combination of edges with length $q$ which is a multiple of~$p_i$, and other self-loops whose lengths are all multiples of $p_i$: thus, the walk length is $0 \bmod p_i$ whereas $r$ is $1 \bmod q$ and therefore $1 \bmod p_i$, a contradiction.
Any optimal solution to our \ewm instance on~$G'$ must therefore visit all the $v_i$'s in $G'$, and clearly for each edge $e$ of $G$ the solution either takes all $q$ edges coding $e$ in~$G'$ or takes none, hence the solution has cost $m' q +\sum_i \ell_i$ 
for some $m'$ and we can easily deduce (in linear time) from the solution a subgraph of~$G$ with~$m'$ edges which is an optimal solution to $k$-SCCS.

Conversely, by the generalized Bezout identity, we know that there
are positive integers $N_1,\dots, N_k$ such that $\sum_{i=1}^k N_i \ell_i = 1\bmod q$.
  Consider any candidate solution $H$ to \kscss{$k$} on $G$ that takes $m'$ edges.
By definition of $H$ being a candidate solution, it is strongly connected and
  features edges incident to all terminals of~$T$. Thus, there is a walk
  $w$ from $v_1$
  to~$v_1$ which uses only edges of~$H$ and visits all terminals of~$T$.
  From $w$, we can build a walk~$w'$ in $G'$ from $v_1$ to $v_1$ that traverses all the same
  $m'$ subdivided edges, and loops exactly $N_i$ times on the loop of length
  $\ell_i$ attached to $v_i$ for each $i\leq k$ at the point where $v_i$ is
  visited. Each of the $m'$ subdivided
  edges outside the loops may be traversed an arbitrary number of times, but has
  length $q$, so $w'$ has length $\sum_{i=1}^k N_i \ell_i \bmod q$ and thus
  the set of edges of~$w'$ is a candidate solution to \ewm on~$G'$.

  Thus, the linear-time reduction, given $G$, consists in building $G'$, computing the optimal solution to \ewm as specified, and recovering from it the subgraph of~$G$ which is an optimal solution to \kscss{$k$}. This concludes the proof.
\end{proof}

The previous reduction shows that our \ewm problem is intuitively at least as complicated as \scss, given that the tractability of \ewm for constant $q$ implies that of \scss for constant $k$. (This being said, we do not claim that the polynomial algorithm given by our results is as efficient as that of earlier works~\cite{FeldmanR06,FeldmannM23}.)
Alternatively, instead of this black-box reduction, we can also adapt our techniques to recapture the tractability of \scss by showing a bound on the cutwidth of optimal solutions, using the segment decomposition. We exemplify below how this is done, and will revisit this afterwards when extending \ewm to support multiple paths.
Recall the notion of timestamp-minimum walks (\cref{def:timestampmin}), and let
us show the following result; note that it applies to \emph{edge-minimal}
solutions (not just optimal solutions).

\begin{lemma}\label{lem:scss-direct}
    Let $G = (V,E)$ and $T = \{v_1, \dots, v_k\}$ be an \scss instance. 
    Every edge-minimal 
    solution $H$ to the instance has cutwidth at most $3k$.

    What is more, for any terminal $v_i$, considering all $H$-walks that start and
    conclude in $v_i$ and visit all terminals from $T$, then any timestamp-minimum $H$-walk among those $H$-walks has at most $k$ segments.
\end{lemma}
\begin{proof}
Let $H$ be an edge-minimal solution to the \scss instance. We show the second part
  of the claim, which implies the first part thanks to
  the Segment Cutwidth Bound (\cref{prp:segment-bound}).
  Up to renumbering the terminals, let us show the claim for
  the terminal $v_1$. We consider the
  walks in $H$ that start at $v_1$, visit all terminals $v_2, \ldots, v_k$, and
  finish by returning to the first terminal $v_1$. (We do not constrain the order in which the terminals of $v_2, \ldots, v_k$ are visited, and we allow revisits of already-visited terminals, including $v_1$.)

  By definition of \scss, the existence of such walks is guaranteed, and any such walk covers $H$ by edge-minimality. Consequently, 
  by \cref{lem:segbound},
  the cutwidth of $H$ is at most $3$ times the number of segments in such a walk. 
It remains to prove that when we fix a walk $w$ which is timestamp-minimum among such walks then $w$ has at most $k$ segments.

We consider the segment decomposition $s_1 \cdots s_\xi$ of~$w$.
We claim that, for each successive segment except the last, there must be one
  terminal of $\{v_2, \ldots, v_k\}$ which does not occur in the edges of earlier segments and is visited by the segment as an intermediate vertex. 

Indeed, let us proceed by contradiction and assume that some segment $s_\sigma$ violates this condition for some $1 \leq \sigma < \xi$. Then using \cref{lem:segment-contains-detour}, since the segment finishes before the end of the walk, we can replace the segment detour $\textnormal{det}_\sigma$ by the path $p_\sigma$. Letting $w'$ be the resulting walk, we know that $w'$ still starts at $v_1$, still finishes at $v_1$, and still visits all terminals in some order; indeed, all terminals visited as intermediate vertices of the segment $s_\sigma$ are visited by earlier segments, so replacing $\textnormal{det}_\sigma$ by $p_\sigma$ did not remove a terminal which does not occur elsewhere.

Furthermore, the new walk $w'$ is still an $H$-walk, because $H$ is edge-minimal. Therefore, a similar argument to the one in the proof of \cref{lem:bounding-number-of-segments} allows to conclude that the new walk $w'$ does not satisfy $w \trianglelefteq w'$. 
Specifically, observe that $w$ and $w'$ can be written as $w=w_1 \textnormal{det}_\sigma w_2$ and $w'=w_1 p_\sigma w_2$ for some (common) walks $w_1,w_2$. So, the same edges must be first-visited in $\textnormal{det}_\sigma w_2$ and $p_\sigma w_2$. But in $w'$, no edge in $p_\sigma$ is first-visited, because all these edges already occur in $w_1$. By contrast, in $w$ at least one edge in $\textnormal{det}_\sigma$ is first-visited. The non-empty set $E'$ of such edges are also visited in $w'$, but they must be visited for the first time in $w_2$ in $w'$. This means the edges from $E'$ are visited closer to the end in $w'$ than in $w$, while the other edges first appearing after $w_1$ are visited at the same position (from the end) in $w$ and $w'$. This implies that we do not have $w \trianglelefteq w'$, which contradicts the timestamp-minimality of~$w$. 

We have thus proved that, for each segment $s_\sigma$ with $1 \leq \sigma < \xi$, there is a terminal of $\{v_2, \ldots, v_k\}$ that occurs as an intermediate vertex of~$s_\sigma$ and does not occur before $s_\sigma$ in~$w$. This immediately implies that $\xi \leq k$, i.e., there are at most $k$ segments, as we wanted to show.
\end{proof}

We remark that the cutwidth bound of $3k$ given by this result is slightly better than the bound of $6k$ shown in~\cite{FeldmannM23}.

\myparagraph{Defining Directed Steiner Network with Modularity constraints.}
The reduction from \scss to \ewm leads to the natural question of whether there
could be a problem which subsumes \scss and \ewm,
or even \dsn and \ewm; and a polynomial-time algorithm that subsumes the
tractability of all these problems. We now introduce such a problem, dubbed
\emph{Directed Steiner Network with Modularity} (\dsnm), and show a
polynomial-time algorithm to solve it.

\begin{definition}
For $k > 0$, the $k$-\emph{Directed Steiner Network with Modularity} problem (\kdsnm{$k$})  takes as input a graph $(V,E)$ and a \emph{$k$-requirement specification} $\mathcal{R}$
  which consists of $k$ tuples $(s_i, t_i, r_i, q_i)$ for $1 \leq i \leq k$ such
  that $s_i, t_i\in V$ and $q_i > 0$ and $0 \leq r_i < q_i$.
Our goal is to compute an optimal solution, i.e.,
an edge-minimum subgraph $H$ of $G$ such that $H$ contains a walk from $s_i$ to $t_i$ of length $r_i\bmod q_i$ for every $i\leq k$.
\end{definition}

Our last result in this paper is to show that the \kdsnm{$k$} problem is in PTIME when $k$ and the modulo values $q_i$ are constants. This result subsumes \cref{th:main-theorem} 
and the tractability of \kdsn{$k$} for constant $k$ shown in~\cite{FeldmanR06,FeldmannM23}. 

\begin{theorem}\label{th:kdsnm-tractable}
    We can compute in $n^{O(k + \log q)} \cdot 2^{O(q (k + \log q)^2)}$ time an optimal solution to \kdsnm{$k$}, where $q$ denotes the least common multiple (LCM) of the $q_i$ in the input $k$-requirement specification.
\end{theorem}

Note that our exponents are given up to constant factors, so we do not claim to recover the same exponents as earlier works on \kdsn{$k$} without modularity constraints (i.e., for $q=1$).

The overall strategy to prove \cref{th:kdsnm-tractable} is to follow the methodology used for \ewm, adjusting the constructions presented so far in the paper. 
Deviating somewhat from \ewm, but following prior work about
\dsn~\cite{FeldmanR06,FeldmannM23}, in our study of \dsnm we will focus on the SCCs of an optimal solution and study them separately, instead of studying the entire solution graph. Our main objective is to bound the cutwidth of SCCs in an optimal solution as a function of $q$ and of the number $k$ of endpoint pairs. We can then easily deduce like in~\cite{FeldmannM23} that the cutwidth of optimal solutions is bounded as a function of $q$ and~$k$. To compute the solutions of bounded cutwidth, we then modify slightly the algorithm of \cref{sec:cw-algo}.

To study the SCCs of optimal solutions, it will be useful to introduce an analogue of the \scss problem featuring modularities. Specifically, we 
call \kscssm{$k$} the problem that takes 
the same input as \kdsnm{$k$}
and returns a smallest \emph{strongly connected} subgraph $H$ of $G$ such that $H$ contains a walk from $s_i$ to $t_i$ of length $r_i\bmod q_i$ for every $i\leq k$.
Note that, unlike
the \scss problem which was simply defined by a set of terminals, in
\kscssm{$k$} we specify a set of tuples connecting pairs of vertices,
because we need to specify which remainder must be achieved by which endpoint pair. 
The difference with \kdsnm{$k$} is only that the solution graph is required to be
strongly connected.
Our study of \kscssm{$k$} is motivated by the fact that the SCCs of optimal
solutions to \kdsnm{$k$} are optimal solutions to instances of \kscssm{$k$}, as was
already known in the setting without modularities~\cite{FeldmanR06}. Namely:

\begin{claim}
\label{clm:sccs-optimal}
    Let $G = (V,E)$ be a directed graph, and $\mathcal{R}$ be a $k$-requirement
    specification. In any optimal solution $H = (V, E')$ to \kdsnm{$k$} on $G$
    for~$\mathcal{R}$, for any SCC $C$ of $H$, there exists a $k$-requirement
    specification $\mathcal{R}'$ such that $C$ (as a set of edges) is an optimal solution to \kscssm{$k$} on~$G$ for~$\mathcal{R'}$.
\end{claim}

\begin{proof}
    Fix the solution $H$ and pick an SCC $C$ of~$H$. We can cover $H$ with the $k$ walks $w_1, \ldots, w_k$ that witness that the requirements of~$\mathcal{R}$ are obeyed. By definition of an SCC, for each walk $w_i$, either $w_i$ does not enter $C$ at all, or it enters $C$ at some vertex $s_i'$ and stays in $C$ until it 
    leaves $C$ at some vertex $t_i'$ (or ends at~$t_i'$ altogether). The walk
    then never returns to~$C$ afterwards.

    Thus, let $\mathcal{R'}$ be defined in the following way: for each $i$ such that $s_i'$ and $t_i'$ are defined, we add a tuple $(s_i', t_i', r_i', q_i')$ where $q_i'$ is the corresponding value in~$\mathcal{R}$ and $r_i'$ is the remainder modulo $q_i'$ of the subwalk of~$w_i$ between $s_i'$ and $t_i'$. We pad $\mathcal{R}'$ by duplicating some tuples if necessary to ensure that it is a $k$-requirement specification.

    We claim that $C$ is an optimal solution to the \kscssm{$k$} problem on
    $G$ and~$\mathcal{R}'$. Indeed, we know by construction that it is a
    candidate solution. Now, assume by contradiction that there is another candidate solution $C'$ to \kscssm{$k$} on~$G$ for~$\mathcal{R}'$ which
has  fewer edges than~$C$.
    Then let $H'$ be the subgraph obtained from~$H$ by removing the edges of~$C$ and adding the edges of~$C'$. We know that $H'$ has fewer edges than $H$, and it is still a candidate solution to \kdsnm{$k$} on~$G$ for~$\mathcal{R}$: the witnessing walks $w_i$ only interacted with $C$ between the $s_i'$ and $t_i'$, and there are walks in~$C'$ achieving the same remainders. Thus, $H'$ is a candidate solution with fewer edges  than $H$; this contradicts the fact that $H$ is an optimal solution. We have established by contradiction that there is no candidate solution having fewer edges than~$C$, hence $C$ is an optimal solution to \kscssm{$k$}. This concludes the proof.
\end{proof}

We next bound the cutwidth of optimal solutions to the \kscssm{$k$} problem specifically. Then we will deduce a cutwidth bound on optimal solutions to \kdsnm{$k$} using a result of~\cite{FeldmannM23}. Finally, we will show an algorithm to compute these solutions.

\myparagraph{Bounding the cutwidth of solutions to \kscssm{$k$}}
Let us start by showing that solutions to \kscssm{$k$} have a small segment number and hence a small cutwidth:
\begin{lemma}\label{lem:scssm-segment-cutwidth}
Let $G$ be a graph and $\mathcal{R}$ be a $k$-requirement specification, denote by
  $q_1, \dots, q_k$ the respective modularities of its tuples, and denote by 
  $q$ the least common multiple (LCM) of $q_1, \dots, q_k$.
Then every optimal solution of the \kscssm{$k$} problem on~$G$ and~$\mathcal{R}$
can be covered with a walk of at most $O(k + \log q)$ segments and therefore has cutwidth at most $O(k + \log q)$.
\end{lemma}
Note that this claim only works for the \kscssm{$k$} problem, not the \kdsnm{$k$} problem, because it assumes the solution can be covered by a single walk which is not true in general for \kdsnm{$k$}.

Let us show the result:

\begin{proof}[Proof of \cref{lem:scssm-segment-cutwidth}]
At a high level, the proof of this result can be decomposed into several steps. In step~1, we define the notion of a \emph{legal covering walk}, which is a walk that covers the edges of a solution in a prescribed order: first visit all terminals to witness that they are strongly connected in a subwalk $w_0$ (similarly to \cref{lem:scss-direct}), then visit all paths with the prescribed modularities in a subwalk $w'$. In step~2, we carefully impose a variant of timestamp-minimality on legal covering walks towards bounding their number of segments. In step~3, we use \cref{lem:scss-direct} (on the first part of the legal covering walk) and an argument adapted from \cref{lem:bounding-number-of-segments} (for the second part) to show that the number of segments is bounded. The point of splitting the walk in two is that $w_0$ visits enough edges to allow us to move to arbitrary endpoints and ensure strong connectedness: this allows us to replay arbitrary detours in any subwalk no matter where they occur, intuitively ``pooling''
the achievable differences $\Delta(w,\dots)$ between all the subwalks.
Let us now explain the three steps in order.

\medskip
\textbf{Step 1: Legal covering walks.}
Let $H$ be an optimal solution to the \kscssm{$k$} instance given by the graph $G=(V,E)$
and the $k$-requirement specification $\mathcal{R}$
formed of the tuples $(s_i', t_i', r_i', q_i)$ for $1 \leq i \leq k$.
We write for convenience the set $T := \{s_1', \ldots, s_k', t_1', \ldots, t_k'\}$.
A \emph{legal covering walk} for $\mathcal{R}$ and~$H$ is an $H$-walk $w$ that  can be decomposed as $w_0 w_1 w_1' \cdots w_{k-1} w_{k-1}' w_k$ where:
\begin{itemize}
    \item The walk $w_0$ starts at $s_1'$, ends at $s_1'$, and in between it visits all vertices of $T \setminus \{s_1'\}$ in some arbitrary order,
      i.e., it covers a graph $H_0$ which would be a solution to the \kscss{$2k$} instance on $G$ with $T$ without modularities. We further require that the subgraph $H_0$ of~$H$ thus covered is edge-minimum \emph{among subgraphs of~$H$} -- note that $H_0$ is in particular edge-minimal, but $H_0$ may not be edge-minimum overall.
    \item Each walk $w_i$ for $1 \leq i \leq k$ goes from $s_i'$ to $t_i'$ and has length $r_i' \bmod q_i$, i.e., the walks $w_i$ witness that we have a solution to the \kdsnm{$k$} instance.
    \item Each walk $w_i'$ for $1 \leq i < k$ goes from $t_i'$ to $s_{i+1}'$ with no requirement on its length.
\end{itemize}

\medskip
\textbf{Step 2: Timestamp-minimality.}
The optimal solution $H$ is edge-minimum, so any legal covering walk $w$
  for~$H$ must cover all edges of~$H$; otherwise the strict subset of edges that it uses would be a solution to \kscssm{$k$} using fewer edges than $H$.
Therefore, we have $G_w = H$, so we know by \cref{lem:segbound} that
the cutwidth of $H$ is at most $3$ times the number of segments of a legal covering walk $w$.

We will now choose $w$ more carefully, imposing timestamp-minimality in a specific way, so as to bound the number of segments.
Let us consider all walks $w_0$ which start and end at $s_1'$ and visit all other terminals of~$T$ between the two endpoints,
and do so while using an edge-minimum subset $H_0$ of the solution $H$ that we are considering.
Now, let us choose one $w_0$ which is timestamp-minimum among all such
  $H_0$-walks.
  (Note that, because $H_0$ is edge-minimum,
  all walks satisfying the requirements while using only edges of~$H_0$ are
  actually using all edges of~$H_0$, i.e., they are $H_0$-walks.)

  We then define a \emph{$w_0$-legal covering walk} for $\mathcal{R}$ and~$H$ as a walk of the form defined in step 1 but for the specific choice of~$w_0$ that we now made, i.e., when decomposing the walk according to the definition of legal covering walks, the first walk in the decomposition must be $w_0$ exactly.
  Recall that
  any $w_0$-legal covering walk $w$ for~$H$
  is an $H$-walk.
  Let us pick an $H$-walk $w$ which is timestamp-minimal among the $w_0$-legal
  covering walks for~$H$.
(Note that $w$ is not necessarily timestamp-minimal among all $H$-walks,
or among all legal covering walks for~$H$, because of the specific structure that we impose on it.)
This choice of~$w$ is the walk for which we will show a bound on the number of segments; as $w$ covers all edges of~$H$ this implies the desired bound on the cutwidth of~$H$ by the Segment Cutwidth Bound (\cref{prp:segment-bound}).

\medskip
\textbf{Step 3: Segments and timestamps in a composite walk.}
We now wish to bound the number of segments of the walk $w$, for the segment decomposition defined in \cref{sec:segment}.
First, by the definition of $w_0$ (timestamp-minimum among edge-minimum
$H_0$-walks that start and end on a same fixed terminal and visit the $2k-1$ other terminals) and by \cref{lem:scss-direct},
we know that there are at most $2k$ segments of~$w$ which start within $w_0$.
Hence, it suffices to show that the number of segments which start after $w_0$ in~$w$ is $O(k+\log_2 q)$, which we do in the rest of the proof.
Recalling that $w = w_0 w_1 w_1' \cdots w_{k-1} w_{k-1}' w_k$, 
a \emph{boundary} is a position between two subwalks in the decomposition above: note that the vertex visited at that position is always a terminal of~$T$.
We then say that a segment is \emph{special} if it starts within $w_0$, or if it walks over a boundary, or if it is the last segment. By what precedes for segments starting within~$w_0$, or by definition otherwise, there are $O(k)$ special segments. Hence, in the sequel, we show that there are $O(\log_2 q)$ non-special segments in~$w$.

Let $s_1,\dots, s_\xi$ be the segment decomposition of the walk $w$, including both the special and the non-special segments.
(The reader should be careful not to confuse the segments $s_1, \ldots, s_\xi$ segments with the terminals $s_1', \ldots, s_k'$.)
We now re-use notation from \cref{sec:segment}, in particular the notion of achievable differences.
For each segment $s_\sigma$, we let $S(w,\sigma)$ denote the subgroup of $\Z{q}$ generated by $\{\Delta(w, s_{\sigma'})\mid \sigma'\leq \sigma\}$, and for convenience we write $S(w,0)$ to denote singleton group with just the identity $0_{\Z{q}}$.
By definition, the $S(w,\sigma)$ form a sequence of increasing subgroups of~$\Z{q}$.
Our goal is to show that 
for each pair of consecutive non-special segments $s_{\sigma-1}$ and $s_{\sigma}$, the set $S(w,\sigma)$ is a strict superset of $S(w,\sigma-1)$.
We do so by a variant of the proof of \cref{lem:bounding-number-of-segments} which we spell out below.

We first show the analogue of invariant (*), dubbed (*'): for each index $1 \leq \sigma < \xi$ such that $s_\sigma$ is non-special,
letting $\sigma' < \sigma$ be maximal such that $s_{\sigma'}$ is special,
splitting $s_{\sigma'} = s_- s_+$ at the rightmost boundary that $s_\sigma'$ contains,
letting $y_\sigma$ be the vertex on which segment $s_\sigma$ ends, for each $r' \in S(w, \sigma)$, there is a walk from
the first terminal~$s_1'$ 
to~$y_\sigma$ which 
uses only edges of $s_1 \cdots s_\sigma$,
has length precisely $|s_1 \cdots s_\sigma|+r' \bmod q$, and starts with~$s_1 \cdots s_{\sigma'} s_-$.

Informally, invariant (*') above allows us to rewrite the walk at a non-special segment so as to adjust the length of the walk by any remainder in $S(w, \sigma)$, while ensuring that the set of traversed edges is the same and that 
walk is only changed after the most recently traversed boundary.
(This ensures that the resulting walk can still be decomposed as a $w_0$-legal
covering walk, in particular the resulting walk still starts with $w_0$,
although like in \cref{lem:bounding-number-of-segments} we do not know the
segment decomposition of the resulting walk.) 
In the proof we will use the edges visited by~$w_0$ as a way to guarantee that the detours necessary to achieve the remainders in~$S(w, \sigma)$ 
can be ``inserted'' precisely at the most recently traversed boundary. 
(Note that this proof is slightly different from that of \cref{lem:bounding-number-of-segments}, where we needed to insert detours at multiple places and change the beginning of the walk entirely.)

Let us now establish invariant (*').
Fix $\sigma$, $y_\sigma$, $r' \in S(w, \sigma)$, $\sigma'$, and $s_{\sigma'} = s_- s_+$.
By definition of $S(w, \sigma)$ we can write $r'= \sum_{1 \leq \sigma'' \leq
\sigma} c_{\sigma''} \Delta(w, \sigma'')$ with $c_{\sigma''} \in \mathbb{N}$ for each $1 \leq \sigma'' \leq \sigma$.
We will modify the walk to add each successive $c_{\sigma''} \Delta(w, \sigma'') \bmod q$ while satisfying the conditions.
In the modifications we will only insert subwalks that restrict themselves to already-visited edges, so that the set of visited edges by the resulting walk is correct.
Thus, let us pick one $\sigma''$.
Using the notations for detours introduced in \cref{lem:segment-contains-detour}, we know that $\Delta(w, \sigma'') = |\text{det}_{\sigma''}| - |p_{\sigma''}|$, and we distinguish are two cases, depending on whether $\sigma'' \geq \sigma'$ or not.

Case 1: if $\sigma'' \geq \sigma'$. Then, recalling the walks $p_{\sigma''}$ and
$\text{det}_{\sigma''}$ and $\overline{p_{\sigma''}}$ from
\cref{lem:segment-contains-detour}, we can use the argument from
\cref{lem:bounding-number-of-segments} to insert $(p_{\sigma''}
\overline{p_{\sigma''}})^q$ at the end of segment $\sigma''$  and replace
$c_{\sigma''}$ copies of $p_{\sigma''}$ by $\text{det}_{\sigma''}$ to modify the
remainder.
This change ensures that the resulting walk still starts with~$s_1 \cdots s_{\sigma'} s_-$.

Case 2: if $\sigma'' < \sigma'$. Then, doing the change as in case 1 would
modify the walk before the latest traversed boundary.
Instead, we will change the walk between $s_-$ and $s_+$, using the presence of
$w_0$ to navigate back to the detour that we wish to use. More precisely,
between $s_-$ and $s_+$ we are at a vertex $v$ from $T$ by definition of a
boundary. Let $w''$ be any subwalk of $w_0$ that navigates from $v$ to $s_1'$,
and let $w'''$ be a prefix of $s_1 \cdots s_{\sigma'-1} s_-$
that navigates from $s_1'$ to the vertex $u$ which starts $\text{det}_{\sigma''}$ and $p_{\sigma''}$. Likewise,
let $\overline{w'''}$ be a subwalk of~$s_1 \cdots s_{\sigma'-1} s_-$ that navigates
from the end of $s{\sigma''}$ to~$v$: here we use the fact that $\sigma'' <
\sigma'$, so the end of~$s{\sigma''}$ is before the occurrence of~$v$ between
$s_-$ and $s_+$ and $\overline{w'''}$ exists. We simply insert 
$(w'' w''' p_{\sigma''} \overline{p_{\sigma''}} \overline{w'''})^q$ 
between $s_-$ and $s_+$: this gives a legal walk (i.e., the concatenation is
possible), the length of the resulting walk modulo $q$ is the same as that
of~$w$ modulo~$q$ 
because the length of the inserted subwalk is a multiple of~$q$, further the
resulting walk 
still starts with~$s_1 \cdots s_{\sigma'} s_-$.
Then, we
finish as in Case~1 by replacing $c_{\sigma''}$ occurrences of $p_{\sigma''}$ by $\text{det}_{\sigma''}$
to modify the remainder. Thus, we have established invariant~(*').

Now, assume by way of contradiction that there is 
a pair of two consecutive non-special segments $s_{\sigma-1}$ and $s_\sigma$
such that $S(w,\sigma) = S(w,\sigma-1)$, meaning that $\Delta(w,
\sigma)$ is in $S(w,\sigma-1)$. We use invariant~(*') to rewrite $w$ to a different
walk that witnesses a contradiction of timestamp-minimality. More precisely, as in
\cref{lem:bounding-number-of-segments}, we first replace $\text{det}_\sigma$ by the
existing path $p_\sigma$, changing the walk length by $-\Delta(w,\sigma) \bmod
q$, and we then use invariant~(*') on the non-special segment~$s_{\sigma-1}$ to compensate this change thanks to the fact that
$\Delta(w, \sigma) \in S(w,\sigma-1)$. We know that the resulting walk is
still an $H$-walk because $H$ is an optimal solution, and it is not
greater than or equal to~$w$ in the timestamp preorder for the same reasons as in
\cref{lem:bounding-number-of-segments}: replacing $\text{det}_\sigma$ by
$p_\sigma$ pushes back towards the end the first visit of at least one edge,
without changing the moment at which the other first-visited edges that are
closer to the end of the segment are first-visited. (This uses the fact that, in
the statement of invariant (*'), the rewritten walk is required to visit the same
edges as the original walk.) So we have a contradiction if we can establish that the rewritten walk is still a $w_0$-legal covering walk (because the timestamp-minimality of~$w$ was asserted among such walks only).

To see this, notice that invariant~(*') ensures that the walk is unchanged to the
left of the latest boundary before $s_\sigma$. Hence, the walk still starts
with~$w_0$. What is more, the walk still contains witnessing subwalks like the
$w_i$ and $w_i'$. This is because, to the left of the latest boundary
before~$s_\sigma$, the walk is unchanged, so the remainders of the 
lengths of the subwalks left of this boundary are unchanged. As for the witnessing
subwalks right of this boundary, for the next subwalk we have ensured that the
change in length caused by replacing $\text{det}_\sigma$ by $p_\sigma$ is
compensated by a modification that happens within the same subwalk thanks to
invariant~(*') 
and thanks to the fact that $s_\sigma$ and $s_{\sigma-1}$ are
non-special. The next subwalks are also unchanged. Hence indeed the rewritten
walk is a $w_0$-legal covering walk and we have obtained our contradiction.
Thus, we have shown that consecutive pairs of non-special segments
$s_{\sigma-1}$ and $s_\sigma$ must ensure that $S(w,\sigma) \supsetneq S(w,\sigma-1)$.

The end of the proof is exactly like \cref{lem:bounding-number-of-segments}. We have proved that the subgroups of non-special segments must form a strictly growing sequence for containment; only for special segments $\sigma$ may it be the case that $S(w, \sigma) = S(w, \sigma-1)$. Hence, by Lagrange's theorem, this implies there are no more than $\log_2 q$ non-special segments in $w_1$, which is what remained to be established.
\end{proof}

\myparagraph{From \kscssm{$k$} to \kdsnm{$k$}.}
We have shown that optimal solutions to the \kscssm{$k$} problem can be assumed to have bounded cutwidth. We now wish to show that the same is true for optimal solutions to the \kdsnm{$k$} problem. This is simple, using 
\cref{lem:scssm-segment-cutwidth} above along with
two lemmas from~\cite{FeldmannM23}:

\begin{lemma}
\label{lem:scss-to-dsn}
Fix a graph $G$ and a $k$-requirement specification $\mathcal{R}$, and let $q$ be the LCM of the $q_i$ in~$\mathcal{R}$.
Let $H$ be an optimal solution to \kdsnm{$k$} on~$G$ and~$\mathcal{R}$.
Then $H$ has cutwidth at most $O(k + \log q)$.
\end{lemma}

\begin{proof}
Fix $G$ and $\mathcal{R}$ and the optimal solution~$H$.
Remember that we can cover $H$ by the $k$ witnessing paths of the \kdsnm{$k$} instance (which covers all edges of~$H$ because $H$ is optimal).
  Thus, let us consider the \emph{condensation graph} of~$H$ in the terminology
  of~\cite{FeldmannM23}, which is the directed acyclic multi-graph obtained by
  contracting all SCCs of~$H$ (removing self-loops, but without removing
  parallel edges). This graph is the union of~$k$ paths, so by 
  Lemma~2.1 of~\cite{FeldmannM23} it has cutwidth at most $O(k)$. Further, 
  we have shown in \cref{clm:sccs-optimal}
  that every SCC of~$H$ is an optimal solution to a \kscssm{$k$}-instance, so by
  \cref{lem:scssm-segment-cutwidth} it has cutwidth $O(k + \log q)$. Now,
  Lemma~2.2 of~\cite{FeldmannM23} allows us to bound the cutwidth of a graph as
  a function of the sum of the cutwidth of its condensation graph and of the
  maximal cutwidth of its SCCs. This gives an $O(k + \log q)$ bound on the
  cutwidth of~$H$, which concludes the proof.
\end{proof}

\myparagraph{Computing optimal solutions.}
We have shown in \cref{lem:scss-to-dsn} that optimal solutions to the \kdsnm{$k$} problem have cutwidth $O(k + \log q)$. All that remains to prove \cref{th:kdsnm-tractable} is to adapt the dynamic algorithm from \cref{sec:cw-algo} to compute solutions to that problem instead of \ewm. As the algorithm is an easy variant of the \ewm algorithm, we only sketch it.

\begin{proof}[Proof of~\cref{th:kdsnm-tractable}]
Given a graph $G$ and given a $k$-requirement specification $\mathcal{R}$, 
writing $(s_i, t_i, r_i, q_i)$ for $1 \leq i \leq k$ the tuples of~$\mathcal{R}$,
letting $q$ be the LCM of the $q_i$,
writing for convenience $T := \{s_1, \ldots, s_k, t_1, \ldots, t_k\}$,
fixing a domain size bound $\omega$ to be specified later,
we build the same configuration graph $\Gamma_{\omega,q}$ as in \cref{sec:cw-algo} up to the following differences which we specify next.

  The first change is in the definition of final configurations: instead of saying that a configuration $(D, \rho)$ is final if $s, t \in D$ and $r \in \rho(s, t)$, we now call $(D, \rho)$ \emph{final} if for each $1 \leq i \leq k$ we have $s_i \in D$ and $t_i \in D$ and some $r_i' \in \rho(s_i, t_i)$ such that $r_i' \equiv r_i \bmod q_i$ (remember that the algorithm keeps track of all remainders modulo the LCM $q$).

The second change is the definition of the domain size bound $\omega$: instead of $(2 + 1) + \omega'$ for $\omega'$ the cutwidth bound, we pick $(2k + 1) + \omega'$, which here amounts to $O(k + \log q)$ because $\omega'$ is $O(k + \log q)$ by \cref{lem:scss-to-dsn}. The intuition is that we want to make sure that the $2k$ terminal vertices of~$T$ can always fit in the domain of configurations in addition to other vertices.

The last change is the completeness proof: when building the witnessing path corresponding to a subgraph, instead of saying that vertices $s$ and $t$ are handled in a special way and are never forgotten, now we do the same for all vertices of~$T$. The rest of the correctness proof is unchanged.

The complexity of the algorithm as a function of $\omega$ is the same, and plugging $\omega = O(k + \log q)$ gives the stated complexity, concluding the proof.
\end{proof}

\section{Future Work}

We now outline possible directions for future work in light of our results:

\myparagraph{Improving our bounds.} One natural direction for further
research would be to improve the complexity bounds that we show. Our algorithm
for \ewm runs in time $n^{O(\log q)} \cdot 2^{O(q \log^2 q)}$. Can the factor of
$q$ in the exponent be improved, e.g., by getting an algorithm in time $n^{O(\log q)}$? 
Or, on the other hand, is the problem \NP-hard? (We have only shown that the
variant of \ewm with edge lengths is \NP-hard, in \cref{prp:npc}.)

\myparagraph{Intermediate problems between shortest walk and edge-minimum walk.} It could be interesting to search for walks that are minimized according to some criterion that is 
``intermediate'' between shortest walk and edge-minimum walk, e.g., if the cost
of each edge is expressed as a function of how many times it is traversed by the
walk. Such a general framework would capture the two problem variants that we
contrast in the introduction: shortest walks are the case where we are charged
$\ell$ when traversing an edge $\ell$ times, and edge-minimum walks are
the case where we are charged $1$ when traversing an edge $\ell>0$ times and
charged $0$ when we do not traverse it.

\myparagraph{Edge-minimum walks satisfying other constraints.} 
Last, one other problem of interest would be to investigate the complexity of finding edge-minimum subgraphs guaranteeing the existence of $st$-walks satisfying other properties.
One very natural example is the following: Given a number $\ell$, and a directed graph $G$ with specified vertices $s,t$, find an edge-minimum $st$-walk of length exactly $\ell$. (This is the same as \ewm except the length is exactly $\ell$, instead of $r\bmod q$.) The trivial algorithm for this problem is in time $O(n^\ell)$, but can we do better? To our knowledge, this problem has not been studied.

Another example is looking for edge-minimum walks that achieve constraints
expressed by a finite semigroup. For instance, assume that each edge of the
graph is labeled by a semigroup element, and that we want a walk whose
evaluation in the semigroup achieves a specific target element of the
semigroup, where evaluating the walk means multiplying the 
labels of its edges in the order that they are traversed. This problem
generalizes the \ewm problem (with lengths on edges), which uses the semigroup
$\Z{q}$. An alternative way to phrase this problem is in the language of regular
path queries (RPQs) mentioned in the introduction: fixing a regular language $L$
on an alphabet~$\Sigma$, and given a directed graph $G$ with terminals $s$ and
$t$ and with edges labeled by letters of~$\Sigma$, find an edge-minimum subgraph
$G$ with an $st$-walk which evaluates to a word that belongs to~$L$.
For which fixed regular languages $L$ can this problem be solved in polynomial
time in~$G$?

\myparagraph{Acknowledgements.}
We are grateful to the reviewers of the conference version for
their helpful feedback. We also would like to thank the Simons Institute Fall 2023 programs ``Logic and Algorithms in Database Theory and AI'' and ``Data Structures and Optimization for Fast Algorithms'' for the initiation of this work.
Last, we are grateful to Xiao Hu and Mikaël Monet for early discussions.

\bibliography{references}

\newcommand{\etalchar}[1]{$^{#1}$}
\begin{thebibliography}{JKMC{\etalchar{+}}24}

\bibitem[AK21]{MR4371468}
Noga Alon and Michael Krivelevich.
\newblock Divisible subdivisions.
\newblock {\em J. Graph Theory}, 98(4), 2021.
\newblock \href {https://doi.org/10.1002/jgt.22716}
  {\path{doi:10.1002/jgt.22716}}.

\bibitem[Ama24]{amarilli2024survey}
Antoine Amarilli.
\newblock Survey of results on the {ModPath} and {ModCycle} problems, 2024.
\newblock \href {https://arxiv.org/abs/2409.00770} {\path{arXiv:2409.00770}}.

\bibitem[APY91]{ArkinPY91}
Esther~M. Arkin, Christos~H. Papadimitriou, and Mihalis Yannakakis.
\newblock Modularity of cycles and paths in graphs.
\newblock {\em J. {ACM}}, 38(2), 1991.
\newblock \href {https://doi.org/10.1145/103516.103517}
  {\path{doi:10.1145/103516.103517}}.

\bibitem[BBG20]{BaganBG20}
Guillaume Bagan, Angela Bonifati, and Beno{\^{\i}}t Groz.
\newblock A trichotomy for regular simple path queries on graphs.
\newblock {\em J. Comput. Syst. Sci.}, 108, 2020.
\newblock URL: \url{https://doi.org/10.1016/j.jcss.2019.08.006}, \href
  {https://doi.org/10.1016/J.JCSS.2019.08.006}
  {\path{doi:10.1016/J.JCSS.2019.08.006}}.

\bibitem[BBR97]{basagni1997difficulty}
Stefano Basagni, Danilo Bruschi, and F~Ravasio.
\newblock On the difficulty of finding walks of length $k$.
\newblock {\em RAIRO-Theoretical Informatics and Applications}, 31(5), 1997.

\bibitem[BHK22]{BjorklundHK22}
Andreas Bj{\"{o}}rklund, Thore Husfeldt, and Petteri Kaski.
\newblock The shortest even cycle problem is tractable.
\newblock In {\em STOC}, 2022.
\newblock \href {https://doi.org/10.1145/3519935.3520030}
  {\path{doi:10.1145/3519935.3520030}}.

\bibitem[CDGS24]{chauhan2024even}
Archit Chauhan, Samir Datta, Chetan Gupta, and Vimal~Raj Sharma.
\newblock The even-path problem in directed single-crossing-minor-free graphs,
  2024.
\newblock \href {https://arxiv.org/abs/2407.00237} {\path{arXiv:2407.00237}}.

\bibitem[CEGS11]{MR2786434}
Chandra Chekuri, Guy Even, Anupam Gupta, and Danny Segev.
\newblock Set connectivity problems in undirected graphs and the directed
  {S}teiner network problem.
\newblock {\em ACM Trans. Algorithms}, 7(2), 2011.
\newblock \href {https://doi.org/10.1145/1921659.1921664}
  {\path{doi:10.1145/1921659.1921664}}.

\bibitem[CFM21]{MR4278929}
Rajesh Chitnis, Andreas~Emil Feldmann, and Pasin Manurangsi.
\newblock Parameterized approximation algorithms for bidirected {S}teiner
  network problems.
\newblock {\em {ACM} Trans. Algorithms}, 17(2), 2021.
\newblock \href {https://doi.org/10.1145/3447584} {\path{doi:10.1145/3447584}}.

\bibitem[CGK94]{MR1285584}
Fan R.~K. Chung, Wayne Goddard, and Daniel~J. Kleitman.
\newblock Even cycles in directed graphs.
\newblock {\em {SIAM} J. Discret. Math.}, 7(3), 1994.
\newblock \href {https://doi.org/10.1137/S0895480192225433}
  {\path{doi:10.1137/S0895480192225433}}.

\bibitem[CM90]{ConsensM90}
Mariano~P. Consens and Alberto~O. Mendelzon.
\newblock Graphlog: a visual formalism for real life recursion.
\newblock In {\em PODS}, 1990.
\newblock \href {https://doi.org/10.1145/298514.298591}
  {\path{doi:10.1145/298514.298591}}.

\bibitem[CMW87]{CruzMW87}
Isabel~F. Cruz, Alberto~O. Mendelzon, and Peter~T. Wood.
\newblock A graphical query language supporting recursion.
\newblock In {\em SIGMOD}, 1987.
\newblock \href {https://doi.org/10.1145/38713.38749}
  {\path{doi:10.1145/38713.38749}}.

\bibitem[DD24]{MR4786546}
Gianlorenzo D'Angelo and Esmaeil Delfaraz.
\newblock Approximation algorithms for node-weighted directed {S}teiner
  problems.
\newblock In {\em IWOCA}, 2024.
\newblock \href {https://doi.org/10.1007/978-3-031-63021-7\_21}
  {\path{doi:10.1007/978-3-031-63021-7\_21}}.

\bibitem[DDS24]{MR4668327}
Shagnik Das, Nemanja Dragani\'{c}, and Raphael Steiner.
\newblock Tight bounds for divisible subdivisions.
\newblock {\em J. Combin. Theory Ser. B}, 165, 2024.
\newblock \href {https://doi.org/10.1016/j.jctb.2023.10.011}
  {\path{doi:10.1016/j.jctb.2023.10.011}}.

\bibitem[Diw24]{diwan2024cycles}
Ajit~A Diwan.
\newblock Cycles of weight divisible by $ k$, 2024.
\newblock \href {https://arxiv.org/abs/2407.01198} {\path{arXiv:2407.01198}}.

\bibitem[DM18]{DinurM18}
Irit Dinur and Pasin Manurangsi.
\newblock {ETH}-hardness of approximating 2-{CSP}s and directed {S}teiner
  network.
\newblock In {\em ITCS}, 2018.
\newblock URL: \url{https://doi.org/10.4230/LIPIcs.ITCS.2018.36}, \href
  {https://doi.org/10.4230/LIPICS.ITCS.2018.36}
  {\path{doi:10.4230/LIPICS.ITCS.2018.36}}.

\bibitem[FHW80]{FortuneHW80}
Steven Fortune, John~E. Hopcroft, and James Wyllie.
\newblock The directed subgraph homeomorphism problem.
\newblock {\em Theor. Comput. Sci.}, 10, 1980.
\newblock \href {https://doi.org/10.1016/0304-3975(80)90009-2}
  {\path{doi:10.1016/0304-3975(80)90009-2}}.

\bibitem[FM23]{FeldmannM23}
Andreas~Emil Feldmann and D\'{a}niel Marx.
\newblock The complexity landscape of fixed-parameter directed {S}teiner
  network problems.
\newblock {\em ACM Trans. Comput. Theory}, 15(3–4), 2023.
\newblock \href {https://doi.org/10.1145/3580376} {\path{doi:10.1145/3580376}}.

\bibitem[FR06]{FeldmanR06}
Jon Feldman and Matthias Ruhl.
\newblock The directed {S}teiner network problem is tractable for a constant
  number of terminals.
\newblock {\em {SIAM} J. Comput.}, 36(2), 2006.
\newblock \href {https://doi.org/10.1137/S0097539704441241}
  {\path{doi:10.1137/S0097539704441241}}.

\bibitem[GKMS24]{GalbyKMS24}
Esther Galby, S{\'{a}}ndor Kisfaludi{-}Bak, D{\'{a}}niel Marx, and Roohani
  Sharma.
\newblock Subexponential parameterized directed {S}teiner network problems on
  planar graphs: {A} complete classification.
\newblock In {\em ICALP}, 2024.
\newblock URL: \url{https://doi.org/10.4230/LIPIcs.ICALP.2024.67}, \href
  {https://doi.org/10.4230/LIPICS.ICALP.2024.67}
  {\path{doi:10.4230/LIPICS.ICALP.2024.67}}.

\bibitem[GNS11]{GuoNS11}
Jiong Guo, Rolf Niedermeier, and Ondrej Such{\'{y}}.
\newblock Parameterized complexity of arc-weighted directed {S}teiner problems.
\newblock {\em {SIAM} J. Discret. Math.}, 25(2), 2011.
\newblock \href {https://doi.org/10.1137/100794560}
  {\path{doi:10.1137/100794560}}.

\bibitem[HS24]{hu2024finding}
Xiao Hu and Stavros Sintos.
\newblock Finding smallest witnesses for conjunctive queries.
\newblock In {\em ICDT}, 2024.
\newblock URL: \url{https://doi.org/10.4230/LIPIcs.ICDT.2024.24}, \href
  {https://doi.org/10.4230/LIPICS.ICDT.2024.24}
  {\path{doi:10.4230/LIPICS.ICDT.2024.24}}.

\bibitem[JKMC{\etalchar{+}}24]{MR4756860}
Alp\'{a}r J\"{u}ttner, Csaba Kir\'{a}ly, Lydia~Mirabel Mendoza-Cadena, Gyula
  Pap, Ildik\'{o} Schlotter, and Yutaro Yamaguchi.
\newblock Shortest odd paths in undirected graphs with conservative weight
  functions.
\newblock {\em Discrete Appl. Math.}, 357, 2024.
\newblock \href {https://doi.org/10.1016/j.dam.2024.05.044}
  {\path{doi:10.1016/j.dam.2024.05.044}}.

\bibitem[Kar72]{karp1972reducibility}
Richard~M. Karp.
\newblock Reducibility among combinatorial problems.
\newblock In {\em Complexity of Computer Computations}. Plenum Press, New York,
  1972.
\newblock \href {https://doi.org/10.1007/978-1-4684-2001-2\_9}
  {\path{doi:10.1007/978-1-4684-2001-2\_9}}.

\bibitem[KK12]{MR2946427}
Naonori Kakimura and Ken-ichi Kawarabayashi.
\newblock Packing cycles through prescribed vertices under modularity
  constraints.
\newblock {\em Adv. in Appl. Math.}, 49(2), 2012.
\newblock \href {https://doi.org/10.1016/j.aam.2012.03.002}
  {\path{doi:10.1016/j.aam.2012.03.002}}.

\bibitem[KKKX23]{MR4538070}
Ken-ichi Kawarabayashi, Stephan Kreutzer, O-joung Kwon, and Qiqin Xie.
\newblock A half-integral {E}rdos-{P}osa theorem for directed odd cycles.
\newblock In {\em SODA}, 2023.
\newblock \href {https://doi.org/10.1137/1.9781611977554.ch118}
  {\path{doi:10.1137/1.9781611977554.ch118}}.

\bibitem[LP84]{LapaughP84}
Andrea~S. LaPaugh and Christos~H. Papadimitriou.
\newblock The even-path problem for graphs and digraphs.
\newblock {\em Networks}, 14(4), 1984.
\newblock URL: \url{https://doi.org/10.1002/net.3230140403}, \href
  {https://doi.org/10.1002/NET.3230140403} {\path{doi:10.1002/NET.3230140403}}.

\bibitem[Lub88]{MR971619}
Anna Lubiw.
\newblock A note on odd/even cycles.
\newblock {\em Discret. Appl. Math.}, 22(1), 1988.
\newblock \href {https://doi.org/10.1016/0166-218X(88)90125-4}
  {\path{doi:10.1016/0166-218X(88)90125-4}}.

\bibitem[McC04]{mccuaig2004polya}
William McCuaig.
\newblock P{\'{o}}lya's permanent problem.
\newblock {\em Electron. J. Comb.}, 11(1), 2004.
\newblock \href {https://doi.org/10.37236/1832} {\path{doi:10.37236/1832}}.

\bibitem[MNP23]{MartensNP23}
Wim Martens, Matthias Niewerth, and Tina Popp.
\newblock A trichotomy for regular trail queries.
\newblock {\em Log. Methods Comput. Sci.}, 19(4), 2023.
\newblock URL: \url{https://doi.org/10.46298/lmcs-19(4:20)2023}, \href
  {https://doi.org/10.46298/LMCS-19(4:20)2023}
  {\path{doi:10.46298/LMCS-19(4:20)2023}}.

\bibitem[Mon83]{MR727545}
Burkhard Monien.
\newblock The complexity of determining a shortest cycle of even length.
\newblock {\em Computing}, 31(4), 1983.
\newblock \href {https://doi.org/10.1007/BF02251238}
  {\path{doi:10.1007/BF02251238}}.

\bibitem[MRY19]{miao2019explaining}
Zhengjie Miao, Sudeepa Roy, and Jun Yang.
\newblock Explaining wrong queries using small examples.
\newblock In {\em SIGMOD}, 2019.
\newblock \href {https://doi.org/10.1145/3299869.3319866}
  {\path{doi:10.1145/3299869.3319866}}.

\bibitem[Ned99]{MR1744686}
Zhivko~Prodanov Nedev.
\newblock Finding an even simple path in a directed planar graph.
\newblock {\em SIAM J. Comput.}, 29(2), 1999.
\newblock \href {https://doi.org/10.1137/S0097539797330343}
  {\path{doi:10.1137/S0097539797330343}}.

\bibitem[{OEI}24]{oeis}
{OEIS Foundation Inc.}
\newblock The {O}n-{l}ine {E}ncyclopedia of {I}nteger {S}equences, 2024.
\newblock Published electronically at \url{http://oeis.org}.

\bibitem[RR18]{MR3875932}
Mehdy Roayaei and Mohammadreza Razzazi.
\newblock Parameterized complexity of directed {S}teiner network with respect
  to shared vertices and arcs.
\newblock {\em Int. J. Found. Comput. Sci.}, 29(7), 2018.
\newblock \href {https://doi.org/10.1142/S0129054118500302}
  {\path{doi:10.1142/S0129054118500302}}.

\bibitem[RST99]{MR1740989}
Neil Robertson, P.~D. Seymour, and Robin Thomas.
\newblock Permanents, {P}faffian orientations, and even directed circuits.
\newblock {\em Ann. of Math. (2)}, 150(3), 1999.
\newblock \href {https://doi.org/10.2307/121059} {\path{doi:10.2307/121059}}.

\bibitem[SS22]{schlotter2022odd}
Ildik{\'o} Schlotter and Andr{\'a}s Seb{\H{o}}.
\newblock Odd paths, cycles and $ t $-joins: Connections and algorithms, 2022.
\newblock \href {https://arxiv.org/abs/2211.12862} {\path{arXiv:2211.12862}}.

\bibitem[ST87]{MR872406}
Paul Seymour and Carsten Thomassen.
\newblock Characterization of even directed graphs.
\newblock {\em J. Combin. Theory Ser. B}, 42(1), 1987.
\newblock \href {https://doi.org/10.1016/0095-8956(87)90061-X}
  {\path{doi:10.1016/0095-8956(87)90061-X}}.

\bibitem[Tho85]{thomassen1985even}
Carsten Thomassen.
\newblock Even cycles in directed graphs.
\newblock {\em Eur. J. Comb.}, 6(1), 1985.
\newblock \href {https://doi.org/10.1016/S0195-6698(85)80025-1}
  {\path{doi:10.1016/S0195-6698(85)80025-1}}.

\bibitem[Tho93]{MR1218331}
Carsten Thomassen.
\newblock The even cycle problem for planar digraphs.
\newblock {\em J. Algorithms}, 15(1), 1993.
\newblock \href {https://doi.org/10.1006/jagm.1993.1030}
  {\path{doi:10.1006/jagm.1993.1030}}.

\bibitem[VY89]{VaziraniY89}
Vijay~V. Vazirani and Mihalis Yannakakis.
\newblock Pfaffian orientations, 0-1 permanents, and even cycles in directed
  graphs.
\newblock {\em Discret. Appl. Math.}, 25(1-2), 1989.
\newblock \href {https://doi.org/10.1016/0166-218X(89)90053-X}
  {\path{doi:10.1016/0166-218X(89)90053-X}}.

\bibitem[Wol11]{MR2847879}
Paul Wollan.
\newblock Packing cycles with modularity constraints.
\newblock {\em Combinatorica}, 31(1), 2011.
\newblock \href {https://doi.org/10.1007/s00493-011-2551-5}
  {\path{doi:10.1007/s00493-011-2551-5}}.

\bibitem[YZ97]{MR1445033}
Raphael Yuster and Uri Zwick.
\newblock Finding even cycles even faster.
\newblock {\em SIAM J. Discrete Math.}, 10(2), 1997.
\newblock \href {https://doi.org/10.1137/S0895480194274133}
  {\path{doi:10.1137/S0895480194274133}}.

\end{thebibliography}

\appendix
\section{Reduction from Undirected Graphs to Directed Graphs}
\label{apx:undirected}

In this section, we solve the \ewm problem (defined in \cref{sec:intro}),
and its generalizations \kscssm{$k$} and \kdsnm{$k$} for any $k \in \mathbb{N}$
(defined in \cref{sec:dsn}), in the setting of undirected graphs.
Specifically, we are given as input an undirected graph, together with terminals
$s,t$ and integers $0 \leq r < q$ (for the undirected analogues of \ewm) or with a $k$-requirement
specification (for the undirected analogues of \kscssm{$k$} and \kdsnm{$k$}).
We are looking for an undirected walk of the prescribed modularity (for
\ewm) or for walks of the prescribed modularities (for \kdsnm{$k$}) which
additionally must form a strongly connected graph (for \kscssm{$k$}), and
which the set of undirected edges used is subset-minimum. 
Like on directed graphs, to find the prescribed walks, it suffices to find the minimum-cardinality subgraph: given such a graph, we can then find the walks using the edges of the subgraph via a product construction. 

We observe that, on undirected graphs, the strong connectivity required by
\kscssm{$k$} becomes a requirement that all terminals belong to the same
connected component in the solution. Similarly, in solutions to \ewm on
undirected graphs, the source and target vertex must belong to the same connected
components, and in solutions to \kdsnm{$k$} each source and target that occur
in a same pair from the $k$-requirement must belong to the same connected component.
However, unlike the Steiner tree problem, solutions to these problems are not
necessarily trees or forests because of the modularity requirements that
must be obeyed (see
\cref{fig:undirected}).

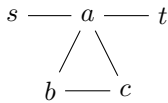
\begin{figure}[h!]
  \centering
  \begin{tikzpicture}
    \node (s) at (0, 0) {$s$};
    \node (a) at (1, 0) {$a$};
    \node (b) at (.5, -1) {$b$};
    \node (c) at (1.5, -1) {$c$};
    \node (t) at (2, 0) {$t$};
    \draw (s) -- (a);
    \draw (a) -- (t);
    \draw (a) -- (b);
    \draw (b) -- (c);
    \draw (c) -- (a);
  \end{tikzpicture}
  \caption{An undirected graph $G$. The only solution to \ewm with terminals $s$
  and $t$ and odd length (i.e., $q=2$ and $r=1$) is the entire graph.}
  \label{fig:undirected}
\end{figure}

\myparagraph{Reducing to modulo 2.}
As we pointed out in \cref{sec:intro}, on graphs where all edges are undirected,
unless the walk is trivial we can use any edge of the solution to move back
and forth and add any multiple of two to the walk length. Thus, for the problems
\ewm and \kscssm{$k$} and \kdsnm{$k$} on undirected graphs, we can assume that all
length constraints are either modulo~1 (i.e., no requirement on the
length) or modulo~2. 

More precisely, for each walk specification with terminals $s, t$ and values $0
\leq r < q$, if $s = t$ and $r=0$ then the walk can be empty and 
can be ignored, otherwise the walk must be nonempty. 
As the walk is nonempty, going back-and-forth on an edge of the walk allows to
add any even number to the length. So, if $q$ is even, then we can achieve $r$
iff we can achieve $r \bmod{2}$, so we can replace the specification by taking
$r' := r \bmod{2}$ and $q' := 2$. If $q$ is odd, then as $2$ and $q$ are
coprime, going back-on-forth on an edge of the walk allows us to achieve the
remainder $r$. Thus we can replace the specification by taking $r' := 0$ and $q' := 1$.

We now use this observation to give a reduction to the setting of directed graphs,
which allows us to use our results as a black-box to show tractability of the
\ewm and \kscssm{$k$} and \kdsnm{$k$} problems on undirected graphs.

\myparagraph{Simple strategy which fails.}
An obvious strategy to reduce from undirected graphs to
directed graphs is to replace each undirected edge $\{u, v\}$ by the two
directed edges $(u, v)$ and $(v, u)$. However, this does not preserve the
criterion that we want to optimize, for the following intuitive reason. Whenever
an undirected edge is part of the solution, then walks can traverse the edge in
either direction. By contrast, in the directed coding, we could include, say,
$(u,v)$ but not $(v,u)$, so that intuitively the cost of the edge is not the
same depending on whether we intend to traverse it only in one direction or in
both directions in the solution walks.

\myparagraph{Correct reduction.}
We now explain how a slightly more elaborate coding can fix the issue.
Remember that we are considering an \ewm or \kdsnm{$k$} or \kscssm{$k$} instance where
walk lengths are counted 
modulo~$1$ or
modulo~$2$.
Let $G$ be the input undirected 
graph, and let $m$ be
its number of edges. Rewrite $G$ to a directed graph~$G'$ by replacing each undirected edge $\{u,
v\}$
by a linear-sized gadget described in \cref{fig:undir2dir},
namely, a directed path $w_{uv,1}, \ldots, w_{uv,8m+1}$ of length $8m$, an
edge from $u$ to~$w_{uv,1}$, an edge from $w_{uv,8m+1}$ to~$u$, and a path of
length 2 with respective intermediate vertices $w_{uv}'$  and $w_{uv}''$ from
$v$ to~$w_{uv,1}$ and from $w_{uv,8m+1}$ to~$v$.

\begin{figure}[h!]
  \centering
  \begin{tikzpicture}[yscale=.7]
  \node (u) at (0, 0) {$u$};
  \node (v) at (0, -4) {$v$};
  \node (wp) at (-1.5, -3) {$w_{uv}'$};
  \node (wpp) at (1.5, -3) {$w_{uv}''$};
  \node (w1) at (-3, -2) {$w_{uv,1}$};
  \node (w2) at (-1, -2) {$w_{uv,2}$};
  \node (w3) at (.25, -2) {};
  \node (cdots) at (1, -2) {$\cdots$};
  \node (wll) at (1.5, -2) {};
  \node (wl) at (3, -2) {$w_{uv,8m+1}$};
  \draw[->] (u) -- (w1);
  \draw[->] (v) -- (wp);
  \draw[->] (wp) -- (w1);
  \draw[->] (wl) -- (u);
  \draw[->] (wl) -- (wpp);
  \draw[->] (wpp) -- (v);
  \draw[->] (w1) -- (w2);
  \draw[->] (w2) -- (w3);
  \draw[->] (wll) -- (wl);
\end{tikzpicture}
  \caption{Directed gadget to code an undirected edge $\{u, v\}$}
\label{fig:undir2dir}
\end{figure}
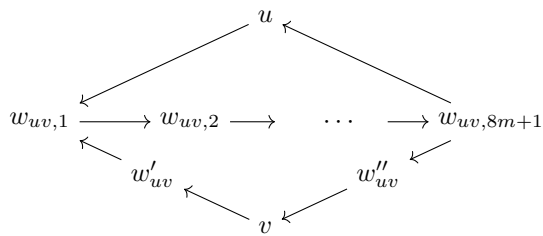

To explain why the reduction is correct, notice that in a minimum-cardinality
subgraph of the graph $G'$, considering the coding $\{u,v\}$ of every edge
of~$G$, it is never useful to keep some edges of the form $(w_{uv,i},
w_{uv,i+1})$ and not others. So in a minimum-cardinality solution we can
distinguish between those edges $\{u,v\}$ of~$G$ where all edges of the path
$w_{uv,1}, \ldots, w_{uv,8m+1}$ have been kept and those where none have been
kept. The cardinality of a subgraph of~$G'$ is then of the form $(8m)m' +
m''$ where $m'$ is the number of edges of~$G$ where the path has been kept, and
where $m''$ are the other edges of~$G'$. Now, observe that in $G'$ there are
precisely $6m$ edges not part of a path, so $m'' \leq 6m$. Hence, a
minimal-cardinality subset of~$G'$ must in fact minimize the number $m'$. 

Furthermore, for every $st$-walk in~$G$ using $m'$ different edges, there
is a directed $st$-walk in~$G'$ of the same parity which uses
$(8m)m'+m''$ different edges, where $m'' \leq 6m$. Indeed, whenever the walk
in $G$ traverses $\{u,v\}$, the walk in~$G'$ traverses either $u$, then the
path, then $v$, or $v$, then the path, then $u$, depending on the direction. The
length of the walk within the gadget coding $\{u,v\}$ is then $1+8m+2$, hence
it is odd and has the same remainder modulo~1 or modulo~2 than the single edge
$\{u,v\}$ traversed by the walk in~$G$.

Conversely, $st$-walks in~$G'$ using $(8m)m'+m''$ edges witness
an undirected walk of length $m'$ in~$G$. 
Note that walks in $G'$ may also include a walk that go from a vertex $u$ to
itself via a gadget, or also go from $v$ to itself within a gadget, however
we can easily see that these cycles all have even length, and so they do
not change the remainder modulo~1 or modulo~2.

Hence, we have reduced $G$ to a directed graph~$G'$ on which the solutions to
our problems \ewm, \kdsnm{$k$} and \kscssm{$k$} for fixed $k\in\mathbb{N}$ are
unchanged.

\end{document}